\documentclass[english,10pt]{article}
\usepackage[T1]{fontenc}
\usepackage[latin9]{inputenc}
\usepackage{geometry}
\usepackage{babel}
\usepackage{verbatim}
\usepackage{float}
\usepackage{bm}
\usepackage{amsmath}
\usepackage{amssymb}
\usepackage{graphicx}
\usepackage{hyperref}
\usepackage{breakurl}

\makeatletter

\floatstyle{ruled}
\newfloat{algorithm}{tbp}{loa}
\providecommand{\algorithmname}{Algorithm}
\floatname{algorithm}{\protect\algorithmname}

\usepackage{babel}
\usepackage{babel}
\usepackage{cite}\usepackage{amsthm}\usepackage{dsfont}\usepackage{array}\usepackage{mathrsfs}\usepackage{comment}\onecolumn

\usepackage{color}\usepackage{babel}

\newcommand{\R}{\mathbb{R}}

\def \Diag {\mathrm{Diag}}
\def \diag {\mathrm{diag}}
\def \objective {\mathit{obj}}
\def \area {\mathit{area}}
\def \domain {\set{D}}
\def \opt {\set{opt}}

\def \saliency {\textup{\saliency}}

\def \path {\mathit{path}}

\def \label {\mathit{label}}

\def \minimize {\textup{minimize} }
\def \maximize {\textup{maximize}}
\def \subjectto {\textup{subject to}}

\usepackage{babel}
\usepackage{arydshln}

\newcommand{\SpectralExpand}{{Spectral-Expanding}}
\newcommand{\SpectralStitch}{{Spectral-Stitching}}

\makeatother

\begin{document}
\theoremstyle{plain} \newtheorem{lem}{\textbf{Lemma}} \newtheorem{prop}{\textbf{Proposition}}\newtheorem{theorem}{\textbf{Theorem}}
\newtheorem{corollary}{\textbf{Corollary}} \newtheorem{assumption}{\textbf{Assumption}}
\newtheorem{example}{\textbf{Example}} \newtheorem{definition}{\textbf{Definition}}
\newtheorem{fact}{\textbf{Fact}} \theoremstyle{definition}

\theoremstyle{remark}\newtheorem{remark}{\textbf{Remark}}

\let\vec=\mathbf \let\mat=\mathbf \let\set=\mathcal \global\long\def\para#1{\noindent{\bf #1}}

\global\long\def\Diag{\mathrm{Diag}}
 \global\long\def\diag{\mathrm{diag}}
 \global\long\def\objective{\mathit{obj}}
 \global\long\def\area{\mathit{area}}
 \global\long\def\domain{\set{D}}
 \global\long\def\opt{\set{opt}}
 \global\long\def\minimize{\textup{minimize}}
 \global\long\def\subjectto{\textup{subject to}}

\global\long\def\minimize{\textup{minimize} }
 \global\long\def\maximize{\textup{maximize}}
 \global\long\def\subjectto{\textup{subject to}}
 \global\long\def\R{\mathbb{R}}
 
\title{Community Recovery in Graphs with Locality}

\author{Yuxin Chen\thanks{Department of Statistics and of Electrical Engineering,
Stanford University, Stanford, CA 94305, USA (email: yxchen@stanford.edu). } \and Govinda Kamath\thanks{Department of Electrical Engineering, Stanford
University, Stanford, CA 94305, USA (email: gkamath@stanford.edu). } \and Changho Suh\thanks{Department of Electrical Engineering, KAIST, Daejeon 305-701, Korea (e-mail:
chsuh@kaist.ac.kr). } \and David Tse\thanks{Department of Electrical Engineering, Stanford
University, Stanford, CA 94305, USA (email: dntse@stanford.edu). }}

\date{February 2016; ~~Revised: June 2016}

\maketitle

\begin{abstract}
Motivated by applications in domains such as social networks and computational biology, we study the problem of community recovery in graphs with locality. In this problem, pairwise noisy measurements of whether two nodes are in the same community or different   communities come mainly or exclusively from nearby nodes rather than uniformly sampled between all node pairs, as in most existing models. We present two algorithms that run nearly linearly in the number of measurements and which achieve the information limits for exact recovery.  
\end{abstract}
\section{Introduction}
\label{sec:intro}



Clustering of data is a central problem that is prevalent across all of science and engineering. One formulation that has received significant attention in recent years is  community recovery  \cite{girvan2002community,fortunato2010community,porter2009communities}, also referred to as correlation clustering \cite{bansal2004correlation} or  graph clustering \cite{jalali2011clustering}. In this formulation, the objective is to cluster individuals into different communities based on pairwise measurements of their relationships, each of which gives some noisy information about whether two  individuals belong to the same community or different communities. While this formulation applies naturally in social networks, it has a broad range of applications in other domains including protein complex detection \cite{chen2006detecting}, image segmentation \cite{shi2000normalized,globerson2015hard}, 
 shape matching \cite{chen2013matching}, etc. See \cite{abbe2015isit} for an introduction of this topic. 

In recent years, there has been a flurry of works on designing  community recovery algorithms  based on idealised generative models of the measurement process. A particular popular model is the {\em Stochastic Block Model} (SBM) \cite{holland1983stochastic,condon2001algorithms}, where the $n$ individuals to be clustered are modeled as nodes on a random graph. In the simplest version of this model with two communities, this random graph is generated such that two nodes  has an edge connecting them with probability $p$ if they are in the same community and probability $q$ if they belong to different communities. 
If $p > q$, then there are statistically more edges within a community than between two communities, which can provide discriminating information for recovering the communities.  
A closely related model is the {\em Censored Block Model} (CBM) \cite{abbe2014decoding}, where one obtains noisy parity measurements on the edges of an  Erd\H{o}s-R\'{e}nyi graph \cite{durrett2007random}. 
Each edge measurement is $0$ with probability $1- \theta$ and $1$ with probability $\theta$ if the two incident vertices are in the same community, and vice versa if they are in different communities.


Both the SBM and the CBM can be unified into one model by viewing the measurement process as a two-step process. First, the edge {\em locations} where there are measurements are determined by randomly and uniformly sampling a complete graph between the nodes. Second, the {\em value} of each edge measurement is obtained as a noisy function of the communities of the two nodes the edge connects. The two models differ only in the noisy functions. Viewed in this light, it is seen that a central assumption underlying both models is that it is equally likely to obtain measurements between {\em any} pair of nodes. 
This is a very unrealistic assumption in many applications: nodes often have {\em locality}  and it is more likely to obtain data on relationships between nearby nodes than far away nodes. For example, in friendship graphs, individuals that live close by are more likely to interact than nodes that are far away.



This paper focuses on the community recovery problem when the measurements are randomly sampled from graphs with locality structure rather than  complete graphs.  Our theory covers a broad range of graphs including rings, lines, 2-D grids, and small-world graphs (Fig.~\ref{fig:graphs}).  Each of these graphs is parametrized by a locality radius $r$ such that nodes within $r$ hops are connected by an edge.  We characterize the information limits for community recovery on these networks, i.e. the minimum number of measurements needed to exactly recover the communities as the number of nodes $n$ scales. We propose two algorithms whose complexities are nearly linear in the number of measurements and which can achieve the information limits of all these networks for a very wide range of the radius $r$. In the special case when the radius $r$ is so large that measurements at all locations are possible, we recover the exact recovery limit identified by \cite{hajek2015achieving} when measurements are randomly sampled from complete graphs.

It is worth emphasizing that various computationally feasible algorithms
\cite{coja2010graph,chaudhuri2012spectral,jalali2011clustering,chen2014weighted,abbe2015community,cai2015robust,mossel2015density}
have been proposed for more general models beyond the SBM and the
CBM, which accommodate multi-community models, the presence of outlier
samples, the case where different edges are sampled at different rates,
and so on. Most of these models, however, fall short of accounting
for any sort of locality constraints. In fact, the results developed
in prior literature often lead to unsatisfactory guarantees when applied
to graphs with locality, as will be detailed in Section \ref{sec:main-results}. 
Another recent work \cite{chen2015information} has determined the order of the information limits in geometric graphs, 
with no tractable algorithms provided therein. 
In contrast, our findings uncover a curious phenomenon: the presence
of locality does not lead to additional computational barriers:
solutions that are information theoretically optimal can often be
achieved computational efficiently and, perhaps more surprisingly, within
nearly linear time.

The paper is structured as follows. We describe the problem formulation in Section \ref{sec:model}, including a concrete application from computational 
biology---called haplotype phasing---which motivates much of our theory. Section \ref{sec:main-results} presents our main results, with extensions of the 
basic theory and numerical results provided in Sections \ref{sec:Extension} and \ref{sec:numerical} respectively. 
Section \ref{sec:discussion} concludes the paper with a few potential extensions. 
The proofs of all results are deferred to the appendices.

\section{Problem Formulation and A Motivating Application}\label{sec:model}

This section is devoted to describing a basic mathematical setup of our problem, and to discussing a concrete application that comes from
computational biology. 

\subsection{Sampling Model}
\label{sec:Problem-Setup}

\textbf{Measurement Graph}. Consider a collection of $n$ vertices
$\mathcal{V}=\{1,\cdots,n\}$, each represented by a binary-valued
vertex variable $X_{i}\in \{ 0, 1 \} $, $1\leq i\leq n$.
Suppose it is only feasible to take pairwise samples over a restricted
set of locations, as represented by a graph ${\cal G}=\left({\cal V},{\cal E}\right)$
that comprises an edge set ${\cal E}$. Specifically, for each edge
$(i,j)\in\mathcal{E}$ one acquires $N_{i,j}$
samples\footnote{Here and throughout, we adopt the convention that $N_{i,j}\equiv0$
for any $(i,j)\notin\mathcal{E}$. } $Y_{i,j}^{(l)}$ ($1\leq l\leq N_{i,j}$), where each sample measures
the parity of $X_{i}$ and $X_{j}$. We will use $\mathcal{G}$ to
encode the locality constraint of the sampling scheme, and shall
pay particular attention to the following families of measurement
graphs. 

\begin{itemize}
\item 
\emph{Complete graph}: $\mathcal{G}$ is called a complete graph
if every pair of vertices is connected by an edge; see Fig.~\ref{fig:graphs}(a).

\item 
\emph{Line}: $\mathcal{G}$ is said
to be a line $\mathcal{L}_{r}$ if, for some locality radius $r$,
$\left(i,j\right)\in\mathcal{E}$ iff $|i-j|\leq r$; see Fig.~\ref{fig:graphs}(b).
 
\item 
\emph{Ring}: $\mathcal{G}=\left({\cal V},{\cal E}\right)$ is said
to be a ring $\mathcal{R}_{r}$ if, for some locality radius $r$,
$\left(i,j\right)\in\mathcal{E}$ iff $i-j\in[-r,r]$ ($\mathsf{mod}~n$); see Fig.~\ref{fig:graphs}(c).

\item 
\emph{Grid}: $\mathcal{G}$ is called a grid if (1) all vertices reside
within a $\sqrt{n}\times\sqrt{n}$ square with integer coordinates,
and (2) two vertices are connected by an edge if they are
at distance not exceeding some radius $r$; see Fig.~\ref{fig:graphs}(d).

\item 
\emph{Small-world graphs}: $\mathcal{G}$ is said to be a small-world
graph if it is a superposition of a complete graph $\mathcal{G}_{0}=\left(\mathcal{V},\mathcal{E}_{0}\right)$
and another graph $\mathcal{G}_{1}=\left(\mathcal{V},\mathcal{E}_{1}\right)$
with locality. See Fig.~\ref{fig:graphs}(e) for an example. 
\end{itemize}
\begin{figure*}
\centering%
\begin{tabular}{ccc}
\quad \includegraphics[height=2.8cm]{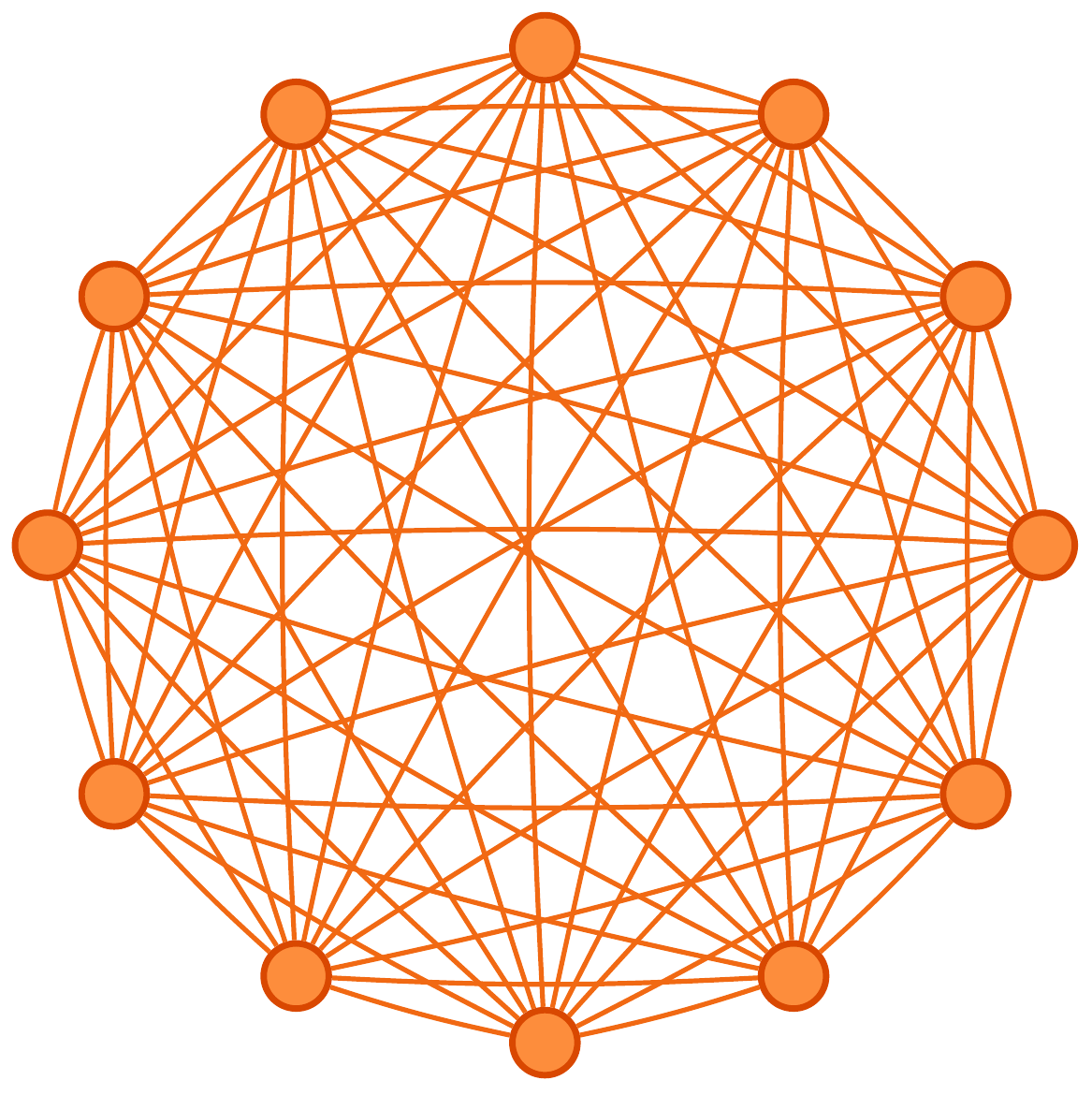} \quad
& \includegraphics[height=2.7cm]{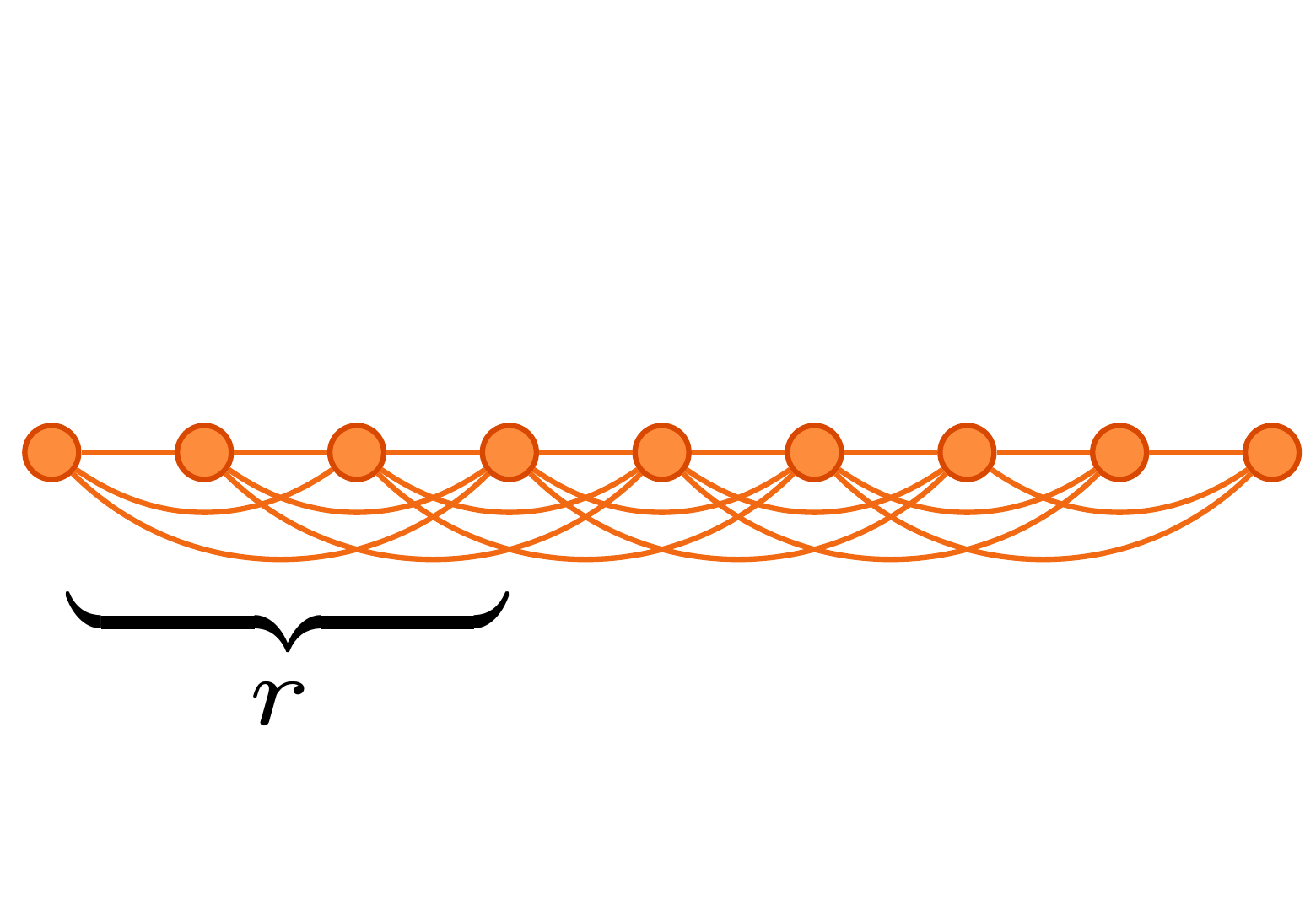}
& \includegraphics[height=2.7cm]{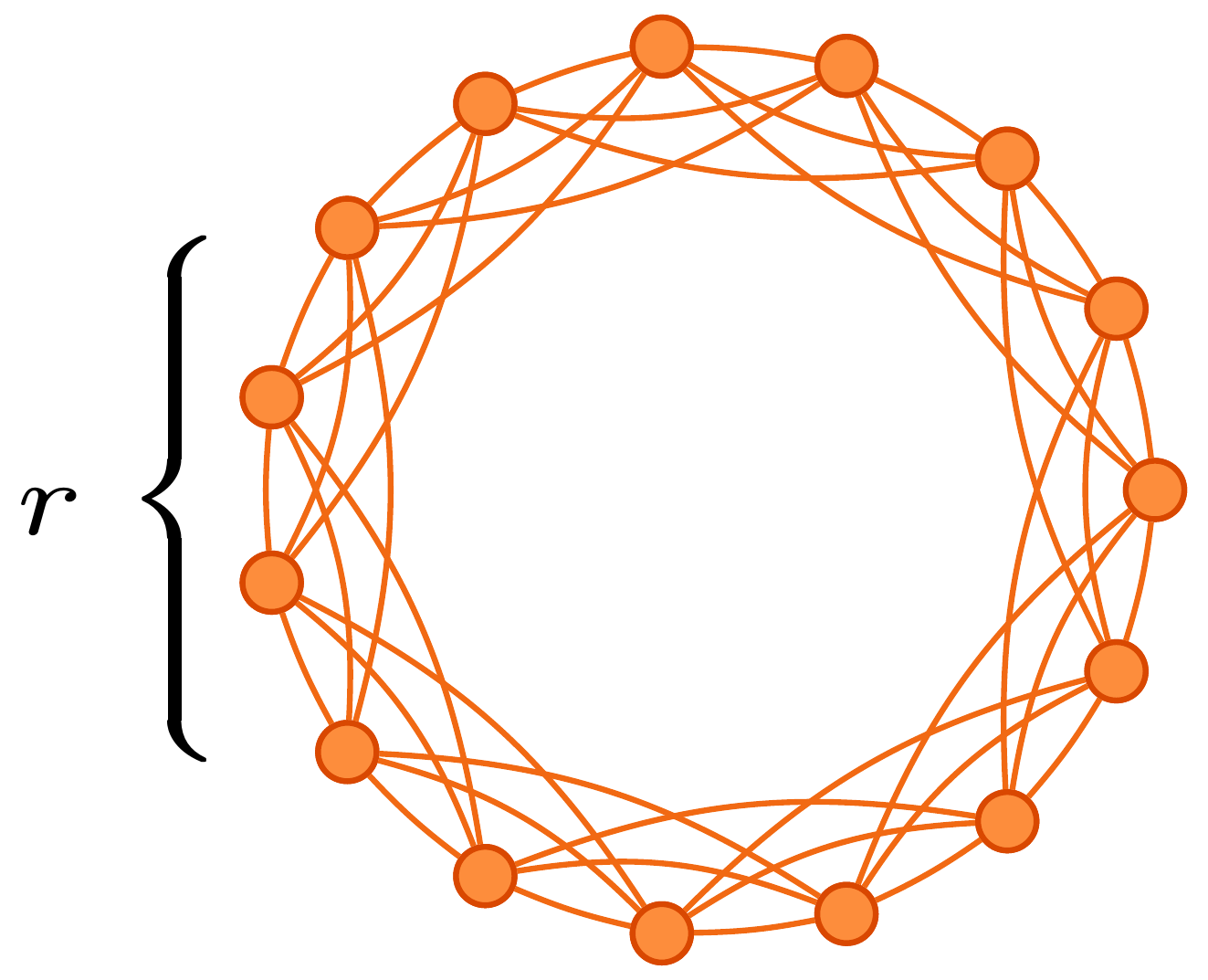}  
\tabularnewline
(a) a complete graph  
& (b) a line $\mathcal{L}_{r}$ 
& (c) a ring $\mathcal{R}_{r}$
\tabularnewline
\end{tabular}
\begin{tabular}{cc}
 \quad \includegraphics[height=2.7cm]{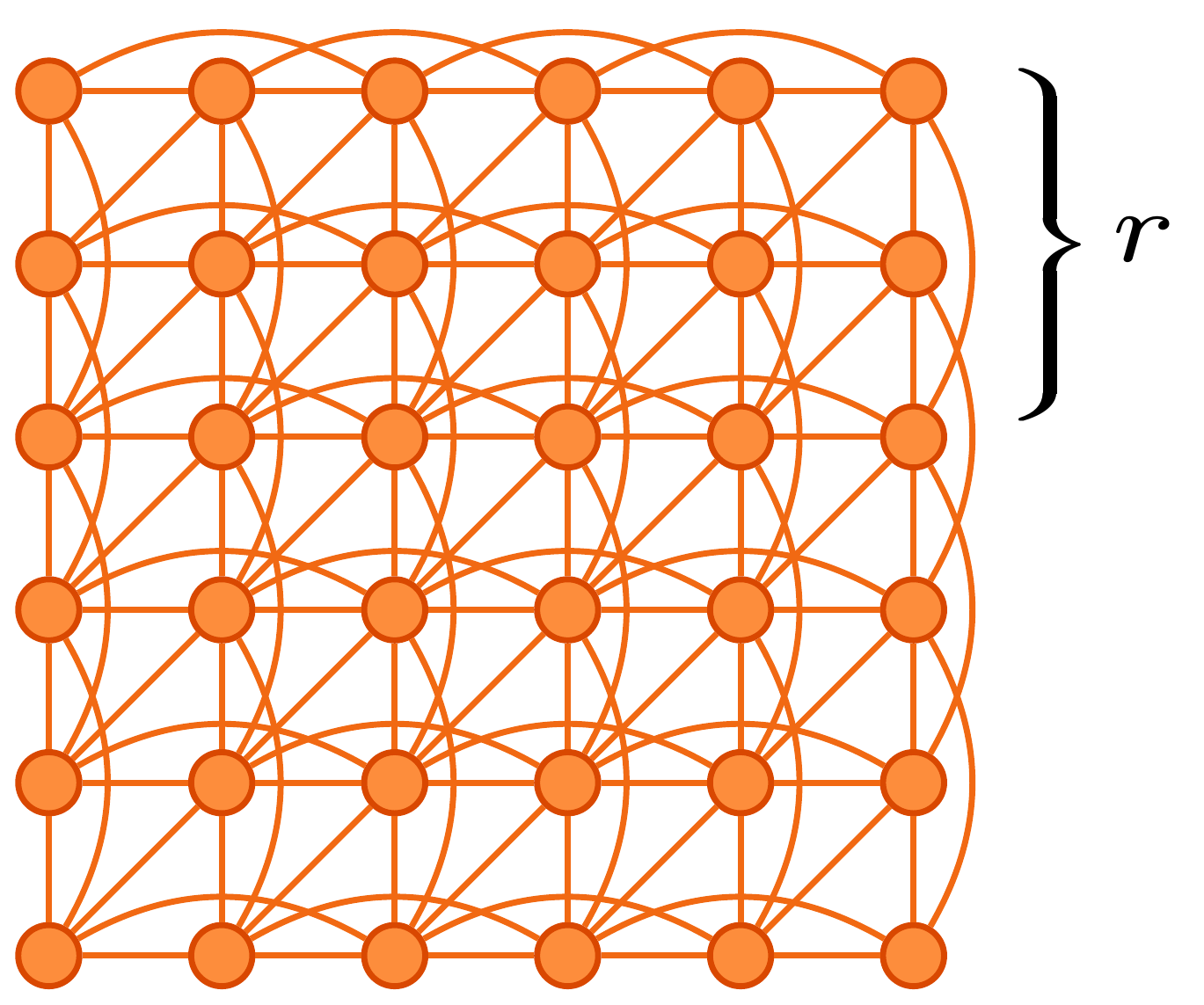}
& \includegraphics[height=2.7cm]{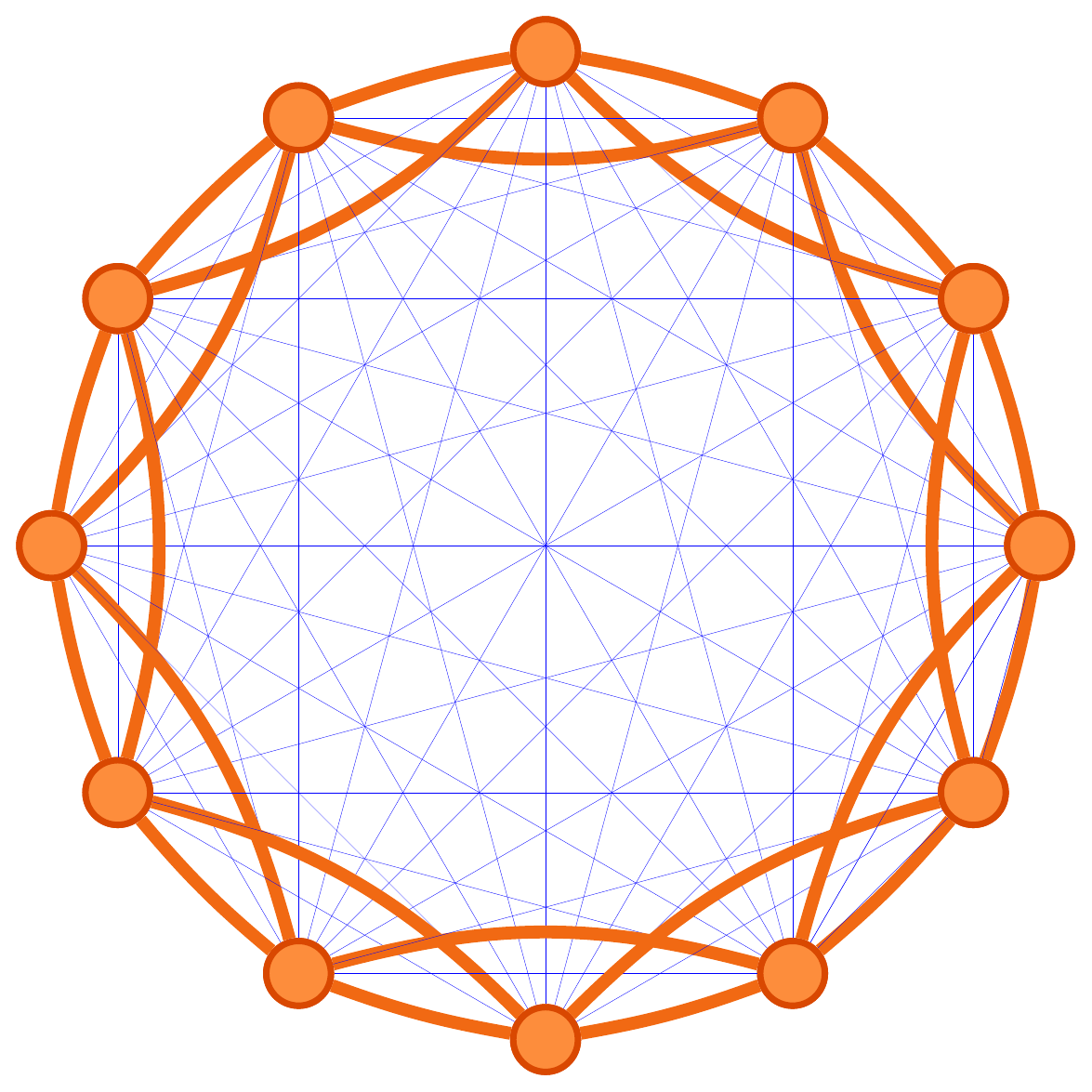}
\tabularnewline
 (d) a grid   
& (e) a small-world graph\tabularnewline
\end{tabular}

\caption{Examples of a complete graph, a line, a ring,  a 2-D grid, and a small-world graph. }
\label{fig:graphs}
\vspace{-0.1in}
\end{figure*}

\vspace{1em}

\noindent\textbf{Random Sampling}. This paper focuses on a random sampling
model, where the number of samples $N_{i,j}$ taken over $(i,j)\in\mathcal{E}$
is independently drawn and obeys\footnote{All results presented in this paper hold under a related model where
$N_{i,j}\sim \mathsf{Bernoulli}\left(\lambda\right)$, as long as $|\mathcal{E}|\gg n\log n$
and $\lambda\leq1$ (which is the regime accommodated in all theorems).
In short, this arises due to the tightness of Poisson approximation
to the Binomial distribution. 
We omit the details for conciseness.
}
$N_{i,j}\sim\mathsf{Poisson}\left(\lambda\right)$
for some average sampling rate $\lambda$. This gives rise to an average
total sample size
\begin{equation}
m:=\sum\nolimits _{(i,j)\in\mathcal{E}}\mathbb{E}\left[N_{i,j}\right]=\lambda\left|\mathcal{E}\right|.\label{eq:sample-size}
\end{equation}
When $m$ is large, the actual sample size sharply concentrates around
$m$ with high probability.

\vspace{1em}

\noindent\textbf{Measurement Noise Model}.  The acquired parity measurements 
are assumed to be independent given $ N_{i,j} $; more
precisely, conditional on $N_{i,j} $, 
\begin{equation}
Y_{i,j}^{(l)}=Y_{j,i}^{(l)}\overset{\text{ind.}}{=}\begin{cases}
X_{i}\oplus X_{j},\quad & \text{with probability }1-\theta\\
X_{i}\oplus X_{j}\oplus1, & \text{else}
\end{cases}\label{eq:BSC}
\end{equation}
for some fixed error rate $0<\theta<1$, where $\oplus$ denotes modulo-2
addition.  This is the same as the noise model  in CBM \cite{abbe2014decoding}. The SBM corresponds to an asymmetric erasure model for the measurement noise, and we expect our results extend to that model as well.

\subsection{Goal: Optimal Algorithm for Exact Recovery}

This paper centers on exact recovery, that is, to reconstruct all
input variables $\bm{X}=[X_{i}]_{1\leq i\leq n}$ precisely up to
global offset. This is all one can hope for since there is absolutely
no basis to distinguish $\bm{X}$ from $\bm{X}\oplus\bm{1}:=\left[X_{i}\oplus1\right]_{1\leq i\leq n}$
given only parity samples. More precisely, for any recovery procedure
$\psi$ the probability of error is defined as 
\begin{align*}
P_{\text{e}}\left(\psi\right):=\max_{\bm{X}\in\left\{ 0,1\right\} ^{n}}\mathbb{P}\left\{ \psi(\bm{Y})\neq\bm{X}\text{ and }\psi(\bm{Y})\neq\bm{X}\oplus\bm{1}\right\} ,\label{eq:p-error}
\end{align*}
where $\bm{Y}:=\{Y_{i,j}^{(l)}\}$. The goal is to develop an algorithm
whose required sample complexity approaches the information limit
$m^{*}$ (as a function of $(n,\theta)$), that is, the minimum sample size $m$
under which $\inf_{\psi}P_{\text{e}}\left(\psi\right)$ vanishes as $n$ scales.
For notational simplicity, the dependency of $m^{*}$ on $\left(n,\theta\right)$
shall often be suppressed when it is clear from the context.

\subsection{Haplotype Phasing: A Motivating Application}

Before proceeding to present the algorithms, we describe here a genome phasing application that motivates this research and show how it can be modeled as a community recovery problem on graphs with locality.

Humans have $23$ pairs of homologous chromosomes, one maternal and one paternal. Each pair are identical sequences of nucleotides A,G,C,T's except on certain documented positions called single nucleotide polymorphisms (SNPs), or genetic variants. 
At each of these positions, one of the chromosomes takes on one of A,G,C or T which is the same as the majority of the population (called the {\it major allele}), while the other chromosome takes on a variant (also called {\it minor allele}).
The problem of haplotype phasing is that of determining which variants are on the same chromosome in each pair,  and has important applications such as in personalized medicine and understanding poylogenetic trees. The advent of next generation sequencing technologies allows haplotype phasing by providing linking reads between multiple SNP locations \cite{browning2011haplotype,donmez2011hapsembler,cai2016structured}.

One can formulate the problem of haplotype phasing as recovery of two communities of SNP locations, those with the variant on the maternal chromosome and those with the variant on the paternal chromosome \cite{si2014haplotype,kamath2015optimal}.   Each pair of linking reads gives a noisy measurement of whether two SNPs have the variant on the same chromosome or different chromosomes. While there are of the order of $n = 10^5$ SNPs on each chromosome, the linking reads are typically only several SNPs or at most $100$ SNPs apart, depending on the specific sequencing technology. Thus, the measurements are sampled from a line graph like in Fig.~\ref{fig:graphs}(b) with locality radius $r \ll n$.

\subsection{Other Useful Metrics and Notation}

It is convenient to introduce some notations that will
be used throughout. 
One key metric that captures the distinguishability between two probability
measures $P_{0}$ and $P_{1}$ is the \emph{Chernoff information}
\cite{cover2006elements}, defined as
{\small 
\begin{equation}
D^{*}\left(P_{0},P_{1}\right):=-\inf_{0\leq\tau\leq1}\log\left\{ \sum\nolimits_y P_{0}^{\tau}\left(y\right)P_{1}^{1-\tau}(y)\right\} .
\label{eq:Chernorff-info}
\end{equation}}
For instance, when $P_{0}\sim\mathsf{Bernoulli}\left(\theta\right)$
and $P_{1}\sim\mathsf{Bernoulli}\left(1-\theta\right)$, $D^{*}$
simplifies to 
\begin{equation}
D^{*}=\mathsf{KL}\left(0.5\hspace{0.2em}\|\hspace{0.2em}\theta\right)=0.5\log\frac{0.5}{\theta}+0.5\log\frac{0.5}{1-\theta},\label{eq:Chernoff-BSC}
\end{equation}
where $\mathsf{KL}\left(0.5\hspace{0.2em}\|\hspace{0.2em}\theta\right)$
is the Kullback-Leibler (KL) divergence between $\mathsf{Bernoulli}(0.5)$
and $\mathsf{Bernoulli}(\theta)$. Here and below, we shall use $\log\left(\cdot\right)$
to indicate the natural logarithm. 

In addition, we denote by $d_{v}$ and $d_{\mathrm{avg}}$ the vertex degree of $v$ and
 the average vertex degree of $\mathcal{G}$, respectively.  We use $\|\bm{M}\|$ to represent the spectral norm of a matrix $\bm{M}$. 
Let ${\bf 1}$ and ${\bf 0}$
be the all-one and all-zero vectors, respectively. We denote by $\mathrm{supp}\left(\boldsymbol{x}\right)$
(resp.~$\left\Vert \boldsymbol{x}\right\Vert _{0}$) the support
(resp.~the support size) of $\boldsymbol{x}$. The standard notion
$f(n)=o\left(g(n)\right)$ means $\underset{n\rightarrow\infty}{\lim}f(n)/g(n)=0$;
$f(n)=\omega\left(g(n)\right)$ means $\underset{n\rightarrow\infty}{\lim}g(n)/f(n)=0$;
$f(n)=\Omega\left(g(n)\right)$ or $f(n)\gtrsim g(n)$ mean there
exists a constant $c$ such that $f(n)\geq cg(n)$; $f(n)=O\left(g(n)\right)$
or $f(n)\lesssim g(n)$ mean there exists a constant $c$ such that
$f(n)\leq cg(n)$; $f(n)=\Theta\left(g(n)\right)$ or $f(n)\asymp g(n)$
mean there exist constants $c_{1}$ and $c_{2}$ such that $c_{1}g(n)\leq f(n)\leq c_{2}g(n)$.

\section{Main Results}
\label{sec:main-results}

This section describes two nearly linear-time algorithms and presents our main results. 
The proofs of all theorems are deferred to the appendices.

\subsection{Algorithms\label{sec:Algorithm}}

\begin{algorithm*}[t]
\begin{enumerate}
\item \textbf{Run spectral method (Algorithm \ref{alg:Algorithm-spectral})
on a core subgraph} induced by $\mathcal{V}_{\mathrm{c}}$, which
yields estimates $X_{j}^{(0)},1\leq j\leq|\mathcal{V}_{\mathrm{c}}|$.
\item \textbf{Progressive estimation}: for $i=|\mathcal{V}_{\mathrm{c}}|+1,\cdots,n$,
\vspace{-0.5em}
\[
X_{i}^{(0)}\leftarrow\mathsf{majority}\left\{ Y_{i,j}^{(l)}\oplus X_{j}^{(0)}\mid j:\text{ }j<i,\text{ }(i,j)\in\mathcal{E},\text{ }1\leq l\leq N_{i,j}\right\} .
\]

\item \textbf{Successive local refinement:} for $t=0,\cdots,T-1$, 
\vspace{-1em}
\begin{eqnarray*}
X_{i}^{(t+1)} & \leftarrow & \mathsf{majority}\left\{ Y_{i,j}^{(l)}\oplus X_{j}^{(t)}\mid j:\text{ }j\neq i,\text{ }(i,j)\in\mathcal{E},\text{ }1\leq l\leq N_{i,j}\right\} ,\quad1\leq i\leq n.
\end{eqnarray*}

\item \textbf{Output} $X_{i}^{(T)}$, $1\leq i\leq n$. 
\end{enumerate}
Here, $\mathsf{majority}\left\{ \cdot\right\} $ represents the majority
voting rule: for any sequence $s_{1},\cdots,s_{k}\in\left\{ 0,1\right\} $,
$\mathsf{majority}\left\{ s_{1},\cdots,s_{k}\right\} $ is equal to
1 if $\sum_{i=1}^{k}s_{i}>k/2$; and 0 otherwise. 

\caption{\label{alg:Algorithm-progressive}\textbf{: \SpectralExpand}}
\end{algorithm*}

\begin{algorithm*}[t]
\begin{enumerate}
\item \textbf{Input: }measurement graph $\mathcal{G}=\left(\mathcal{V},\mathcal{E}\right)$,
and samples $\left\{ Y_{i,j}^{(l)}\in\{0,1\}\mid j:\text{ }j<i,\text{ }(i,j)\in\mathcal{E},\text{ }1\leq l\leq N_{i,j}\right\} $. 
\item Form a sample matrix $\bm{A}$ such that\textbf{ 
\[
\bm{A}_{i,j}=\begin{cases}
\bm{1}\big\{ Y_{i,j}^{(1)}=0\big\}-\bm{1}\big\{ Y_{i,j}^{(1)}=1\big\},\quad & \text{if }(i,j)\in\mathcal{E};\\
0, & \text{else}.
\end{cases}
\]
}
\item Compute the leading eigenvector $\bm{u}$ of \textbf{$\bm{A}$}, and
for all $1\leq i\leq n$ set
\[
X_{i}^{(0)}=\begin{cases}
1,\quad & \text{if }\bm{u}_{i}\geq0,\\
0, & \text{else}.
\end{cases}
\]

\item \textbf{Output} $X_{i}^{(0)}$, $1\leq i\leq n$. 
\end{enumerate}
\vspace{-0.7em}

\caption{\label{alg:Algorithm-spectral}\textbf{: Spectral initialization}}
\end{algorithm*}

\subsubsection{\SpectralExpand}

The first algorithm, called {\SpectralExpand}, consists of
three stages. For concreteness, we start by describing the procedure
when the measurement graphs are lines / rings; see Algorithm \ref{alg:Algorithm-progressive}
for a precise description of the algorithm and Fig.~\ref{fig:3stages} for a graphical illustration.

\begin{figure*}

\includegraphics[width=0.75\textwidth]{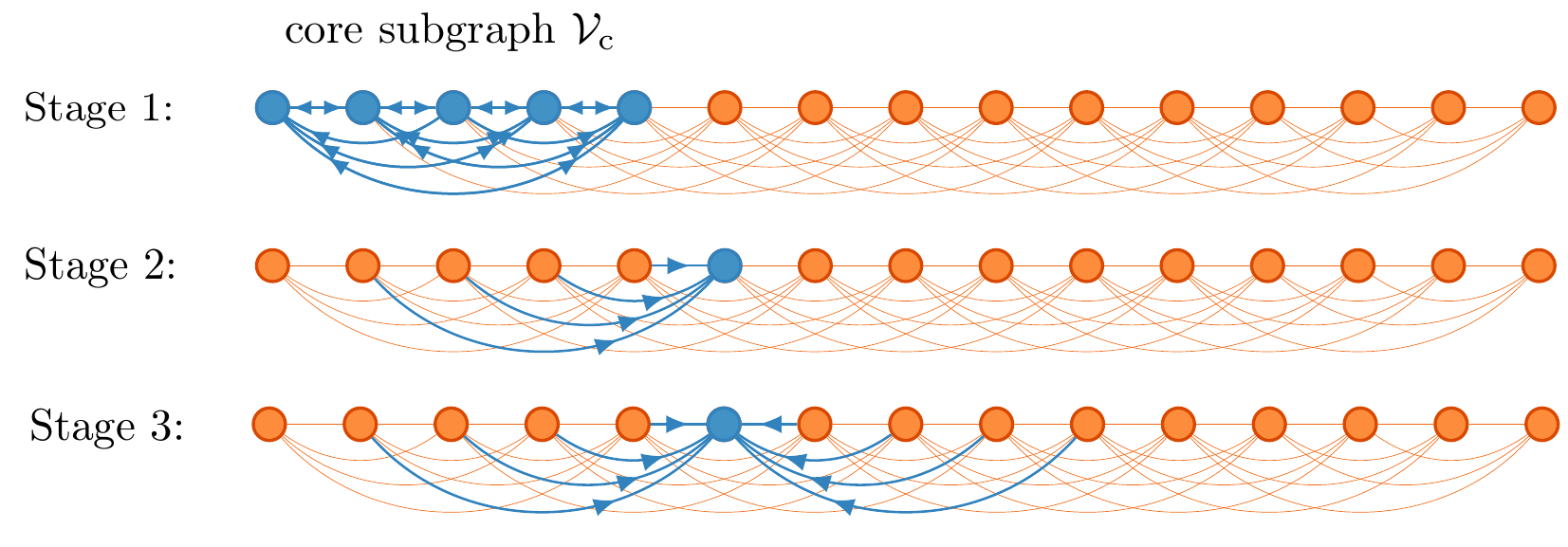} 
\centering
\caption{Illustration of the information flow in {\SpectralExpand}: (a) Stage 1 concerns
recovery in a core complete subgraph;
(b) Stage 2 makes a forward pass by progressively propagating information
through backward samples; (c) Stage 3 refines each $X_{v}$ by employing
all samples incident to $v$. 
}
\label{fig:3stages}
\end{figure*}

\begin{itemize}
\item  \textbf{Stage 1: spectral metrod on a core subgraph}. Consider a subgraph
$\mathcal{G}_{\mathrm{c}}$ induced by $\mathcal{V}_{\mathrm{c}}:=\left\{ 1,\cdots,r\right\} $,
and it is self-evident that $\mathcal{G}_{\mathrm{c}}$ is a complete
subgraph. We run a spectral method (e.g.~\cite{chin2015stochastic}) on $\mathcal{G}_{\mathrm{c}}$
using samples taken over $\mathcal{G}_{\mathrm{c}}$, in the hope
of obtaining approximate recovery of $\left\{ X_{i}\mid i\in\mathcal{V}_{\mathrm{c}}\right\} $.
Note that the spectral method can be replaced by other efficient algorithms,
including semidefinite programming (SDP) \cite{javanmard2015phase} and a variant of belief propagation (BP) \cite{mossel2013belief}. 

\item \textbf{Stage 2: progressive estimation of remaining vertices.} For
each vertex $i>|\mathcal{V}_{\mathrm{c}}|$, compute an estimate of $X_{i}$
by majority vote using \emph{backward samples}---those samples linking
$i$ and some $j<i$. The objective is to ensure that a large fraction
of estimates obtained in this stage are accurate. As will be discussed
later, the sample complexity required for approximate recovery is
much lower than that required for exact recovery, and hence the task
is feasible even though we do not use any forward samples to estimate
$X_{i}$. 

\item \textbf{Stage 3: successive local refinement.} Finally, we clean up
all estimates using both backward and forward samples in order to
maximize recovery accuracy. This is achieved by running local majority
voting from the neighbors of each vertex until convergence. 
In contrast to many prior work, no sample splitting is required, namely,  
we
reuse all samples in all iterations in all stages. 
As we shall see, this stage is the bottleneck for exact information recovery. 

\end{itemize}

\begin{remark} The proposed algorithm falls under the category of a general non-convex paradigm, which starts with an approximate estimate (often via spectral methods) followed by iterative refinement.  This paradigm has been successfully applied to a wide spectrum of applications ranging from matrix completion \cite{keshavan2010matrix,jain2013low} to phase retrieval \cite{chen2015solving} to community recovery~\cite{chaudhuri2012spectral,abbe2014exact,gao2015achieving}. 
\end{remark}

An important feature of this algorithm is its low computational
complexity. First of all, the spectral method can be performed within $O\left( m_{\mathrm{c}} \log n\right)$ time by
means of the power method, where $m_{\mathrm{c}}$ indicates the number of samples falling on $\mathcal{G}_{\mathrm{c}}$. Stage 2 entails
one round of majority voting, whereas the final stage---as we will
demonstrate---converges within at most $O\left(\log n\right)$ rounds
of majority voting. Note that each round of majority voting can be
completed in linear time, i.e.~in time proportional to reading all
samples. Taken collectively, we see that {\SpectralExpand}
can be accomplished within $O\left(m\log n\right)$ flops, which is
nearly linear time. 

Careful readers will recognize that Stages 2-3 bear similarities
with BP, and might wonder whether Stage 1 can also be replaced with
standard BP. Unfortunately, we are not aware of any approach to analyze
the performance of vanilla BP without a decent initial guess. Note,
however, that the spectral method is already nearly linear-time, and is hence at least as fast as any feasible procedure.

While the preceding paradigm is presented for lines / rings, it easily
extends to a much broader family of graphs with locality.
The only places that need to be adjusted are: 
%

\begin{enumerate}
\item 
 \textbf{The core subgraph $\mathcal{V}_{\mathrm{c}}$}. One would
like to ensure that $|\mathcal{V}_{\mathrm{c}}|\gtrsim d_{\mathrm{avg}}$
and that the subgraph $\mathcal{G}_{\mathrm{c}}$ induced by $\mathcal{V}_{\mathrm{c}}$
forms a (nearly) complete subgraph, in order to guarantee decent recovery
in Stage 1.

\item 
 \textbf{The ordering of the vertices}. Let $\mathcal{V}_{\mathrm{c}}$
form the first $|\mathcal{V}_{\mathrm{c}}|$ vertices of $\mathcal{V}$,
and make sure that each $i>|\mathcal{V}_{\mathrm{c}}|$ is connected
to at least an order of $d_{\mathrm{avg}}$ vertices in $\left\{ 1,\cdots,i-1\right\} $.
This is important because each vertex needs to be incident to sufficiently
many backward samples in order for Stage 2 to be
successful. 
\end{enumerate}

\subsubsection{\SpectralStitch}

\begin{algorithm*}[t]
\begin{enumerate}
\item \textbf{Split} all vertices into several (non-disjoint) vertex subsets
each of size $W$ as follows\textbf{ 
\[
\mathcal{V}_{l}:=\left\{ i\text{ }\left|\text{ }\left(i-1\right)W/2+1\leq l\leq\left(i-1\right)W/2+W\right.\right\} ,\quad l=1,2,\cdots,
\]
}and \textbf{run spectral method (Algorithm \ref{alg:Algorithm-spectral})
on each subgraph induced by} $\mathcal{V}_{l}$, which yields
estimates $\{ X_{j}^{\mathcal{V}_{l}}\mid j\in\mathcal{V}_{l}\} $
for each $l\geq1$.
\item \textbf{Stitching}: set $X_{j}^{(0)}\leftarrow X_{j}^{\mathcal{V}_{1}}$
for all $j\in\mathcal{V}_{1}$; ~for $l=2,3,\cdots$, 
\begin{align*}
X_{j}^{(0)}\leftarrow X_{j}^{\mathcal{V}_{l}}\text{ }\left(\forall j\in\mathcal{V}_{l}\right) & \quad\text{if }\sum\nolimits_{j\in\mathcal{V}_{l}\cap\mathcal{V}_{l-1}}X_{j}^{\mathcal{V}_{l}}\oplus X_{j}^{\mathcal{V}_{l-1}}\leq 0.5\left|\mathcal{V}_{l}\cap\mathcal{V}_{l-1}\right|;\\
\text{and}\quad X_{j}^{(0)}\leftarrow X_{j}^{\mathcal{V}_{l}}\oplus1\text{ }\left(\forall j\in\mathcal{V}_{l}\right) & \quad\text{otherwise}.
\end{align*}

\item \textbf{Successive local refinement} and output $X_{i}^{(T)}$, $1\leq i\leq n$ (see Steps 3-4 of Algorithm \ref{alg:Algorithm-progressive}).  
\vspace{-0.5em}
\end{enumerate}
\caption{\label{alg:Algorithm-stitch}\textbf{: \SpectralStitch} }
\end{algorithm*}

We now turn to the $2^{\text{nd}}$ algorithm called {\SpectralStitch}, which shares similar spirit as  {\SpectralExpand} and, in fact,  
differs from {\SpectralExpand} only in Stages 1-2.  

\begin{itemize}
\item
\textbf{Stage 1: node splitting and spectral estimation}.
Split $\mathcal{V}$ into several overlapping subsets $\mathcal{V}_l~(l\geq 1)$ of size $W$, such that any two adjacent subsets share $W/2$ common vertices. 
We choose the size $W$ of each $\mathcal{V}_l$ to be $r$ for rings / lines, and on the order of $d_{\mathrm{avg}}$ for other graphs. We then run spectral methods separately on each subgraph $\mathcal{G}_l$ induced by $\mathcal{V}_l$, in the hope of achieving approximate estimates $\{ X_{i}^{\mathcal{V}_l}\mid i\in\mathcal{V}_{l}\} $---up to global phase---for each subgraph. 

\item
\textbf{Stage 2: stiching the estimates}. 
The aim of this stage is to stitch together the outputs of Stage 1 computed in isolation for the collection of overlapping subgraphs. If approximate recovery (up to some global phase) has been achieved in Stage 1 for each $\mathcal{V}_l$, 
then the outputs for any two adjacent subsets are positively correlated only when they have matching global phases. This simple observation allows us to calibrate the global phases for all preceding estimates, thus yielding a vector $\{X_i^{(0)}\}_{1\leq i\leq n}$ that is approximately faithful to the truth modulo some global phase.  
\end{itemize}

The remaining steps of {\SpectralStitch} follow the same local refinement procedure as in {\SpectralExpand}, and we can employ the same ordering of vertices as in {\SpectralExpand}. 
See Algorithm \ref{alg:Algorithm-stitch} and Fig.~\ref{fig:Stitch}. As can be seen, the first 2 stages of {\SpectralStitch}---which can also be completed in nearly linear time---are more ``symmetric'' than those of {\SpectralExpand}. More precisely, {\SpectralExpand} emphasizes a single core subgraph $\mathcal{G}_{\mathrm{c}}$ and computes all other estimates based on $\mathcal{G}_{\mathrm{c}}$, while {\SpectralStitch} treats each subgraph $\mathcal{G}_l$ almost equivalently. This symmetry nature might be practically beneficial when the acquired data deviate from our assumed random sampling model.

\begin{figure*}
\centering
\includegraphics[width=0.7\textwidth]{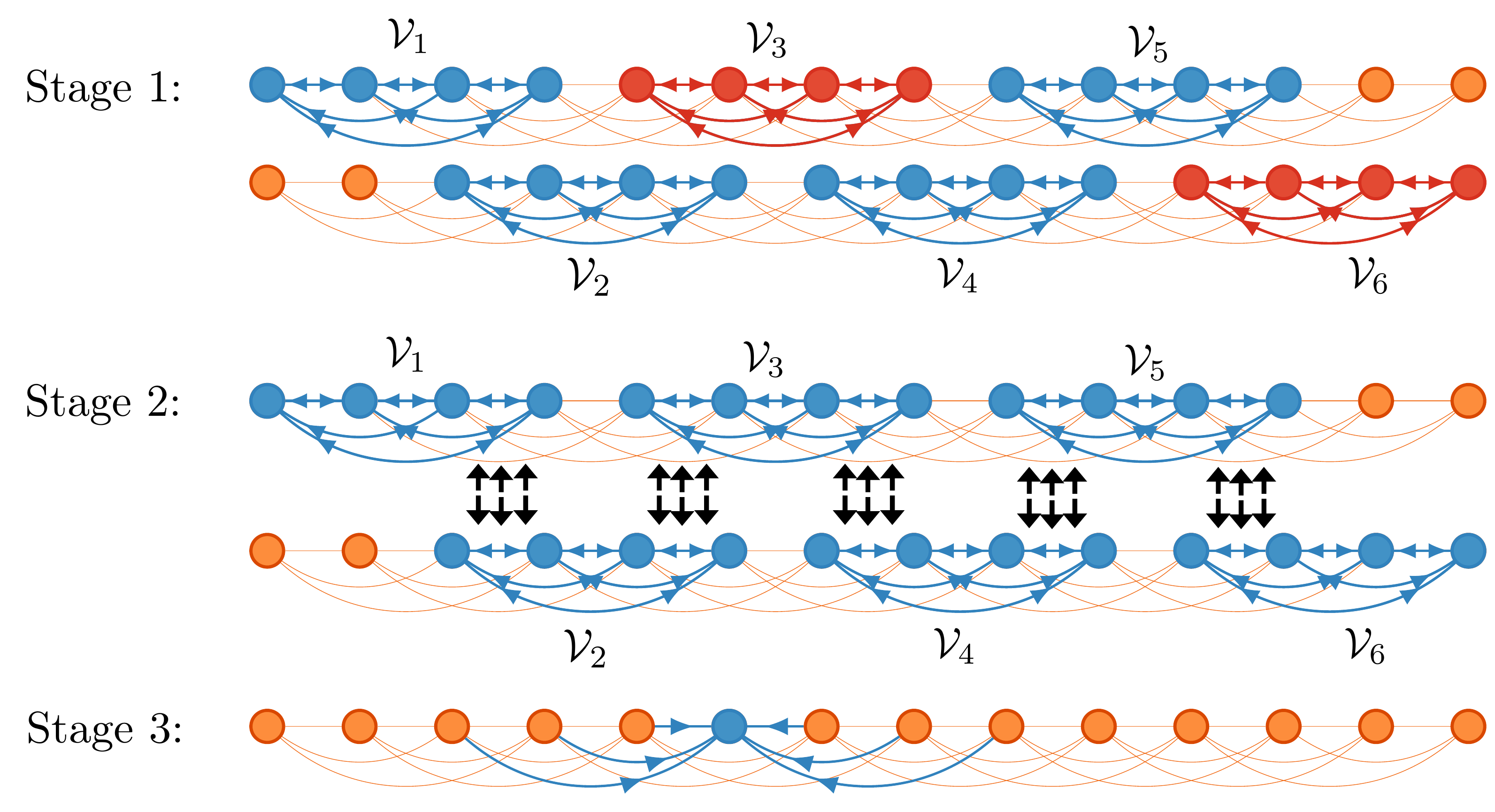}
\caption{Illustration of the information flow in {\SpectralStitch}:  
(a) Stage 1 runs spectral methods for a collection of overlapping subgraphs of size $W$ separately; (b) Stage 2 stitches all groups of estimates
 together using the information coming from their overlaps; (c)  Stage 3 cleans up all estimates by employing
all samples incident to each node. 
}
\label{fig:Stitch}
\end{figure*}

\subsection{Theoretical Guarantees: Rings\label{sec:performance-rings}}

We start with the performance of our algorithms
for rings. This class of graphs---which is spatially
invariant---is arguably the simplest model exhibiting locality structure.

\subsubsection{Minimum Sample Complexity\label{sec:theory-rings}}

Encouragingly, the proposed algorithms succeed in achieving
the minimum sample complexity, as stated below.


\begin{theorem}\label{theorem:rings}Fix $\theta>0$ and any small
$\epsilon>0$. Let $\mathcal{G}$ be a ring $\mathcal{R}_{r}$ with
locality radius $r$, and suppose
\begin{equation}
m\geq\left(1+\epsilon\right)m^{*},\label{eq:achievability-ring}
\end{equation}
where 
\begin{equation}
m^{*}=\frac{n\log n}{2\left(1-e^{-\mathsf{KL}\left(0.5\|\theta\right)}\right)}.\label{eq:optimal-sample-complexity-rings}
\end{equation}
Then with probability approaching one\footnote{More precisely, the proposed algorithms succeed with probability exceeding
$1-c_{1}r^{-9}-C_{2}\exp\{ -c_{2}\frac{m}{n}(1-e^{-D^{*}})\} $
for some constants $c_{1},c_{2},C_{2}>0$. }, {\SpectralExpand}  (resp.~{\SpectralStitch}) converges to the
ground truth within $T=O\left(\log n\right)$ iterations, provided
that $r\gtrsim\log^{3}n$ (resp.~$r\geq n^{\delta}$ for an arbitrary
constant $\delta>0$). 

Conversely, if $m<\left(1-\epsilon\right)m^{*}$, then the probability
of error $P_{\mathrm{e}}(\psi)$ is approaching one for any algorithm
$\psi$. \end{theorem}

\begin{remark}When $r=n-1$, a ring  reduces to a complete graph (or an equivalent Erd\H{o}s-R\'enyi model). For this case, computationally feasible algorithms have been extensively studied \cite{swamy2004correlation,jalali2011clustering,chen2014improved,chen2014weighted,chen2013matching},  most of which focus only on the scaling results. Recent work \cite{hajek2015achieving,jog2015information} succeeded in characterizing the sharp threshold for this case, and it is immediate to check that the sample complexity we derive in (\ref{eq:optimal-sample-complexity-rings})
matches the one presented in \cite{hajek2015achieving,jog2015information}.  \end{remark}

\begin{remark}Theorem \ref{theorem:rings} requires $r\gtrsim\mathsf{poly}\log(n)$
because each node needs to be connected to sufficiently many neighbors
in order to preclude ``bursty'' errors. 
The condition $r\gtrsim\log^{3}n$
might be improved to a lower-order $\mathsf{poly}\log\left(n\right)$ term
using more refined analyses. 
When $r\lesssim\log n$, one can 
compute the maximum likelihood (ML) estimate via dynamic programming
\cite{kamath2015optimal} within polynomial time.
\end{remark}

Theorem \ref{theorem:rings} uncovers a surprising insensitivity phenomenon
for rings: as long as the measurement graph is sufficiently connected,
the locality constraint does not alter the sample complexity limit
and the computational limit at all. This subsumes as special cases two
regimes that exhibit dramatically different graph structures: (1)
complete graphs, where the samples are taken in a global manner, and
(2) rings with $r=O(\mathrm{poly}\log\left(n\right))$, where the
samples are constrained within highly local neighborhood. 
In addition, Theorem \ref{theorem:rings} does not impose any assumption on the ground truth $\{X_i: 1\leq i\leq n\}$; in fact, the
success probability of the proposed algorithms is independent of the true community assignment. 

Notably, both \cite{abbe2014decoding} and \cite{hajek2015exact} have derived general sufficient recovery conditions of SDP which, however, depend on the second-order graphical metrics of $\mathcal{G}$ \cite{durrett2007random} (e.g.~the spectral gap or Cheeger constant). When applied to rings (or other graphs with locality), the sufficient sample complexity given therein is significantly larger than the information limit\footnote{For instance, the sufficient sample complexity given in \cite{abbe2014decoding} scales as $\frac{n\log n}{h_{\mathcal{G}}D^*}$ with $h_{\mathcal{G}}$ denoting the Cheeger constant. Since $h_{\mathcal{G}} = O(1/n)$ for rings / lines, this results in a sample size that is about $n$ times larger than the information limit. 
}. This is in contrast to our finding, which reveals that for many graphs with locality, both the information and computation limits often depend only upon the vertex degrees independent of these second-order graphical metrics.

\subsubsection{Bottlenecks for Exact Recovery\label{sub:Bottleneck-rings}}

Before explaining the rationale of the proposed algorithms, we
provide here some heuristic argument as to why $n\log n$ samples are necessary
for exact recovery and where the recovery bottleneck lies. 

Without loss of generality, assume $\bm{X}=[0,\cdots,0]^{\top}$.
Suppose the genie tells us the correct labels of all nodes except
$v$. Then all samples useful for recovering $X_{v}$
reside on the edges connecting $v$ and its neighbors, and there
are $\mathsf{Poisson}(\lambda d_{v})$ such samples. Thus,
this comes down to testing between two conditionally i.i.d.~distributions
with a Poisson sample size of mean $\lambda d_{v}$. From the large
deviation theory, the ML rule 
fails in recovering $X_{v}$ with probability
\begin{equation}
P_{\text{e},v}\approx\exp\left\{ -\lambda d_{v}(1-e^{-D^{*}})\right\} ,\label{eq:large-deviation}
\end{equation}
where $D^*$ is the large deviation exponent. 
The above argument concerns a typical error event for recovering a
single node $v$, and it remains to accommodate all vertices. Since
the local neighborhoods of two vertices $v$ and $u$ are nearly non-overlapping,
the resulting typical error events for recovering $X_{v}$ and $X_{u}$
become almost independent and disjoint. As a result, the probability of error of
the ML rule $\psi_{\mathrm{ml}}$ is approximately lower bounded by
\begin{equation}
P_{\mathrm{e}}(\psi_{\mathrm{ml}})\gtrsim \sum\nolimits _{v=1}^{n}P_{\mathrm{e},v}\approx n\exp\left\{ -\lambda d_{\mathrm{avg}}(1-e^{-D^{*}})\right\} ,\label{eq:union}
\end{equation}
where one uses the fact that $d_{v}\equiv d_{\mathrm{avg}}$. Apparently, the right-hand
side of (\ref{eq:union}) would vanish only if 
\begin{equation}
\lambda d_{\mathrm{avg}}(1-e^{-D^{*}})>\log n.\label{eq:lambda-d-LB}
\end{equation}
Since the total sample size is $m=\lambda\cdot\frac{1}{2}nd_{\mathrm{avg}}$,
this together with (\ref{eq:lambda-d-LB}) confirms the sample complexity
lower bound 
\[
m=\frac{1}{2}\lambda nd_{\mathrm{avg}}>\frac{n\log n}{2\left(1-e^{-D^{*}}\right)}=m^{*}.
\]
As we shall see, the above error events---in which only a single variable
is uncertain---dictate the hardness of exact recovery.

\subsubsection{Interpretation of Our Algorithms}

The preceding argument suggests that the recovery bottleneck of an
optimal algorithm should also be determined by the aforementioned typical error
events. This is the case for both {\SpectralExpand} and {\SpectralStitch},
as revealed by the intuitive arguments below. While the intuition
is provided for rings, it contains all important ingredients that
apply to many other graphs. 

To begin with, we provide an heuristic argument for {\SpectralExpand}.

\begin{enumerate}
\item[(i)] Stage 1 focuses on a core complete subgraph $\mathcal{G}_{\mathrm{c}}$.
In the regime where $m\gtrsim n\log n$, the total number of samples
falling within $\mathcal{G}_{\mathrm{c}}$ is on the order of $\frac{|\mathcal{V}_{\mathrm{c}}|}{n}\cdot m\geq|\mathcal{V}_{\mathrm{c}}|\log n$,
which suffices in guaranteeing partial recovery using spectral methods
\cite{chin2015stochastic}. In fact, the sample size we have available
over $\mathcal{G}_{\mathrm{c}}$ is way above the degrees of freedom
of the variables in $\mathcal{G}_{\mathrm{c}}$ (which is $r$).

\item[(ii)] With decent initial estimates for $\mathcal{G}_{\mathrm{c}}$
in place, one can infer the remaining pool of vertices one by one
using existing estimates together with backward samples. One important
observation is that each vertex is incident to many---i.e.~about
the order of $\log n$---backward samples. That said, we are effectively
operating in a high signal-to-noise ratio (SNR) regime. While existing
estimates are imperfect, the errors occur only to a small fraction
of vertices. Moreover, these errors are in some sense randomly distributed
and hence fairly spread out, thus precluding the possibility of bursty
errors. Consequently, one can obtain correct estimate for each of
these vertices with high probability, leading to a vanishing fraction
of errors in total.

\item[(iii)] Now that we have achieved approximate recovery, all remaining
errors can be cleaned up via local refinement using all backward and
forward samples. For each vertex, since only a vanishingly small fraction
of its neighbors contain errors, the performance of local refinement
is almost the same as in the case where all neighbors have been perfectly
recovered. 

\end{enumerate}

The above intuition extends to {\SpectralStitch}. Following the argument in (i), 
we see that the spectral method returns nearly accurate estimates for each of the subgraph $\mathcal{G}_{\mathrm{l}}$ induced by $\mathcal{V}_l$, except for the global phases 
 (this arises because each $\mathcal{G}_{\mathrm{l}}$ has been estimated in isolation, without using any information concerning the global phase). 
Since any two adjacent $\mathcal{G}_{l}$ and $\mathcal{G}_{l+1}$ have sufficient overlaps, this allows us to 
calibrate the global phases for $\{X_i^{\mathcal{V}_l}:i\in \mathcal{V}_l\}$ and $\{X_i^{\mathcal{V}_{l+1}}:i\in \mathcal{V}_{l+1}\}$. Once we obtain approximate recovery for all variables simultaneously, the remaining errors can then be cleaned up by Stage 3 as in {\SpectralExpand}.

We emphasize that the first two stages of both algorithms---which aim at approximate
recovery---require only $O\left(n\right)$ samples (as long as the
pre-constant is sufficiently large). In contrast, the final stage is
the bottleneck: it succeeds as long as local refinement for each vertex
is successful. The error events for this stage are almost equivalent
to the typical events singled out in Section \ref{sub:Bottleneck-rings},
justifying the information-theoretic optimality of both algorithms.

\subsection{Theoretical Guarantees: Inhomogeneous Graphs\label{sec:theory-general}}

The proposed algorithms are guaranteed to succeed for a much broader
class of graphs with locality beyond rings, including those
that exhibit inhomogeneous vertex degrees. The following theorem formalizes
this claim for two of the most important instances: lines and grids. 

\begin{figure}
  \centering
  \includegraphics[width=0.35\textwidth]{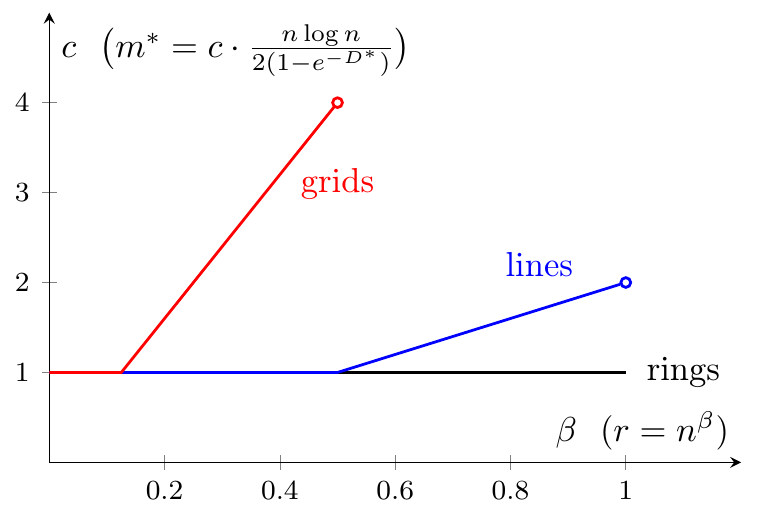}\qquad\includegraphics[width=0.35\textwidth]{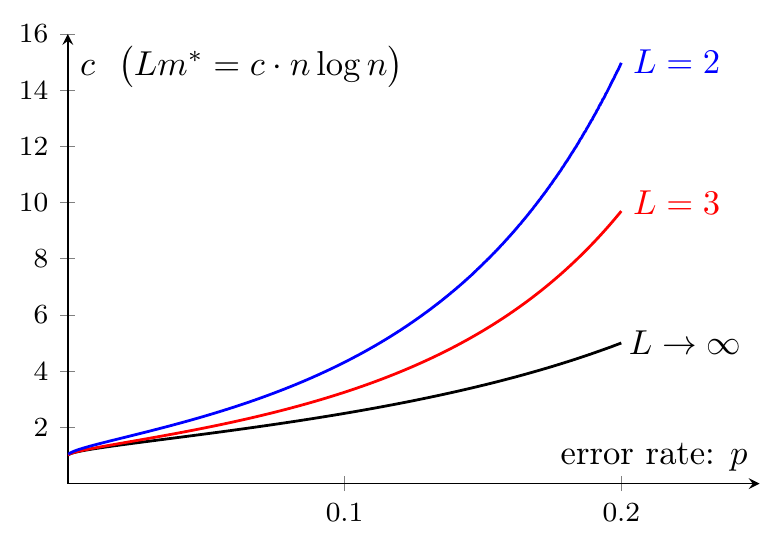}
  \caption{(Left) Minimum sample complexity $m^*$ vs.~locality radius $r$; (Right) Minimum number $Lm^*$ of vertices being measured (including repetition) vs.~single-vertex error rate $p$.}
  \label{fig:beyond-pairwise}
\end{figure}

\begin{figure}
\centering\includegraphics[scale=0.36]{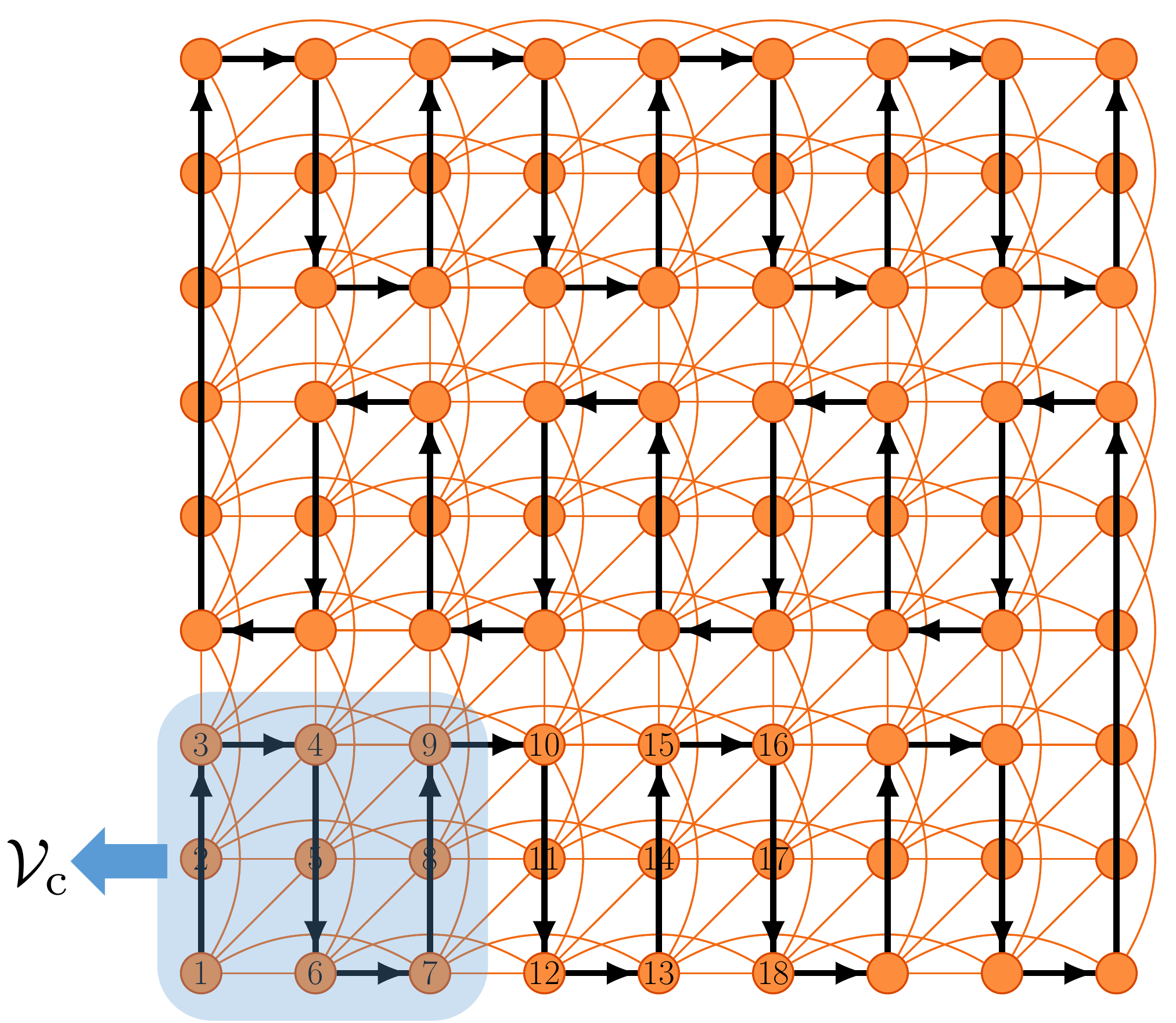}
\caption{\label{fig:grid-order}Labeling / ordering of the vertex set for a
grid, where the core subgraph consists of the $r^{2}$ vertices on
the bottom left.}
\end{figure}

\begin{theorem}\label{theorem:Lines-Grids}Theorem \ref{theorem:rings}
continues to hold for the following families of measurement graphs:

(1) Lines with $r=n^{\beta}$ for some constant $0<\beta<1$, where
\begin{equation}
m^{*}=\frac{\max\left\{ 1/2,~\beta\right\} n\log n}{1-e^{-\mathsf{KL}\left(0.5 \hspace{0.1em} \| \hspace{0.1em} \theta\right)}};
\label{eq:sample-complexity-lines}
\end{equation}

(2) Grids with $r=n^{\beta}$ for some constant $0<\beta<0.5$, where
\begin{equation}
m^{*}=\frac{\max\left\{ 1/2,~4\beta\right\} n\log n}{1-e^{-\mathsf{KL}\left(0.5 \hspace{0.1em} \| \hspace{0.1em}\theta\right)}}.
\label{eq:sample-complexity-grid}
\end{equation}

\end{theorem}

\begin{remark}
Note that both {\SpectralExpand}  and {\SpectralStitch}  rely on
the labeling / ordering of the vertex set $\mathcal{V}$. 
For lines, it suffices to employ the same ordering and
core subgraph as for rings. 
For grids, we can start by taking the
core subgraph to be a subsquare of area $r^{2}$ lying on the bottom
left of the grid, and then follow a serpentine trajectory running
alternately from the left to the right and then back again; see Fig.~\ref{fig:grid-order}
for an illustration.
\end{remark}

\begin{remark}Careful readers will note that for lines (resp.~grids), $m^{*}$ does not
converge to $\frac{n\log n}{2(1-e^{-\mathsf{KL}\left(0.5\|\theta\right)})}$
as $\beta\rightarrow1$ (resp.~$\beta\rightarrow 0.5$), which is the case of complete graphs. This
arises because $m^{*}$ experiences a more rapid drop in the regime
where $\beta=1$ (resp.~$\beta=0.5$). For instance, for a line with $r=\gamma n$ for
some constant $\gamma>0$, one has $m^{*}=\frac{\left(1-\gamma/2\right)n\log n}{1-e^{-\mathsf{KL}\left(0.5\|\theta\right)}}$.
\end{remark}

Theorem \ref{theorem:Lines-Grids} characterizes the effect of locality
radius upon the sample complexity limit; see Fig.~\ref{fig:beyond-pairwise}
for a comparison of three classes of graphs. In contrast to rings, lines
and grids are spatially varying models due to the presence of boundary
vertices, and the degree of graph inhomogeneity increases in the locality
radius $r$. To be more concrete, consider, for example, the first $d_{\mathrm{avg}}/\log n$
vertices of a line, which have degrees around $d_{\mathrm{avg}}/2$.
In comparison, the set of vertices lying away from the boundary have
degrees as large as $d_{\mathrm{avg}}$. This tells us that the first
few vertices form a weakly connected component, thus presenting an
additional bottleneck for exact recovery. This issue is negligible
unless the size of the weakly connected component is exceedingly large.
As asserted by Theorem \ref{theorem:Lines-Grids}, the minimum sample
complexity for lines (resp.~grids) is identical to that for rings
unless $r\gtrsim\sqrt{n}$ (resp.~$r\gtrsim n^{1/8}$). 
Note that the curves for lines and grids (Fig.~\ref{fig:beyond-pairwise})
have distinct hinge points primarily because the vertex degrees of
the corresponding weakly connected components differ.

More precisely, the insights developed in Section \ref{sub:Bottleneck-rings}
readily carry over here. Since the error probability of the ML
rule is lower bounded by (\ref{eq:union}),
everything boils down to determining the smallest $\lambda$ (called
$\lambda^{*}$) satisfying
\[
\sum\nolimits _{v=1}^{n}\exp\left\{ -\lambda^{*}d_{v}\left(1-e^{-D^{*}}\right)\right\} \rightarrow0,
\]
which in turn yields 
$m^{*}=\frac{1}{2}\lambda^{*}d_{\mathrm{avg}}n$. 
The two cases accommodated by Theorem \ref{theorem:Lines-Grids} can
all be derived in this way.

\subsection{Connection to Low-Rank Matrix Completion\label{sub:Connection-MC}}

One can aggregate all correct parities into a matrix $\bm{Z}=\left[Z_{i,j}\right]_{1\leq i,j\leq n}$
such that $Z_{i,j}=1$ if $X_{i}=X_{j}$ and $Z_{i,j}=-1$ otherwise.
It is straightforward to verify that $\mathrm{rank}\left(\bm{Z}\right)=1$,
with each $Y_{i,j}^{(l)}$ being a noisy measurement of $Z_{i,j}$.
Thus, our problem falls under the category of low-rank matrix
completion, a topic that has inspired a flurry of research
(e.g.~\cite{ExactMC09,keshavan2010mc,CanLiMaWri09,chandrasekaran2011rank,chen2011low}).
Most prior works, however, concentrated on samples taken over an Erd\H{o}s\textendash Rényi
model, without investigating sampling schemes with locality constraints.
One exception is \cite{bhojanapalli2014universal}, which explored
the effectiveness of SDP under general sampling schemes. However,
the sample complexity required therein increases significantly as
the spectral gap of the measurement graph drops, which does not deliver
optimal guarantees. We believe that the approach developed herein
will shed light on solving general matrix completion problems from
samples with locality.

\section{Extension: Beyond Pairwise Measurements\label{sec:Extension}}

The proposed algorithms are applicable to numerous scenarios beyond
the basic setup in Section \ref{sec:Problem-Setup}. 
This section presents two important extension beyond pairwise measurements. 


\subsection{Sampling with Nonuniform Weight\label{sec:nonuniform-weight}}


In many applications, the sampling rate is nonuniform across different edges; for instance, it might fall off with distance between two incident vertices. In the haplotype phasing application,  Fig.~\ref{fig:haplotype_assembly}(a) gives an example of a distribution of the separation between mate-paired reads (insert size). One would naturally wonder whether our algorithm works under this type of more realistic models.

More precisely, suppose the number of samples over each $(i,j)\in\mathcal{E}$
is independently generated obeying 
\begin{equation}
N_{i,j}\overset{\text{ind}}{\sim}\mathsf{Poisson}\left(\lambda w_{i,j}\right),\label{eq:nonuniform-model}
\end{equation}
where $w_{i,j}>0$ incorporates a sampling rate weighting for each
edge. This section focuses on lines / grids / rings for concreteness, where
we impose the following assumptions in order to make the sampling
model more ``symmetric'':
\begin{enumerate}
\item[(i)] \emph{Lines / grids}: $w_{i,j}$ depends only on the Euclidean
distance between vertices $i$ and $j$;

\item[(ii)] \emph{Rings}: $w_{i,j}$ depends only on $i-j$ ($\mathsf{mod}$
$n$). 
\end{enumerate}

\begin{theorem}\label{theorem:nonuniform-weight}
Theorems \ref{theorem:rings}-\ref{theorem:Lines-Grids}
continue to hold under the above nonuniform sampling model,
provided that $\frac{\max_{(i,j)\in\mathcal{E}}w_{i,j}}{\min_{(i,j)\in\mathcal{E}}w_{i,j}}$
is bounded. \end{theorem}

Theorem \ref{theorem:nonuniform-weight} might be surprising at first
glance: both the performance of our algorithms  and
the fundamental limit depend only on the weighted average of the vertex
degrees, and are insensitive to the degree distributions. This can
be better understood by examining the three stages of {\SpectralExpand} and {\SpectralStitch}.
To begin with, Stages 1-2 are still guaranteed to work gracefully,
since we are still operating in a high SNR regime irrespective of
the specific values of $\left\{ w_{i,j}\right\} $. The main task
thus amounts to ensuring the success of local clean-up. Repeating
our heuristic treatment in Section \ref{sub:Bottleneck-rings}, one
sees that the probability of each singleton error event (i.e.~false
recovery of $X_{v}$ when the genie already reveals the true labels
of other nodes) depends only on the average number of samples incident
to each vertex $v$, namely, 
\[
\mathbb{E}\left[N_{v}\right]:=\sum\nolimits _{j}\mathbb{E}\left[N_{v,j}\right]=\sum\nolimits _{j:(v,j)\in\mathcal{E}}\lambda w_{v,j}.
\]
Due to the symmetry assumptions on $\left\{ w_{i}\right\} $, the
total sample size $m$ scales linearly with $\mathbb{E}\left[N_{v}\right]$,
and hence the influence of $\left\{ w_{i}\right\} $ is absorbed into
$m$ and ends up disappearing from the final expression.

Another prominent example is the class of small-world graphs. In various
human social networks, one typically observes both local friendships
and a (significantly lower) portion of long-range connections, and
small-world graphs are introduced to incorporate this feature. To
better illustrate the concept, we focus on the following spatially-invariant
instance, but it naturally generalizes to a much broader family. 

\begin{itemize}
\item \emph{Small-world graphs}. Let $\mathcal{G}$ be a superposition of
a complete graph $\mathcal{G}_{0}=\left(\mathcal{V},\mathcal{E}_{0}\right)$
and a ring $\mathcal{G}_{1}$ with connectivity radius $r$. The sampling
rate is given by
\[
N_{i,j}\overset{\text{ind.}}{\sim}\begin{cases}
\mathsf{Poisson}\left(w_{0}\right),\quad & \text{if }(i,j)\in\mathcal{E}_{0};\\
\mathsf{Poisson}\left(w_{1}\right), & \text{else}.
\end{cases}
\]
We assume that $\frac{w_{0}n^{2}}{w_{1}nr}=O\left(1\right)$, in order
to ensure higher weights for local connections. 
\end{itemize}

\begin{theorem}\label{theorem:small-world}Theorem \ref{theorem:rings}
continues to hold under the above small-world graph model, provided that $r\gtrsim \log^3 n$. \end{theorem}

\subsection{Beyond Pairwise Measurements\label{sec:BeyondPairwise}}

In some applications, each measurement may cover more than two nodes in the graph. In the haplotype phasing application, for example, a new sequencing technology called 10x \cite{10xGenomics} generates barcodes to mark reads from the same chromosome (maternal or paternal), and more than two reads can have the same barcode. For concreteness,
we suppose the locality constraint is captured by rings, and consider
the type of multiple linked samples as follows. 

\begin{itemize}
\item
\textbf{Measurement (hyper)-graphs}. Let $\mathcal{G}_{0}=\left(\mathcal{V},\mathcal{E}_{0}\right)$
be a ring $\mathcal{R}_{r}$, and let $\mathcal{G}=(\mathcal{V},\mathcal{E})$
be a hyper-graph such that (i) every hyper-edge is incident to $L$
vertices in $\mathcal{V}$, and (ii) all these $L$ vertices are mutually
connected in $\mathcal{G}_{0}$. 

\item
\textbf{Noise model}. On each hyper-edge $e=(i_{1},\cdots,i_{L})\in\mathcal{G}$,
we obtain $N_{e}\overset{\text{ind.}}{\sim}\mathsf{Poisson}\left(\lambda\right)$
multi-linked samples $\{Y_{e}^{(l)}\mid1\leq l\leq N_{e}\}$. Conditional
on $N_{e}$, each sample $Y_{e}^{(l)}$ is an independent copy of
\begin{equation}
Y_{e}=\begin{cases}
\left(Z_{i_{1}},\cdots,Z_{i_{L}}\right),\quad & \text{with prob. }0.5,\\
\left(Z_{i_{1}}\oplus1,\cdots,Z_{i_{L}}\oplus1\right),\quad & \text{else},
\end{cases}\label{eq:L-wise-model}
\end{equation}
where $Z_{i}$ is a noisy measurement of $X_{i}$ such that
\begin{align}
Z_{i}=\begin{cases}
X_{i},  & \text{with probability }1-p; \\
X_{i}\oplus 1, &\text{otherwise.}
\end{cases}
\end{align}
 Here,
$p$ represents the error rate for measuring a single vertex. For
the pairwise samples considered before, one can think of the parity
error rate $\theta$ as $\mathbb{P}\left\{ Z_{i}\oplus Z_{j}\neq X_{i}\oplus X_{j}\right\} $ or, 
equivalently, $\theta=2p(1-p)$. 
\end{itemize}

We emphasize that a random global phase is incorporated into each
sample (\ref{eq:L-wise-model}). That being said, each sample reveals
only the \emph{relative} similarity information among these $L$ vertices,
without providing further information about the absolute cluster membership.

Since Algorithm \ref{alg:Algorithm-progressive} and Algorithm \ref{alg:Algorithm-stitch} operate only upon pairwise measurements, one alternative is to convert each
$L$-wise sample $Y_{e}=\left(Y_{i_{1}},\cdots,Y_{i_{L}}\right)$
into ${L \choose 2}$ pairwise samples of the form $Y_{i_{j}}\oplus Y_{i_{l}}$
(for all $j\neq l$), and then apply the spectral methods on these parity samples. In addition, the majority voting procedure specified in Algorithm \ref{alg:Algorithm-progressive}  needs to 
be replaced by certain local maximum likelihood rule as well, in order to take advantage of the mutual data correlation within each $L$-wise measurement.
The modified algorithms are summarized in Algorithms \ref{alg:Algorithm-multi-SpectralExpand} and \ref{alg:Algorithm-multi-SpectralStitch}.
Interestingly,  these algorithms are still information-theoretically
optimal, as asserted by the following theorem.

\begin{theorem}\label{theorem:beyond-pairwise}Fix $L\geq2$, and consider Algorithms \ref{alg:Algorithm-multi-SpectralExpand} and \ref{alg:Algorithm-multi-SpectralStitch}
. Theorem
\ref{theorem:rings} continues to hold under the above $L$-wise sampling
model,  with $m^{*}$ replaced by
\[
m^{*}:=\frac{n\log n}{L\left(1-e^{-D\left(P_{0},P_{1}\right)}\right)}.
\]
Here,
\begin{equation}
\begin{cases}
P_{0} = (1-p)\mathsf{Binomial}\left(L-1,p\right)+p\mathsf{Binomial}\left(L-1,1-p\right);\\
P_{1} = p\mathsf{Binomial}\left(L-1,p\right)+\left(1-p\right)\mathsf{Binomial}\left(L-1,1-p\right).
\end{cases}
\end{equation}

 \end{theorem}

Here, the Chernoff information $D(P_0,P_1)$ can be expressed in closed form as
\begin{equation}
D(P_{0},P_{1})=-\log\left\{ \sum_{i=0}^{L-1}{L-1 \choose i}\sqrt{\left\{ p^{i}\left(1-p\right)^{L-i}+\left(1-p\right)^{i}p^{L-i}\right\} \left\{ p^{i+1}\left(1-p\right)^{L-i-1}+\left(1-p\right)^{i+1}p^{L-i-1}\right\} }\right\} .\label{eq:expression-Chernoff}
\end{equation}
In particular, when $L=2$, this reduces to\footnote{This follows since, when $L=2$,
\[
D(P_{0},P_{1})=-\log\left\{ 2\sqrt{(\left(1-p\right)^{2}+p^{2})\left(2p\left(1-p\right)\right)}\right\} =-\log\left\{ 2\sqrt{\left(1-\theta\right)\theta}\right\} =\mathsf{KL}\left(0.5\hspace{0.2em}\|\hspace{0.2em}\theta\right).
\]} 
$D(P_{0},P_{1})=\mathsf{KL}\left(0.5\hspace{0.2em}\|\hspace{0.2em}\theta\right)$
for $\theta:=2p(1-p)$, which matches our results with pairwise samples.

Interestingly,   $D(P_{0},P_{1})$ enjoys a very simple asymptotic limit
as $L$ scales, as stated in the following lemma.

\begin{lem}\label{lemma:Chernoff-L-infity} Fix any $0<p<1/2$. The Chernoff information $D(P_{0},P_{1})$ given in Theorem \ref{theorem:beyond-pairwise} obeys
\begin{equation}
\lim_{L\rightarrow\infty}D\left(P_{0},P_{1}\right)=\mathsf{KL}\left(0.5\hspace{0.2em}\|\hspace{0.2em}p\right).\label{eq:D-limit-large-L}
\end{equation}
\end{lem}
\begin{proof}See Appendix \ref{sec:Proof-of-Lemma-Chernoff-L-infty}.\end{proof}

\begin{remark}The asymptotic limit (\ref{eq:D-limit-large-L}) admits
a simple interpretation. Consider the typical event where only $X_{1}$
is uncertain and $X_{2}=\cdots=X_{n}=0$. Conditional on $Z_{1}$,
the $L-1$ parity samples $\left(Z_{1}\oplus Z_{2},\cdots,Z_{1}\oplus Z_{L}\right)$
are i.i.d., which reveals accurate information about $Z_{1}\oplus0$
in the regime where $L\rightarrow\infty$ (by the law of large number).
As a result, the uncertainty arises only because $Z_{1}$ is a noisy
version of $X_{1}$, which behaves like passing $X_{1}$ through a
binary symmetric channel with crossover probability $p$. This essentially
boils down to distinguishing between $\mathsf{Bernoulli}\left(p\right)$
(when $X_{1}=0$) and $\mathsf{Bernoulli}\left(1-p\right)$ (when
$X_{1}=1$), for which the associated Chernoff information is known
to be $\mathsf{KL}\left(0.5\hspace{0.2em}\|\hspace{0.2em}p\right)$. \end{remark}

%

With Theorem \ref{theorem:beyond-pairwise} in place, we can determine
the benefits of multi-linked sampling. To enable a fair comparison,
we evaluate the sampling efficiency in terms of $Lm^{*}$ rather than
$m^{*}$, since $Lm^{*}$ captures the total number of vertices (including
repetition) one needs to measure. As illustrated in Fig.~\ref{fig:beyond-pairwise},
the sampling efficiency improves as $L$ increases, but there exists
a fundamental lower barrier given by 
$
\frac{n\log n}{1-e^{-\mathsf{KL}\left(0.5\|p\right)}}.
$
This lower barrier, as plotted in the black curve of Fig.~\ref{fig:beyond-pairwise},
corresponds to the case where $L$ is approaching infinity.

\begin{algorithm*}
\begin{enumerate}
\item Break each $L$-wise sample $Y_{e}=\left(Y_{i_{1}},\cdots,Y_{i_{L}}\right)$
into ${L \choose 2}$ pairwise samples of the form $Y_{i_{j}}\oplus Y_{i_{l}}$
(for all $j\neq l$), and\textbf{ run spectral method (Algorithm \ref{alg:Algorithm-spectral}) on a core subgraph} induced by $\mathcal{V}_{\mathrm{c}}$
using these parity samples. This yields estimates $X_{j}^{(0)},1\leq j\leq|\mathcal{V}_{\mathrm{c}}|$.
\item \textbf{Progressive estimation}: for $k=|\mathcal{V}_{\mathrm{c}}|+1,\cdots,n$,
\[
X_{k}^{(0)}\leftarrow\mathsf{local-ML}_{\left\{ X_{i}^{(0)}\mid1\leq i<k\right\} }\left\{ Y_{e}^{(l)}\mid e=\left(i_{1},\cdots,i_{L}\right)\text{ with }i_{L}=k,\text{ }1\leq l\leq N_{e}\right\} .
\]

\item \textbf{Successive local refinement:} for $t=0,\cdots,T-1$, 
\begin{eqnarray*}
X_{k}^{(t+1)} & \leftarrow & \mathsf{local-ML}_{\left\{ X_{i}^{(t)}\mid i\neq k\right\} }\left\{ Y_{e}^{(l)}\mid k\in e,\text{ }1\leq l\leq N_{e}\right\} ,\quad1\leq k\leq n.
\end{eqnarray*}

\item \textbf{Output} $X_{k}^{(T)}$, $1\leq k\leq n$. 
\end{enumerate}
Here, $\mathsf{local-ML}_{\mathcal{X}}\left\{ \cdot\right\} $ represents
the local maximum likelihood rule: for any sequence $s_{1},\cdots,s_{N}\in\left\{ 0,1\right\} ^{L}$,
\[
\mathsf{local-ML}_{\left\{ Z_{i}\mid1\leq i<k\right\} }\left\{ s_{1},\cdots,s_{N}\right\} =\begin{cases}
1,\quad & \text{if }\sum_{j=1}^{N}\log\frac{\mathbb{P}\left\{ s_{j}\mid X_{k}=1,X_{i}=Z_{i}\text{ }(1\leq i<k)\right\} }{\mathbb{P}\left\{ s_{j}\mid X_{k}=0,X_{i}=Z_{i}\text{ }(1\leq i<k)\right\} }\geq0,\\
0, & \text{else},
\end{cases}
\]
and
\[
\mathsf{local-ML}_{\left\{ Z_{i}\mid i\neq k\right\} }\left\{ s_{1},\cdots,s_{N}\right\} =\begin{cases}
1,\quad & \text{if }\sum_{j=1}^{N}\log\frac{\mathbb{P}\left\{ s_{j}\mid X_{k}=1,X_{i}=Z_{i}\text{ }(i\neq k)\right\} }{\mathbb{P}\left\{ s_{j}\mid X_{k}=0,X_{i}=Z_{i}\text{ }(i\neq k)\right\} }\geq0,\\
0, & \text{else}.
\end{cases}
\]

\caption{\label{alg:Algorithm-multi-SpectralExpand}: {\SpectralExpand}  for
multi-linked samples}
\end{algorithm*}

\begin{algorithm*}
\begin{enumerate}
\item Break each $L$-wise sample $Y_{e}=\left(Y_{i_{1}},\cdots,Y_{i_{L}}\right)$
into ${L \choose 2}$ pairwise samples of the form $Y_{i_{j}}\oplus Y_{i_{l}}$
(for all $j\neq l$). Run Steps 1-2 of Algorithm \ref{alg:Algorithm-stitch}
using these pairwise samples to obtain estimates $X_{j}^{(0)},1\leq j\leq n$. 
\item Run Step 3 of Algorithm \ref{alg:Algorithm-multi-SpectralExpand}
and output $X_{k}^{(T)}$, $1\leq k\leq n$. 
\end{enumerate}
\vspace{-0.5em}
\caption{\label{alg:Algorithm-multi-SpectralStitch}: {\SpectralStitch}  for
multi-linked samples}
\end{algorithm*}

\begin{figure*}
\centering%
\begin{tabular}{cc}
\includegraphics[height=.24\textwidth]{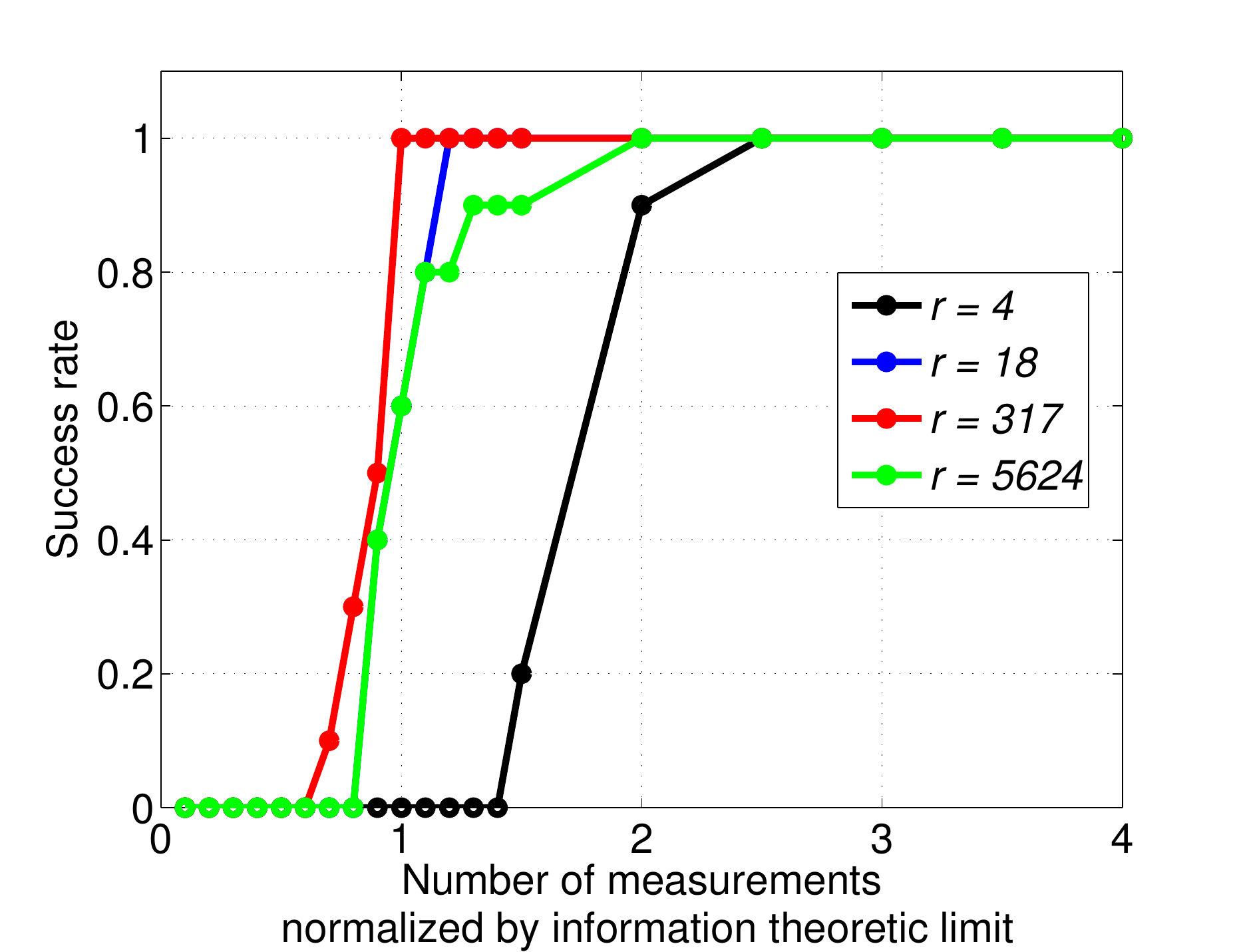} 
& \includegraphics[height=.24\textwidth]{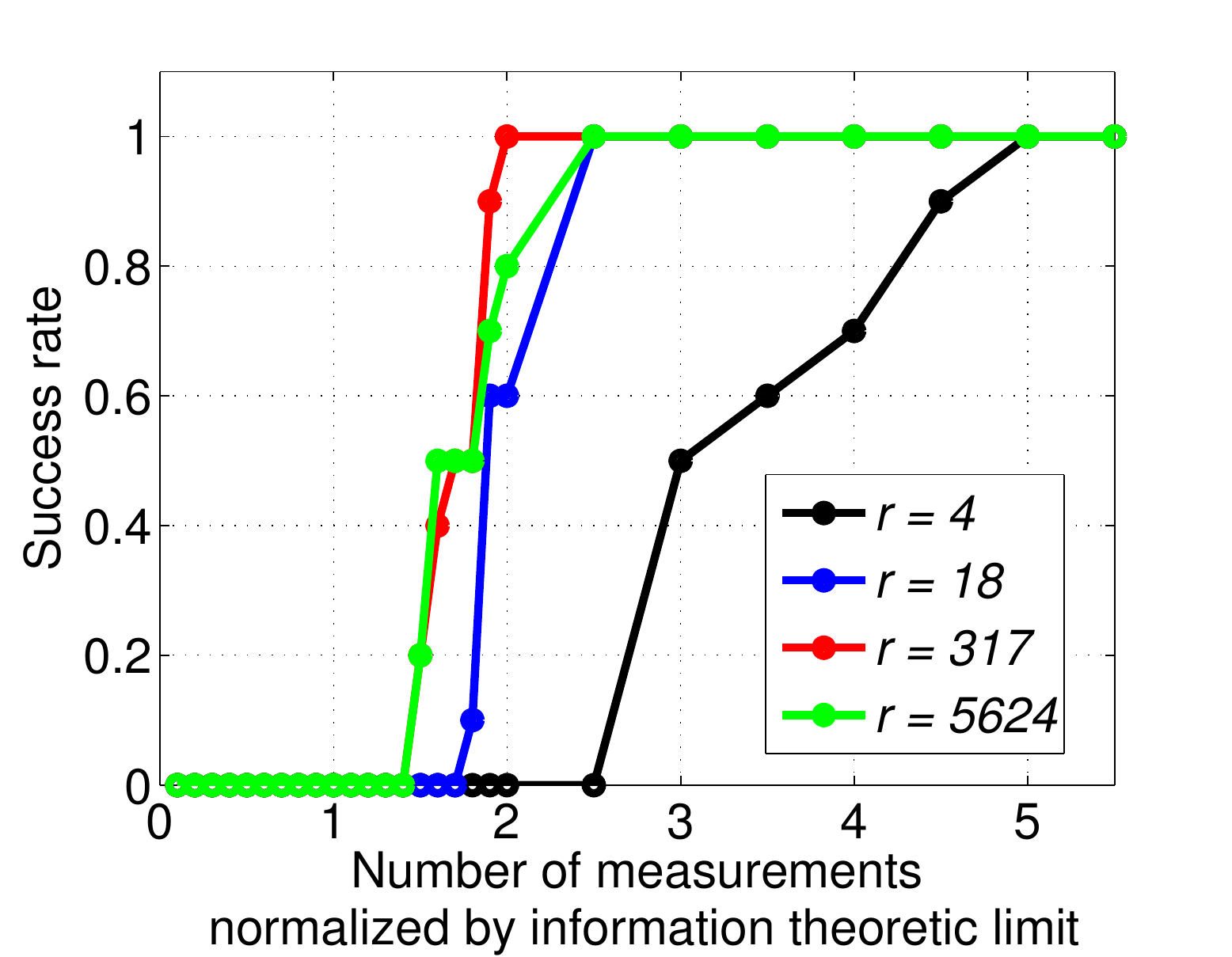} 
\tabularnewline 
 (a)  &  (b)   
\tabularnewline
 
 \includegraphics[height=.24\textwidth]{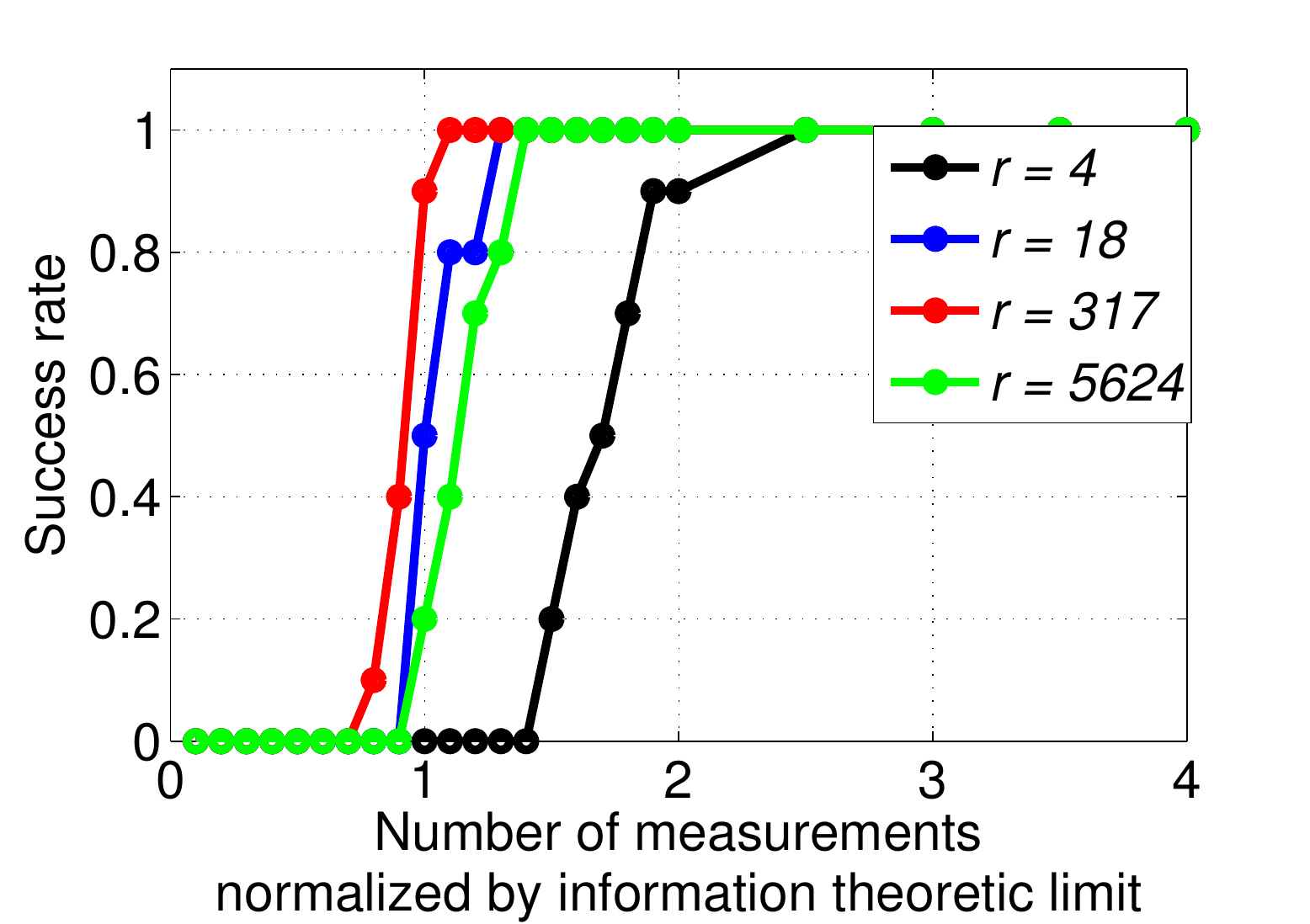} 
& \includegraphics[height=.24\textwidth]{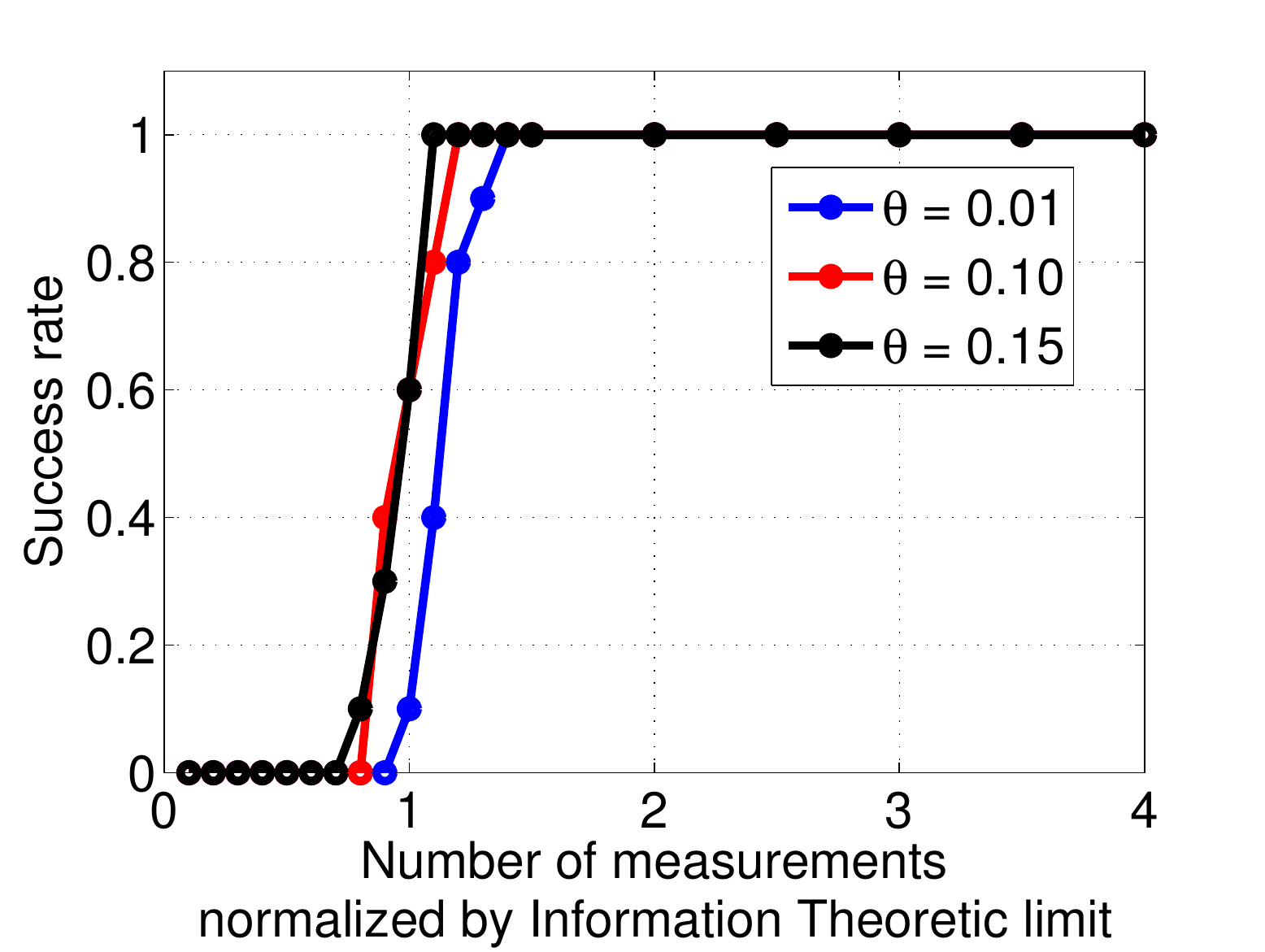}
  \tabularnewline 
 (c)  &  (d)  \tabularnewline
\end{tabular}
\caption{Empirical success rate of \SpectralExpand ~for: (a) Rings $\mathcal{R}_r$;  
(b)  Small world graphs; 
(c) Rings $\mathcal{R}_{r}$ with non-uniform sampling rate; 
and (d) Rings $\mathcal{R}_{18}$ with varied measurement error rate $\theta$.
Here, the x-axis is the sample size $m$ normalized by the information limit $m^*$.}
\label{fig:plots}
\end{figure*}

\begin{figure*}
 \centering
 \begin{tabular}{ccc}
 \includegraphics[height=0.2\linewidth]{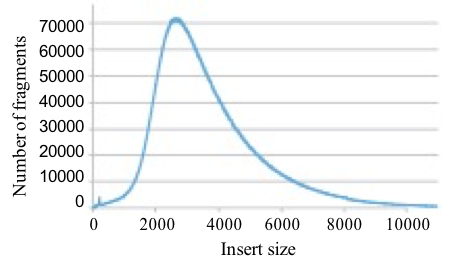} &
  \includegraphics[height=.2\linewidth]{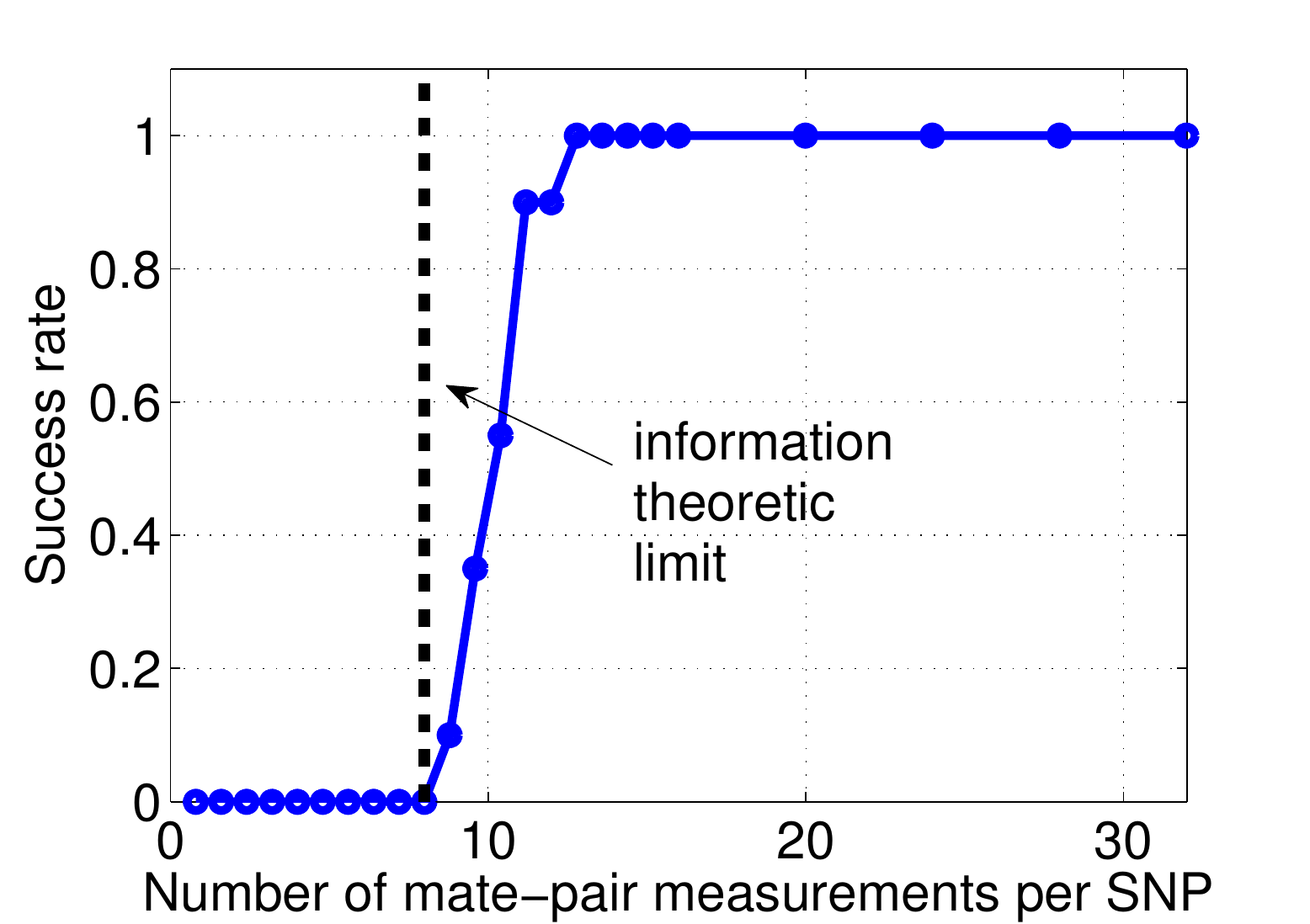} & 
  \includegraphics[height=.2\linewidth]{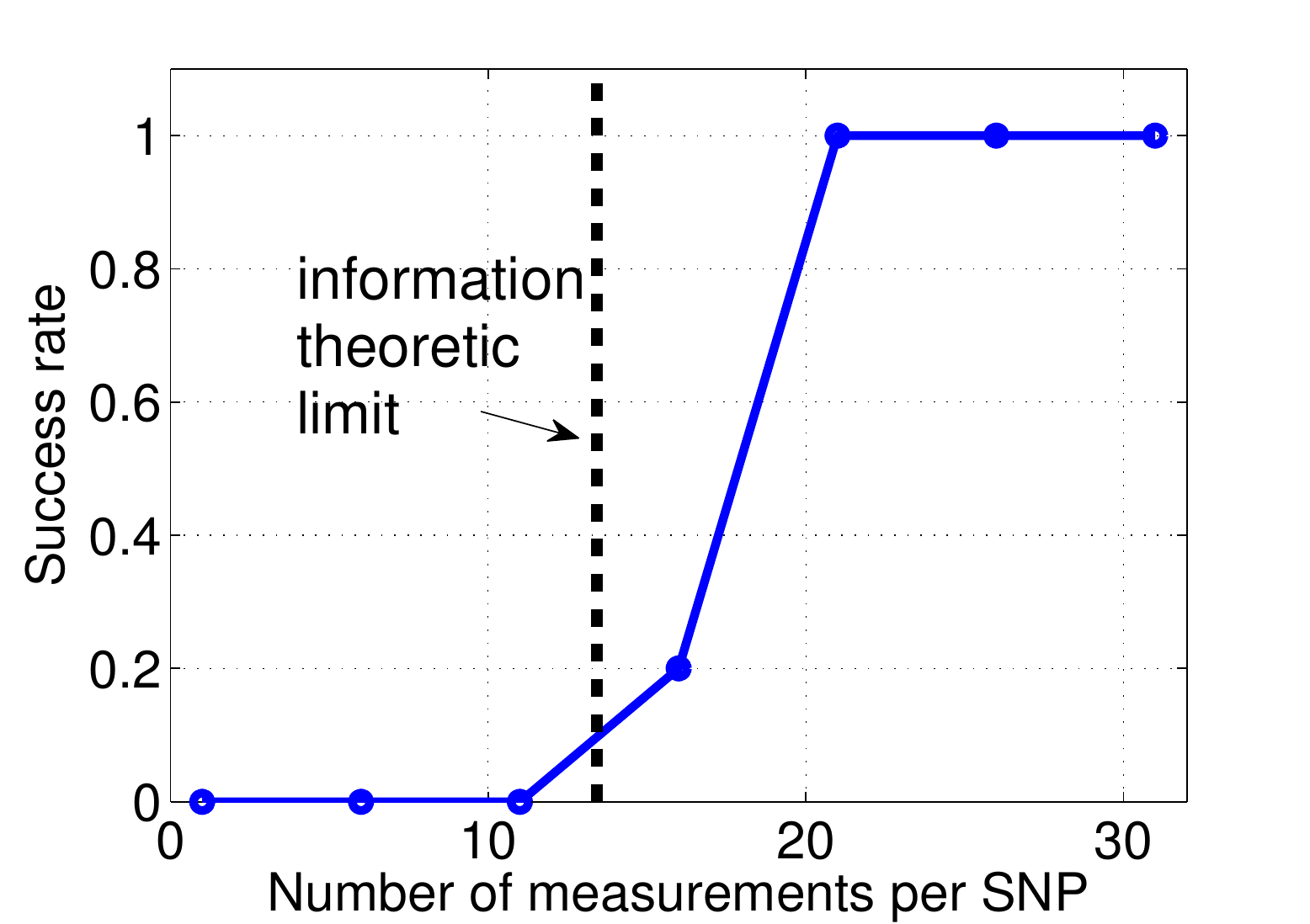}\tabularnewline 
  (a)  & (b) & (c) \tabularnewline
 \end{tabular}
 \caption{(a) Insert size distribution \cite{matepair2012};
 (b) Performance of {\SpectralExpand} on a simulation of haplotype
phasing  from mate pair reads;
 (c) Performance of {\SpectralExpand} on a simulation of haplotype
phasing from 10x like reads.  }
 \label{fig:haplotype_assembly}
\end{figure*}

\section{Numerical Experiments}
\label{sec:numerical}

To verify the practical applicability of the proposed algorithms, we have conducted simulations in various settings. All these experiments focused on graphs with $n=100,000$ vertices, and used an error rate of $\theta=10\%$ unless otherwise noted. 
For each point, the empirical success rates averaged over $10$ Monte Carlo runs are reported.

\begin{enumerate}
 \item[(a)]
\emph{Regular rings.} We ran Algorithm \ref{alg:Algorithm-progressive} on rings $\mathcal{R}_{r}$ for various values of locality radius $r$ (Fig.~\ref{fig:plots}(a)), with the runtime reported in Table \ref{tab:runtime};
 
 \item[(b)] 
\emph{Small-world graphs.} We ran Algorithm \ref{alg:Algorithm-progressive} on small-world graphs, where the aggregate sampling rate for $\mathcal{R}_{r}$ is chosen to be on the same order as that for the complete graph (Fig.~\ref{fig:plots}(b));

 \item[(c)] 
\emph{Rings with nonuniform sampling weight.} 
We ran Algorithm \ref{alg:Algorithm-progressive} for rings with  nonuniform sampling rate (Fig.~\ref{fig:plots}(c)). Specifically, the distance between two vertices involved in a parity sample is drawn according to $\mathsf{Poisson}(r/2)$;

\item[(d)] \emph{Rings with different error rates.} We varied the error rate $\theta$ for rings with $r=18=n^{0.25}$, and plotted the empirical success rate (Fig.~\ref{fig:plots}(d)).

\end{enumerate}

 \begin{table}
 \begin{center}
 {\centering
 \begin{tabular}{c|c|c|c|c}
 \hline
  &$r=n^{0.2}$ 
  & $r=n^{0.25}$ & $r=n^{0.5}$ & $r=n^{0.75}$\\
  \hline
 Time (seconds/run) 
  &$3.47$ 
  & $3.55$ & $6.45$ & $58.4$\\
 \hline
 \end{tabular}
 }
\caption{The time taken to run {\SpectralExpand} on a MacBook Pro equipped with a 2.9 GHz Intel Core i5 and 8GB of memory over rings $\mathcal{R}_r$, where $n=100,000$, $\theta=10\%$ and $m = 1.5m^*$. 
All experiments converge to the truth within 2 iterations.
}
 \label{tab:runtime}
 \end{center}
 \end{table}

\begin{table}[h]

\centering%
\begin{tabular}{c|c|c|c|c|c|c|c}
\hline 
Chromosome number & 1 & 2 & 3 & 4 & 5 & 6 & 7\tabularnewline
\hline 
$n$ & 176477 & 191829 & 163492 & 168206 & 156352 & 152397 & 134685\tabularnewline
\hline 
$m$ & 547154 & 574189 & 504565 & 490476 & 475430 & 467606 & 414557\tabularnewline
\hline 
$L\cdot m$ & 3618604 & 3858642 & 3360817 & 3303015 & 3144785 & 3117837 & 2712355\tabularnewline
\hline 
Switch error rate & 3.45 & 3.26 & 2.91 & 2.93 & 2.97 & 2.77 & 3.53\tabularnewline
\hline 
\end{tabular}

\vspace{0.8em}

\begin{tabular}{c|c|c|c|c|c|c|c}
\hline 
Chromosome number & 8 & 9 & 10 & 11 & 12 & 13 & 14\tabularnewline
\hline 
$n$ & 128970 & 101243 & 121986 & 114964 & 112500 & 86643 & 76333\tabularnewline
\hline 
$m$ & 394285 & 318274 & 379992 & 357550 & 362396 & 269773 & 247339\tabularnewline
\hline 
$L\cdot m$ & 2601225 & 2042234 & 2476832 & 2313387 & 2323349 & 1719035 & 1562680\tabularnewline
\hline 
Switch error rate & 2.96 & 3.46 & 3.17 & 3.29 & 3.49 & 2.93 & 3.18\tabularnewline
\hline 
\end{tabular}

\vspace{0.8em}

\begin{tabular}{c|c|c|c|c|c|c|c|c}
\hline 
Chromosome number & 15 & 16 & 17 & 18 & 19 & 20 & 21 & 22\tabularnewline
\hline 
$n$ & 65256 & 73712 & 59788 & 68720 & 56933 & 53473 & 34738 & 34240\tabularnewline
\hline 
$m$ & 210779 & 231418 & 194980 & 214171 & 183312 & 171949 & 105313 & 102633\tabularnewline
\hline 
$L\cdot m$ & 1322975 & 1486718 & 1216416 & 1345967 & 1168408 & 1071688 & 660016 & 640396\tabularnewline
\hline 
Switch error rate & 3.49 & 3.73 & 4.57 & 3.01 & 4.93 & 3.70 & 3.75 & 5.07\tabularnewline
\hline 
\end{tabular}

\caption{The parameters and performance of {\SpectralStitch} when run on the NA$12878$ data-set from 10x-genomics \cite{loupe}.}
\label{tab:Loupe}
\end{table}

\begin{figure}
\centering
\includegraphics[width=0.47\textwidth]{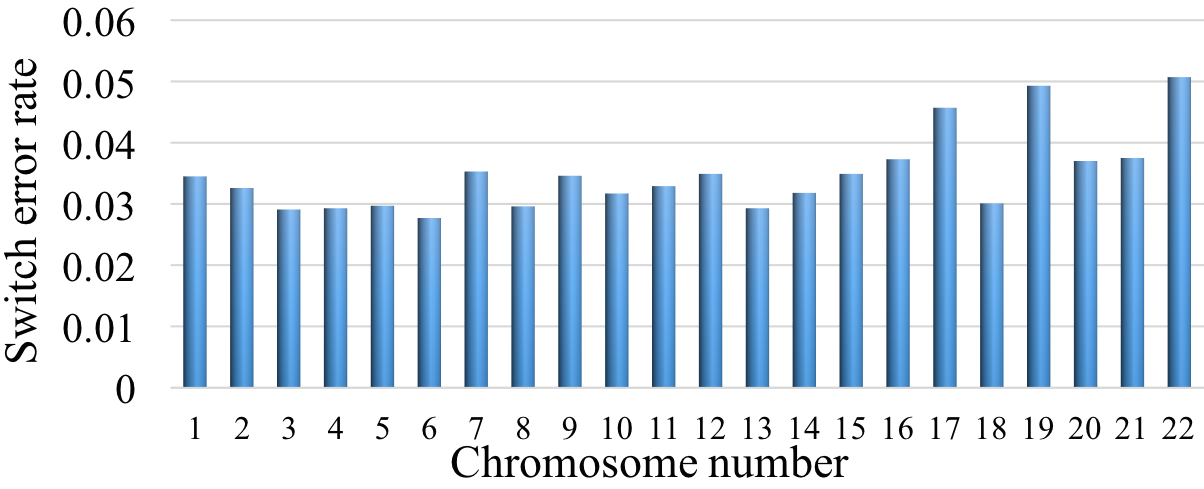}
\caption{The switch error rates of {\SpectralStitch} when run on the NA$12878$ data-set from 10x-genomics \cite{loupe}.}
 \label{fig:NA12878}
\end{figure}


We have also simulated a model of the haplotype phasing problem by assuming that the genome has a SNP periodically every $1000$ base pairs. The insert length distribution, i.e.~the distribution of the genomic distance between the linking reads, is given in  Fig.~\ref{fig:plots}(c)
for Illumina reads, and a draw from $\mathsf{Poisson}(3.5)$ truncated within the interval ${1,\cdots,9}$ 
is a reasonable approximation for the number of SNPs between two measured SNPs. We 
then ran the simulation on the rings $\mathcal{R}_{9}$, with non-uniform 
sampling weight. Using the nominal error rate of $p=1\%$ for the short reads, the error
rates of the measurements is  $2p(1-p) \approx 2\%$.
The empirical performance is shown in Fig.~\ref{fig:plots}(d).


Additionally, we have simulated reads  generated by 10x-Genomics \cite{10xGenomics} , which corresponds to the model in Section \ref{sec:Extension}.
Each measurement consists of multiple linked reads, which is generated by first randomly picking a segment of length $100$ SNPs (called a {\em fragment}) on the line graph and then generating 
$\mathsf{Poisson}(9)$ number of linked reads uniformly located in this segment. The noise rate per read is $p=0.01$.  The empirical result is shown in  Fig.~\ref{fig:plots}(e). The information theoretic limit is calculated using Theorem \ref{theorem:beyond-pairwise}, with $L$ set to infinity (since the number of vertices involved in a measurement is quite large here).

To evaluate the performance of our algorithm on real data, we ran {\SpectralStitch}  for Chromosomes 1-22
on the NA$12878$ data-set
made available by 10x-Genomics \cite{loupe}. The nominal error
rate per read is $p=1\%$, and the average number of SNPs touched by each sample is $L\in [6,7]$.
The number of SNPs $n$ ranges from $34240$ to $191829$, with the sample size $m$ from $102633$ to $574189$.  
Here, we split all vertices into overlapping subsets of size $W=100$.
The performance is measured in terms of the \emph{switch error rate}, that is, the fraction of positions where
we need to switch the estimate to match the ground truth. The performance on Chromosomes 1-22 is
reported in Tab.~\ref{tab:Loupe} and Fig.~\ref{fig:NA12878}.



\section{Discussion}
\label{sec:discussion}

We have presented two efficient algorithms that are information theoretically optimal. Rather than resorting to a ``global'' method that attempts recovery of all nodes all at once,  the proposed algorithms emphasize local subgraphs whose nodes are mutually well connected, and then propagate information across different parts of the graph.  This way we are able to exploit the locality structure in a computationally feasible manner. 

This paper leaves open numerous directions for further investigation. To begin with, the current work concentrates on the simplest setup where only two communities are present. It would be of interest to extend the results to the case with $M>2$ communities, which naturally arises in many applications including haplotype phasing for polyploid species \cite{das2015sdhap}. Furthermore, what would be the information and computation limits in the regime where the number $M$ of communities scales with $n$? In fact, there often exists a computational barrier away from the information limit when the measurement graph is drawn from the Erd\H{o}s-R\'{e}nyi model for large $M$ (e.g.~\cite{chen2014statistical}). How will this computational barrier be influenced by the locality structure of the measurement graph? In addition, the present theory operates under the assumption that $L$ is a fixed constant, namely, each multi-linked measurement entails only a small number of samples. Will the proposed algorithms still be optimal if $L$ is so large that it has to scale with $n$?

More broadly, it remains to develop a unified and systematic approach to accommodate a broader family of graphs beyond the instances considered herein. In particular, what would be an optimal recovery scheme if the graph is far from spatially-invariant or if there exist a few fragile cuts? 
Finally, as mentioned before, it would be interesting to see how to develop more general low-rank matrix completion paradigms, when the revealed entries come from a sampling pattern that exhibits locality structure.






\appendix

\section{Preliminaries\label{sec:Preliminary}}

Before continuing to the proofs, we gather a few facts that will be useful throughout.
First of all, recall that the maximum likelihood (ML) decision rule
achieves the lowest Bayesian probability of error, assuming uniform
prior over two hypotheses of interest. The resulting error exponent
is determined by the Chernoff information, as given in the following
lemma. 

\begin{lem}\label{lem:Chernoff-information}Fix any $\epsilon>0$.
Suppose we observe a collection of $N_{z}$ random variables $\bm{Z}=\left\{ Z_{1},\cdots,Z_{N_{z}}\right\} $
that are i.i.d.~given $N_{z}$. Consider two hypotheses $H_{0}$:
$Z_{i}\sim P_{0}$ and $H_{1}$: $Z_{i}\sim P_{1}$ for two given
probability measures $P_{0}$ and $P_{1}$. Assume that the Chernoff
information $D^{*}=D\left(P_{0},P_{1}\right)>0$ and the alphabet
of $Z_{i}$ are both finite and fixed, and that $\max_{z}\frac{P_{1}\left(z\right)}{P_{0}\left(z\right)}<\infty$. 

(a) Conditional on $N_{z}$, one has 
\begin{equation}
\exp\left\{ -\left(1+\epsilon\right)N_{z}D^{*}\right\} \leq P_{0}\left(\left.\frac{P_{1}\left(\bm{Z}\right)}{P_{0}\left(\bm{Z}\right)}\geq1\right|N_{z}\right)\leq\exp\left\{ -N_{z}D^{*}\right\} ,\label{eq:Chernoff-Bound-fix-N}
\end{equation}
where the lower bound holds when $N_{z}$ is sufficiently large. 

(b) If $N_{z}\sim\mathsf{Poisson}\left(N\right)$, then 
\begin{equation}
\exp\left\{ -\left(1+\epsilon\right)N\left(1-e^{-D^{*}}\right)\right\} \leq P_{0}\left(\frac{P_{1}\left(\bm{Z}\right)}{P_{0}\left(\bm{Z}\right)}\geq1\right)\leq\exp\left\{ -N\left(1-e^{-D^{*}}\right)\right\} ,\label{eq:Poisson-Chernoff}
\end{equation}
where the lower bound holds when $N$ is sufficiently large. \end{lem}

\begin{proof}See Appendix \ref{sec:Proof-of-Lemma-Chernoff}.\end{proof}

We emphasize that the best achievable error exponent coincides with
the Chernoff information $D^{*}$ when the sample size is fixed, while
it becomes $1-e^{-D^{*}}$---which is sometimes termed the Chernoff-Hellinger
divergence---when the sample size is Poisson distributed. 

The next result explores the robustness of the ML test. In particular,
we control the probability of error when the ML decision boundary
is slightly shifted, as stated below. 

\begin{lem}\label{lemma:Poisson}Consider any $\epsilon>0$, and
let $N\sim\mathsf{Poisson}\left(\lambda\right)$. 

(a) Fix any $0<\theta<0.5$. Conditional on $N$, draw $N$ independent
samples $Z_{1},\cdots,Z_{N}$ such that $Z_{i}\sim\mathsf{Bernoulli}\left(\theta\right)$,
$1\leq i\leq N$. Then one has 
\begin{eqnarray}
\mathbb{P}\left\{ \sum_{i=1}^{N}Z_{i}\geq\frac{1}{2}N-\epsilon\lambda\right\}  & \leq & \exp\left(\epsilon\cdot2\log\frac{1-\theta}{\theta}\lambda\right)\exp\left\{ -\lambda\left(1-e^{-\mathsf{KL}\left(0.5\|\theta\right)}\right)\right\} .\label{eq:Poisson}
\end{eqnarray}

(b) Let $P_{0}$ and $P_{1}$ be two distributions obeying $\max_{z}\frac{P_{1}\left(z\right)}{P_{0}\left(z\right)}<\infty$.
Conditional on $N$, draw $N$ independent samples $Z_{i}\sim P_{0}$,
$1\leq i\leq N$. Then one has 
\begin{eqnarray}
P_{0}\left\{ \sum_{j=1}^{N}\log\frac{P_{1}\left(Z_{i}\right)}{P_{0}\left(Z_{i}\right)}\geq-\epsilon\lambda\right\}  & \leq & \exp\left(\epsilon\lambda\right)\exp\left\{ -\lambda\left(1-e^{-D^{*}}\right)\right\} ,\label{eq:Poisson-LL}
\end{eqnarray}
where $D^{*}=D\left(P_{0},P_{1}\right)$ denotes the Chernoff information
between $P_{0}$ and $P_{1}$. \end{lem}\begin{proof}See Appendix
\ref{sec:Proof-of-Lemma-Poisson}.\end{proof}

Further, the following lemma develops an upper bound on the tail of
Poisson random variables.

\begin{lem}\label{lem:Poisson-LD}Suppose that $N\sim\mathsf{Poisson}\left(\epsilon\lambda\right)$
for some $0<\epsilon<1$. Then for any $c_{1}>2e$, one has 
\begin{eqnarray*}
\mathbb{P}\left\{ N\geq\frac{c_{1}\lambda}{\log\frac{1}{\epsilon}}\right\}  & \leq & 2\exp\left\{ -\frac{c_{1}\lambda}{2}\right\} .
\end{eqnarray*}
\end{lem}\begin{proof}See Appendix \ref{sec:Proof-of-Lemma-Poisson-LD}.\end{proof}

Additionally, our analysis relies on the well-known Chernoff-Hoeffding
inequality \cite[Theorem 1, Eqn (2.1)]{hoeffding1963probability}.

\begin{lem}[\textbf{Chernoff-Hoeffding Inequality}]\label{lem:Chernoff-Hoeffding}Suppose
$Z_{1},\cdots,Z_{n}$ are independent~Bernoulli random variables
with mean $\mathbb{E}\left[Z_{i}\right]\leq p$. Then for any $1>q\geq p$,
one has
\[
\mathbb{P}\left\{ \frac{1}{n}\sum\nolimits _{j=1}^{n}Z_{j}\geq q\right\} \leq\exp\left\{ -n\mathsf{KL}\left(q\hspace{0.2em}\|\hspace{0.2em}p\right)\right\} ,
\]
where $\mathsf{KL}\left(q\hspace{0.2em}\|\hspace{0.2em}p\right):=q\log\frac{q}{p}+(1-q)\log\frac{1-q}{1-p}$.
\end{lem}

We end this section with a lower bound on the KL divergence between
two Bernoulli distributions. 

\begin{fact}\label{fact:KL_LB}For any $0\leq q\leq\tau\leq1$, 
\begin{eqnarray*}
\mathsf{KL}\left(\tau\hspace{0.2em}\|\hspace{0.2em}q\right):=\tau\log\frac{\tau}{q}+\left(1-\tau\right)\log\frac{1-\tau}{1-q} & \geq & \tau\log\left(\tau/q\right)-\tau.
\end{eqnarray*}
\end{fact}\begin{proof}By definition,
\begin{eqnarray*}
\mathsf{KL}\left(\tau\hspace{0.2em}\|\hspace{0.2em}q\right) & \overset{(\text{i})}{\geq} & \tau\log\frac{\tau}{q}+\left(1-\tau\right)\log\left(1-\tau\right)\\
 & \overset{(\text{ii})}{\geq} & \tau\log\left(\tau/q\right)-\tau,
\end{eqnarray*}
where (i) follows since $\log\frac{1}{1-q}\geq0$, and (ii) arises
since $\left(1-\tau\right)\log\left(1-\tau\right)\geq-\left(1-\tau\right)\tau\geq-\tau$.
\end{proof}

\section{Performance Guarantees of \SpectralExpand \label{sec:Proof-spectral-expand}}

The analyses for all the cases follow almost identical arguments. In what follows, we separate the proofs into two parts:  (1) the optimality of {\SpectralExpand}  and {\SpectralStitch}, and (2) the minimax lower bound, where each part accommodates all models studied in this work. 

We start with the performance guarantee of {\SpectralExpand} in this section. 
Without loss of generality, we will assume $X_{1}=\cdots=X_{n}=0$
throughout this section. For simplicity of presentation, we will focus
on the most challenging boundary case where $m\asymp n\log n$, but
all arguments easily extend to the regime where $m\gg n\log n$.

\subsection{Stage 1 gives approximate recovery for $\mathcal{G}_{\mathrm{c}}$}

This subsection demonstrates that the spectral method (Algorithm \ref{alg:Algorithm-spectral})
is successful in recovering a portion $1-o\left(1\right)$ of the
variables in $\mathcal{V}_{\mathrm{c}}$ with high probability, as
stated in the following lemma. 

\begin{lem}\label{lem:Spectral}Fix $\theta>0$. Suppose that $\mathcal{G}=\left(\mathcal{V},\mathcal{E}\right)$
is a complete graph and the sample size $m\gtrsim n\log n$. The estimate
$\bm{X}^{(0)}=\left[X_{i}^{(0)}\right]_{1\leq i\leq n}$ returned
by Algorithm \ref{alg:Algorithm-spectral} obeys
\begin{equation}
\min\left\{ \|\bm{X}^{(0)}-\bm{X}\|_{0},\|\bm{X}^{(0)}+\bm{X}\|_{0}\right\} =o\left(n\right)\label{eq:Error-spectral}
\end{equation}
with probability exceeding $1-O\left(n^{-10}\right)$. \end{lem}

\begin{proof}See Appendix \ref{sec:Proof-of-Lemma-Spectral}.\end{proof}

\begin{remark}Here, $1-O\left(n^{-10}\right)$ can be replaced by
$1-O\left(n^{-c}\right)$ for any other positive constant $c>0$.\end{remark}

\begin{remark}It has been shown in \cite[Theorem 1.6]{chin2015stochastic}
that a truncated version of the spectral method returns reasonably
good estimates even in the sparse regime (i.e. $m\asymp n$). Note
that truncation is introduced in \cite[Theorem 1.6]{chin2015stochastic}
to cope with the situation in which some rows of the sample matrix are
``over-represented''. This becomes unnecessary in the regime where
$m\gtrsim n\log n$, since the number of samples incident to each
vertex concentrates around $\Theta\left(\log n\right)$, thus precluding
the existence of ``over-represented'' rows. \end{remark}

According to Lemma \ref{lem:Spectral}, Stage 1 accurately recovers
$(1-o\left(1\right))|\mathcal{V}_{\mathrm{c}}|$ variables in $\mathcal{V}_{\mathrm{c}}$
modulo some global phase, as long as 
\[
\lambda|\mathcal{V}_{\mathrm{c}}|^{2}\gtrsim|\mathcal{V}_{\mathrm{c}}|\cdot\log n.
\]
Since $\lambda\asymp m/n^{2}$ and $|\mathcal{V}_{\mathrm{c}}|\asymp d_{\mathrm{avg}}$,
this condition is equivalent to 
\[
m\gtrsim n\log n,
\]
which falls within our regime of interest. Throughout the rest of
the section, we will assume without loss of generality that 
\[
\frac{1}{|\mathcal{V}_{\mathrm{c}}|}\sum_{i=1}^{\mathcal{V}_{\mathrm{c}}}\bm{1}\left\{ X_{i}^{(0)}\neq X_{i}\right\} =o\left(1\right),
\]
i.e.~the first stage obtains approximate recovery along with the
correct global phase.

\subsection{Stage 2 yields approximate recovery for $\mathcal{V}\backslash\mathcal{V}_{\mathrm{c}}$}

For concreteness, we start by establishing the achievability for lines
and rings, which already contain all important ingredients for proving
the more general cases.

\subsubsection{Lines / rings}

We divide all vertices in $\mathcal{V}\backslash\mathcal{V}_{\mathrm{c}}$
into small groups $\left\{ \mathcal{V}_{i}\right\} $, each consisting
of $\epsilon\log^{3}n$ adjacent vertices\footnote{Note that the errors occurring to distinct vertices are statistically
dependent in the progressive estimation stage. The approach we propose
is to look at a group of vertices simultaneously, and to bound the
fraction of errors happening within this group. In order to exhibit
sufficiently sharp concentration, we pick the group size to be at
least $\epsilon\log^{3}n$. A smaller group is possible via more refined
arguments. }: 
\[
\mathcal{V}_{i}:=\left\{ |\mathcal{V}_{\mathrm{c}}|+\left(i-1\right)\epsilon\log^{3}n+1,\text{ }\cdots,\text{ }|\mathcal{V}_{\mathrm{c}}|+i\cdot\epsilon\log^{3}n\right\} ,
\]
where $\epsilon>0$ is some arbitrarily small constant. In what follows,
we will control the estimation errors happening within each group.
For notational simplicity, we let $\mathcal{V}_{0}:=\mathcal{V}_{\mathrm{c}}$.
An important vertex set for the progressive step, denoted by $\mathcal{V}_{\rightarrow i}$,
is the one encompassing all vertices preceding and connected to $\mathcal{V}_{i}$;
see Fig.~\ref{fig:illustration-V} for an illustration. 

The proof is recursive, which mainly consists in establishing the
claim below. To state the claim, we need to introduce a collection
of events as follows
\begin{eqnarray*}
\mathcal{A}_{0} & := & \left\{ \text{at most a fraction }\frac{\epsilon}{2}\text{ of progressive estimates }\left\{ X_{j}^{(0)}:\text{ }j\in\mathcal{V}_{\mathrm{c}}\right\} \text{ is incorrect}\right\} ;\\
\mathcal{A}_{i} & := & \left\{ \text{at most a fraction }\epsilon\text{ of progressive estimates }\left\{ X_{j}^{(0)}:\text{ }j\in\mathcal{V}_{i}\right\} \text{ is incorrect}\right\} ,\quad i\geq1.
\end{eqnarray*}

\begin{lem}\label{lem:claim-Ai}For any $i\geq0$, conditional on
$\mathcal{A}_{0}\cap\cdots\cap\mathcal{A}_{i}$, one has
\begin{equation}
\mathbb{P}\left\{ \mathcal{A}_{i+1}\mid\mathcal{A}_{0}\cap\cdots\cap\mathcal{A}_{i}\right\} \geq1-O\left(n^{-10}\right).\label{eq:lemma-Ai}
\end{equation}
As a result, one has 
\begin{eqnarray}
\mathbb{P}\left\{ \cap_{i\geq0}\mathcal{A}_{i}\right\}  & \geq & 1-O\left(n^{-9}\right).\label{eq:lemma-Ai-all}
\end{eqnarray}
\end{lem}

Apparently, $\mathcal{A}_{0}$ holds with high probability; see the
analysis for Stage 1. Thus, if (\ref{eq:lemma-Ai}) holds, then (\ref{eq:lemma-Ai-all})
follows immediately from the union bound. In fact, (\ref{eq:lemma-Ai-all})
suggests that for any group $\mathcal{V}_{i}$, only a small fraction
of estimates obtained in this stage would be incorrect, thus justifying
approximate recovery for this stage. Moreover, since the neighborhood
$\mathcal{N}\left(v\right)$ of each node $v\in\mathcal{V}_{i}$ is
covered by at most $O\left(\frac{d_{\mathrm{avg}}}{|\mathcal{V}_{i}|}\right)$
groups, the event $\cap_{i\geq0}\mathcal{A}_{i}$ immediately suggests
that there are no more than $O\left(\epsilon\cdot|\mathcal{V}_{i}|\right)O\left(\frac{d_{\mathrm{avg}}}{|\mathcal{V}_{i}|}\right)=O\left(\epsilon d_{\mathrm{avg}}\right)$
errors occurring to either the neighborhood $\mathcal{N}(v)$ or the
backward neighborhood $\mathcal{N}\left(v\right)\cap\mathcal{V}_{\rightarrow i}$.
This observation will prove useful for analyzing Stage 3, and hence
we summarize it in the following lemma. 

\begin{lem}\label{lem:error-dist-Stage2}There are at most $O\left(\epsilon d_{\mathrm{avg}}\right)$
errors occurring to either the neighborhood $\mathcal{N}(v)$ or the
backward neighborhood $\mathcal{N}\left(v\right)\cap\mathcal{V}_{\rightarrow i}$.
\end{lem}

The rest of the section is devoted to establish the claim (\ref{eq:lemma-Ai})
in Lemma \ref{lem:claim-Ai}. 

\begin{proof}[\textbf{Proof of Lemma \ref{lem:claim-Ai}}]As discussed
above, it suffices to prove (\ref{eq:lemma-Ai}) (which in turn justifies
(\ref{eq:lemma-Ai-all})). The following argument is conditional on
$\mathcal{A}_{0}\cap\cdots\cap\mathcal{A}_{i}$ and all estimates
for $\mathcal{V}_{0}\cup\cdots\cup\mathcal{V}_{i}$; we shall suppress
the notation by dropping this conditional dependence whenever it is
clear from the context.

\begin{figure}
\centering\includegraphics[width=0.9\textwidth]{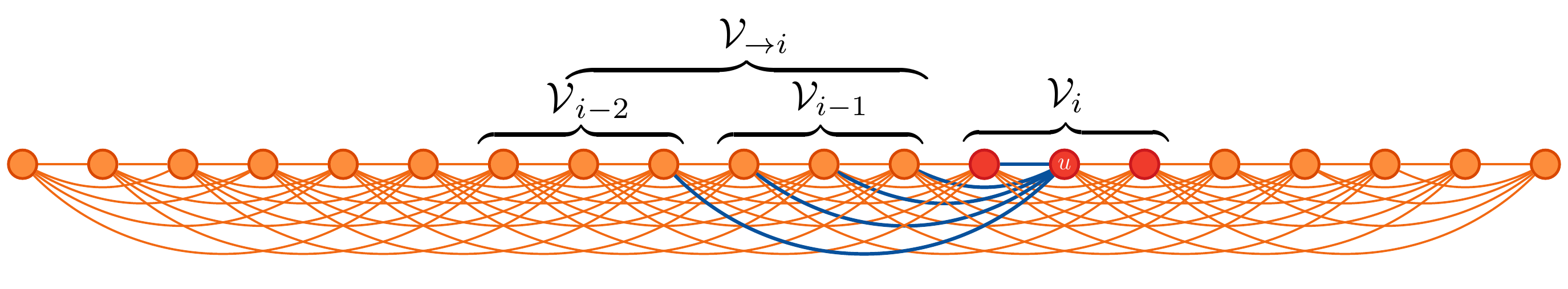}\\

(a)

\includegraphics[width=0.9\textwidth]{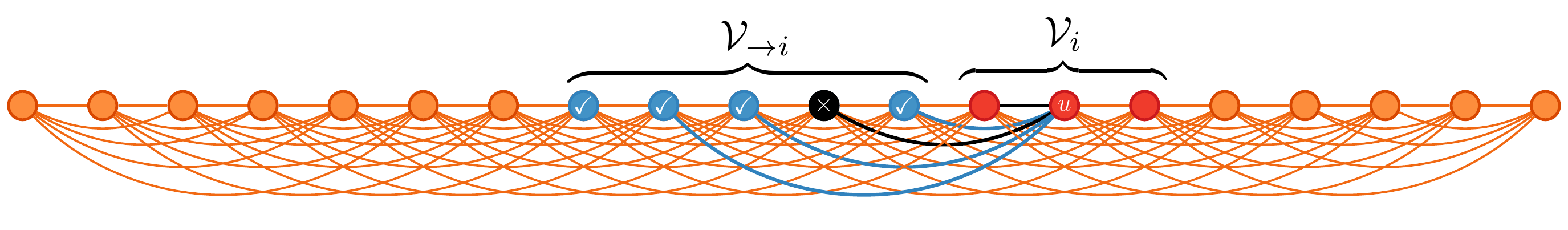}\\

(b)

\caption{\label{fig:illustration-V}(a) Illustration of $\mathcal{V}_{i}$,
$\mathcal{V}_{\rightarrow i}$, and $\mathcal{B}_{u}$, where $\mathcal{B}_{u}$
is the set of samples lying on the dark blue edges. (b) Illustration
of $\mathcal{B}_{u}^{\mathrm{good}}$ and $\mathcal{B}_{u}^{\mathrm{bad}}$,
which correspond to the set of samples lying on the set of blue edges
and black edges, respectively. Here, the symbol $\checkmark$ (resp.~$\times$)
indicates that the associated node has been estimated correctly (resp.~incorrectly).}
\end{figure}

Consider any vertex $u\in\mathcal{V}_{i+1}$. In the progressive estimation
step, each $X_{u}^{(0)}$ relies on the preceding estimates $\left\{ X_{j}^{(0)}\mid j:\text{ }j<u,\text{ }(j,u)\in\mathcal{E}\right\} $,
as well as the set $\mathcal{B}_{u}$ of backward samples incident
to $u$, that is,
\[
\mathcal{B}_{u}:=\left\{ Y_{u,j}^{(l)}\mid j<u,\text{ }(j,u)\in\mathcal{E},\text{ }1\leq l\leq N_{u,j}\right\} ;
\]
see Fig.~\ref{fig:illustration-V}(a). We divide $\mathcal{B}_{u}$
into two parts
\begin{itemize}
\item $\mathcal{B}_{u}^{\mathrm{good}}$: the set of samples $Y_{u,j}^{(l)}$
in $\mathcal{B}_{u}$ such that (i) $X_{j}^{(0)}=X_{j}$, and (ii)
$j\in\mathcal{V}_{\rightarrow\left(i+1\right)}$; 
\item $\mathcal{B}_{u}^{\mathrm{bad}}$: the remaining samples $\mathcal{B}_{u}\backslash\mathcal{B}_{u}^{\mathrm{good}}$;
\end{itemize}
and set 
\[
N_{u}^{\mathrm{good}}:=\left|\mathcal{B}_{u}^{\mathrm{good}}\right|\quad\text{and}\quad N_{u}^{\mathrm{bad}}:=\left|\mathcal{B}_{u}^{\mathrm{bad}}\right|.
\]
In words, $\mathcal{B}_{u}^{\mathrm{good}}$ is associated with those
preceding estimates in $\mathcal{V}_{\rightarrow\left(i+1\right)}$
that are consistent with the truth, while $\mathcal{B}_{u}^{\mathrm{bad}}$
entails the rest of the samples that might be unreliable. See Fig.~\ref{fig:illustration-V}(b)
for an illustration. The purpose of this partition is to separate
out $\mathcal{B}_{u}^{\mathrm{bad}}$, which only accounts for a small
fraction of all samples. 

We now proceed to analyze the majority voting procedure which, by
definition, succeeds if the total votes favoring the truth exceeds
$\frac{1}{2}\left(N_{u}^{\mathrm{good}}+N_{u}^{\mathrm{bad}}\right)$.
To preclude the effect of $\mathcal{B}_{u}^{\mathrm{bad}}$, we pay
particular attention to the part of votes obtained over $\mathcal{B}_{u}^{\mathrm{good}}$;
that is, the partial score
\[
\text{score}_{u}^{\mathrm{good}}:=\sum_{Y_{u,j}^{(l)}\in\mathcal{B}_{u}^{\mathrm{good}}}X_{j}^{(0)}\oplus Y_{u,j}^{(l)}.
\]
It is self-evident to check that the above success condition would
hold if 
\[
\text{score}_{u}^{\mathrm{good}}<\frac{1}{2}\left(N_{u}^{\mathrm{good}}+N_{u}^{\mathrm{bad}}\right)-\left|\mathcal{B}_{u}^{\mathrm{bad}}\right|=\frac{1}{2}N_{u}^{\mathrm{good}}-\frac{1}{2}N_{u}^{\mathrm{bad}},
\]
and we further define the complement event as 
\[
\mathcal{D}_{u}:=\left\{ \text{score}_{u}^{\mathrm{good}}\geq\frac{1}{2}N_{u}^{\mathrm{good}}-\frac{1}{2}N_{u}^{\mathrm{bad}}\right\} .
\]
The main point to work with $\mathcal{D}_{u}$ is that conditional
on all prior estimates in $\mathcal{V}_{\rightarrow\left(i+1\right)}$,
the $\mathcal{D}_{u}$'s are independent across all $u\in\mathcal{V}_{i+1}$. 

We claim that 
\begin{eqnarray}
 &  & \mathbb{P}\left\{ \mathcal{D}_{u}\right\} =\exp\left\{ -\Theta\left(\log n\right)\right\} :=P_{\mathrm{e},1}.\label{eq:defn-Pe1-1}
\end{eqnarray}
If this claim holds, then we can control the number of incorrect estimates
within the group $\mathcal{V}_{i+1}$ via the Chernoff-Hoeffding inequality.
Specifically,
\begin{align*}
 & \mathbb{P}\left\{ \frac{1}{\epsilon\log^{3}n}\sum_{u\in\mathcal{V}_{i+1}}\bm{1}\left\{ X_{u}^{(0)}\neq X_{u}\right\} \geq\frac{1}{\log^{2}n}\right\} \leq\mathbb{P}\left\{ \frac{1}{\epsilon\log^{3}n}\sum_{u\in\mathcal{V}_{i+1}}\bm{1}\left\{ \mathcal{D}_{u}\right\} \geq\frac{1}{\log^{2}n}\right\} \\
 & \quad\overset{(\text{a})}{\leq}\exp\left\{ -\epsilon\log^{3}n\cdot\mathsf{KL}\left(\left.\frac{1}{\log^{2}n}\hspace{0.2em}\right\Vert \hspace{0.2em}\mathbb{\mathbb{E}}\left[\bm{1}\left\{ \mathcal{D}_{u}\right\} \right]\right)\right\} \leq\exp\left\{ -\epsilon\log^{3}n\cdot\mathsf{KL}\left(\left.\frac{1}{\log^{2}n}\hspace{0.2em}\right\Vert \hspace{0.2em}P_{\mathrm{e},1}\right)\right\} \\
 & \quad\overset{(\text{b})}{\leq}\exp\left\{ -\epsilon\log^{3}n\cdot\frac{1}{\log^{2}n}\left(\log\frac{1}{P_{\mathrm{e},1}\log^{2}n}-1\right)\right\} \\
 & \quad\overset{(\text{c})}{=}\exp\left\{ -\Theta\left(\epsilon\log^{2}n\right)\right\} <O\left(\frac{1}{n^{10}}\right),
\end{align*}
where (a) follows from Lemma \ref{lem:Chernoff-Hoeffding}, (b) arises
from Fact \ref{fact:KL_LB}, and (c) is a consequence of (\ref{eq:defn-Pe1-1}).
This reveals that the fraction of incorrect estimates for $\mathcal{V}_{i}$
is vanishingly small with high probability, thus establishing the
claim (\ref{eq:lemma-Ai}). 

Finally, it remains to prove (\ref{eq:defn-Pe1-1}). To this end,
we decouple $\mathcal{D}_{u}$ into two events: 
\begin{equation}
\mathbb{P}\left\{ \mathcal{D}_{u}\right\} \leq\mathbb{P}\left\{ N_{u}^{\mathrm{bad}}\geq\frac{c_{0}\log n}{\log\frac{1}{\epsilon}}\right\} +\mathbb{P}\left\{ \text{score}_{u}^{\mathrm{good}}\geq\frac{1}{2}N_{u}^{\mathrm{true}}-\frac{1}{2}\frac{c_{0}\log n}{\log\frac{1}{\epsilon}}\right\} \label{eq:upper_bound_Bu}
\end{equation}
for some universal constant $c_{0}>0$. Recall that each edge is sampled
at a Poisson rate $\lambda\asymp\frac{m}{|\mathcal{E}_{0}|}\asymp\frac{n\log n}{nd_{\mathrm{avg}}}\asymp\frac{\log n}{d_{\mathrm{avg}}}$,
and that the average number of samples connecting $u$ and other nodes
in $\mathcal{V}_{i+1}$ is atmost $\lambda\cdot O\left(\epsilon d_{\mathrm{avg}}\right)$
(recalling our assumption that $r\gtrsim\log^{3}n$). On the event
$\mathcal{A}_{0}\cap\cdots\cap\mathcal{A}_{i}$, the number of wrong
labels in $\mathcal{V}_{\rightarrow\left(i+1\right)}$ is $O\left(\epsilon d_{u}\right)$,
and hence
\begin{eqnarray}
\mathbb{E}\left[N_{u}^{\mathrm{bad}}\right] & \leq & O\left(\lambda\epsilon d_{u}\right)\leq\epsilon c_{2}\log n\label{eq:Nbad-UB}
\end{eqnarray}
for some constant $c_{2}>0$. This further gives
\[
\mathbb{E}\left[N_{u}^{\mathrm{good}}\right]\geq\lambda c_{3}d_{u}-\mathbb{E}\left[N_{u}^{\mathrm{bad}}\right]\geq\left(1-c_{4}\epsilon\right)c_{3}\lambda d_{u}
\]
for some constants $c_{3},c_{4}>0$. Thus, Lemma \ref{lem:Poisson-LD}
and the inequality (\ref{eq:Nbad-UB}) taken collectively yield
\begin{eqnarray*}
\mathbb{P}\left\{ N_{u}^{\mathrm{bad}}\geq\frac{c_{1}c_{2}\log n}{\log\frac{1}{\epsilon}}\right\}  & \leq & 2\exp\left\{ -\frac{c_{1}c_{2}\log n}{2}\right\} 
\end{eqnarray*}
for any $c_{1}>2e$. In addition, in view of Lemma \ref{lemma:Poisson},
there exists some function $\tilde{\xi}\left(\cdot\right)$ such that
\begin{eqnarray*}
\mathbb{P}\left\{ \text{score}_{u}^{\mathrm{good}}\geq\frac{1}{2}N_{u}^{\mathrm{true}}-\frac{1}{2}\frac{c_{0}\log n}{\log\frac{1}{\epsilon}}\right\}  & \leq & \exp\left\{ -\left(1-o_{n}\left(1\right)\right)\left(1-\tilde{\xi}\left(\epsilon\right)\right)c_{3}\lambda d_{u}\left(1-e^{-D^{*}}\right)\right\} \\
 & = & \exp\left\{ -\Theta\left(\log n\right)\right\} ,
\end{eqnarray*}
where $\tilde{\xi}\left(\epsilon\right)$ is independent of $n$ and
vanishes as $\epsilon\rightarrow0$. Putting these bounds together
reveals that: when $\lambda\asymp\frac{\log n}{d_{\mathrm{avg}}}$
and $c_{0}=c_{1}c_{2}$, there exists some function $\hat{\xi}\left(\epsilon\right)$
independent of $n$ such that 
\begin{eqnarray}
(\ref{eq:upper_bound_Bu}) & \leq & 2\exp\left\{ -\frac{c_{1}c_{2}\log n}{2}\right\} +\exp\left\{ -\left(1-o_{n}\left(1\right)\right)\left(1-\tilde{\xi}\left(\epsilon\right)\right)\frac{\lambda d_{\mathrm{avg}}}{2}\left(1-e^{-D^{*}}\right)\right\} \nonumber \\
 & = & \exp\left\{ -\Theta\left(\log n\right)\right\} :=P_{\mathrm{e},1},\label{eq:defn-Pe1}
\end{eqnarray}
where $\hat{\xi}\left(\epsilon\right)$ vanishes as $\epsilon\rightarrow0$.
This finishes the proof.

\end{proof}

\subsubsection{Beyond lines / rings\label{sub:Beyond-lines-Stage2}}

The preceding analysis only relies on very few properties of lines
/ rings, and can be readily applied to many other graphs. In fact,
all arguments continue to hold as long as the following assumptions
are satisfied:
\begin{enumerate}
\item In $\mathcal{G}$, each vertex $v$ ($v>|\mathcal{V}_{\mathrm{c}}|$)
is connected with at least $\Theta(d_{\mathrm{avg}})$ vertices in
$\left\{ 1,\cdots,v-1\right\} $ by an edge; 
\item For any $v\in\mathcal{V}_{i}$ ($i\geq1$), its \emph{backward} neighborhood
$\mathcal{N}(v)\cap\mathcal{V}_{\rightarrow i}$ is covered by at
most $O\left(\frac{d_{\mathrm{avg}}}{|\mathcal{V}_{i}|}\right)=O\left(\frac{d_{\mathrm{avg}}}{\epsilon\log^{3}n}\right)$
distinct groups among $\mathcal{V}_{1},\cdots,\mathcal{V}_{i-1}$. 
\end{enumerate}
In short, the first condition ensures that the information of a diverse
range of prior vertices can be propagated to each $v$, whereas the
second condition guarantees that the estimation errors are fairly
spread out within the backward neighborhood associated with each node. 

We are now in position to look at grids, small-world graphs, as well
as lines / rings with nonuniform weights. 

\begin{itemize} \item[(a)] It is straightforward to verify that the
choices of $\mathcal{V}_{\mathrm{c}}$ and the ordering of $\mathcal{V}$
suggested in Section \ref{sec:theory-general} satisfy
Conditions 1-2, thus establishing approximate recovery for grids. 

\item[(b)] Suppose $\frac{\max_{(i,j)\in\mathcal{E}}w_{i,j}}{\min_{(i,j)\in\mathcal{E}}w_{i,j}}$
is bounded. Define the weighted degree as
\[
d_{v}^{w}:=\sum_{i:(i,v)\in\mathcal{E}}w_{i,v}
\]
and let the average weighted degree be $d_{\mathrm{avg}}^{w}:=\frac{1}{n}\sum d_{v}^{w}$.
Then all arguments continue to hold if $d_{v}$ and $d_{\mathrm{avg}}$
are replaced by $d_{v}^{w}$ and $d_{\mathrm{avg}}^{w}$, respectively.
This reveals approximate recovery for lines / rings / grids under
sampling with nonuniform weight.

\item[(c)] The proof for small-world graphs follows exactly the same
argument as for rings. 

\item[(d)] For the case with multi-linked samples, we redefine several
metrics as follows:
\begin{itemize}
\item $\mathcal{B}_{u}$: the set of backward samples $\left\{ Y_{e}^{(l)}\mid u\in e,\text{ }j<u\text{ for all other }j\in e,\text{ }1\leq l\leq N_{e}\right\} $,
where $e$ represents the hyper-edge;
\item $\mathcal{B}_{u}^{\mathrm{good}}$: the set of samples $Y_{e}^{(l)}$
in $\mathcal{B}_{u}$ such that (i) $X_{j}^{(0)}=X_{j}$ for all $j\in e$
and $j\neq u$, and (ii) $j\in\mathcal{V}_{\rightarrow\left(i+1\right)}$
for all $j\in e$ and $j\neq u$; 
\item $\mathcal{B}_{u}^{\mathrm{bad}}$: the remaining samples $\mathcal{B}_{u}\backslash\mathcal{B}_{u}^{\mathrm{good}}$.
\item We also need to re-define the score $\text{score}_{u}^{\mathrm{good}}$
as
\[
\text{score}_{u}^{\mathrm{good}}:=\sum_{Y_{e}^{(l)}\in\mathcal{B}_{u}^{\mathrm{good}}}\log\frac{\mathbb{P}\left\{ Y_{e}^{(l)}\mid X_{u}=1,\text{ }X_{i}=X_{i}^{(0)}\text{ }(i\in e,i\neq u)\right\} }{\mathbb{P}\left\{ Y_{e}^{(l)}\mid X_{u}=0,\text{ }X_{i}=X_{i}^{(0)}\text{ }(i\in e,i\neq u)\right\} }
\]
with the decision boundary replaced by $0$ and the event $\mathcal{D}_{u}$
replaced by
\[
\mathcal{D}_{u}:=\left\{ \text{score}_{u}^{\mathrm{good}}\geq0-s_{\max}\cdot|\mathcal{B}_{u}^{\mathrm{bad}}|\right\} =\left\{ \text{score}_{u}^{\mathrm{good}}\geq-s_{\max}N_{u}^{\mathrm{bad}}\right\} .
\]
 Here, $s_{\max}$ indicates the maximum possible likelihood ratio
for each $L$-wise sample:
\[
s_{\max}:=\max_{Y_{e},\left\{ Z_{i}\right\} }\left|\log\frac{\mathbb{P}\left\{ Y_{e}\mid X_{u}=1,\text{ }X_{i}=Z_{i}\text{ }(i\in e,i\neq u)\right\} }{\mathbb{P}\left\{ Y_{e}\mid X_{u}=0,\text{ }X_{i}=Z_{i}\text{ }(i\in e,i\neq u)\right\} }\right|.
\]
 
\end{itemize}
With these metrics in place, all proof arguments for the basic setup
carry over to the multi-linked sample case. 

\end{itemize}

\subsection{Stage 3 achieves exact recovery\label{sub:Stage-3-achievability}}

We now turn to the last stage, and the goal is to prove that $\bm{X}^{(t)}$
converges to $\bm{X}$ within $O\left(\log n\right)$ iterations.
Before proceeding, we introduce a few more notations that will be
used throughout. 
\begin{itemize}
\item For any vertex $v$, denote by $\mathcal{N}\left(v\right)$ the \emph{neighborhood}
of $v$ in $\mathcal{G}$, and let $\mathcal{S}\left(v\right)$ be
the set of samples that involve $v$;
\item For any vector $\bm{Z}=\left[Z_{1},\cdots,Z_{n}\right]^{\top}$ and
any set $\mathcal{I}\subseteq\left\{ 1,\cdots,n\right\} $, define
the $\ell_{0}$ norm restricted to $\mathcal{I}$ as follows 
\[
\|\bm{Z}\|_{0,\mathcal{I}}:=\sum\nolimits _{i\in\mathcal{I}}\bm{1}\left\{ Z_{i}\neq0\right\} .
\]

\item Generalize the definition of the majority vote operator such that
\[
\mathsf{majority}\left(\bm{Z}\right)=\left[\mathsf{majority}_{1}(Z_{1}),\cdots,\mathsf{majority}_{n}(Z_{n})\right]^{\top}
\]
 obtained by applying $\mathsf{majority}_{v}\left(\cdot\right)$ component-wise,
where 
\[
\mathsf{majority}_{v}\left(Z_{v}\right):=\begin{cases}
1,\quad & \text{if }Z_{v}\geq\frac{1}{2}|\mathcal{S}\left(v\right)|;\\
0, & \text{else}.
\end{cases}
\]

\item Let $\bm{V}_{\bm{Z}}$ (resp.~$\bm{V}_{\bm{X}}$) denote the local
voting scores using $\bm{Z}=\left[Z_{i}\right]_{1\leq i\leq n}$ (resp.~$\bm{X}=\left[X_{i}\right]_{1\leq i\leq n}={\bf 0}$)
as the current estimates, i.e.~for any $1\leq u\leq n$,
\begin{eqnarray}
\left(\bm{V}_{\bm{Z}}\right)_{u} & = & \sum_{Y_{i,u}^{(l)}\in\mathcal{S}\left(u\right)}Y_{i,u}^{(l)}\oplus Z_{i};\label{eq:defn-Vz}\\
\left(\bm{V}_{\bm{X}}\right)_{u} & = & \sum_{Y_{i,u}^{(l)}\in\mathcal{S}\left(u\right)}Y_{i,u}^{(l)}\oplus X_{i}=\sum_{y_{i,u}^{(l)}\in\mathcal{S}\left(u\right)}Y_{i,u}^{(l)}.\label{eq:defn-Vx}
\end{eqnarray}
With these notations in place, the iterative procedure can be succinctly
written as
\[
\bm{X}^{(t+1)}=\mathsf{majority}\left(\bm{V}_{\bm{X}^{(t)}}\right).
\]

\end{itemize}
The main subject of this section is to prove the following theorem. 

\begin{theorem}\label{theorem:contraction}Consider any $0<\epsilon\leq\epsilon_{0}$,
where $\epsilon_{0}$ is some sufficiently small constant. Define

\begin{equation}
\mathcal{Z}_{\epsilon}:=\left\{ \bm{Z}\in\left\{ 0,1\right\} ^{n}\mid\forall v:\text{ }\left\Vert \bm{Z}-\bm{X}\right\Vert _{0,\mathcal{N}\left(v\right)}\leq\epsilon d_{v}\right\} .\label{eq:error-spreading-condition}
\end{equation}
Then with probability approaching one, 
\[
\mathsf{majority}\left(\bm{V}_{\bm{Z}}\right)\in\mathcal{Z}_{\frac{1}{2}\epsilon},\qquad\forall\bm{Z}\in\mathcal{Z}_{\epsilon}\text{ and }\text{ }\forall\epsilon\in\left[\frac{1}{d_{\max}},\epsilon_{0}\right].
\]
\end{theorem}

\begin{remark}When the iterate falls within the set $\mathcal{Z}_{\epsilon}$
(cf.~(\ref{eq:error-spreading-condition})), there exist only a small
number of errors occurring to the neighborhood of each vertex. This
essentially implies that (i) the fraction of estimation errors is
low; (ii) the estimation errors are fairly spread out instead of clustering
within the neighborhoods of a few nodes. \end{remark}

\begin{remark}This is a uniform result: it holds regardless of whether
$\bm{Z}$ is statistically independent of the samples $\bm{Y}$ or
not. This differs from many prior results (e.g.~\cite{chaudhuri2012spectral})
that employ fresh samples in each stage in order to decouple the statistical
dependency. \end{remark}

Note that the subscript of $\mathcal{Z}_{\epsilon}$ indicates the
fraction of estimation errors allowed in an iterate. According to
the analyses for the the preceding stages, Stage 3 is seeded with
some initial guess $\bm{X}^{(0)}\in\mathcal{Z}_{\epsilon}$ for some
arbitrarily small constant $\epsilon>0$. This taken collectively
with Theorem \ref{theorem:contraction} gives rise to the following
error contraction result: for any $t\geq0$, 
\begin{equation}
\|\bm{X}{}^{(t+1)}-\bm{X}\|_{0,\mathcal{N}\left(v\right)}=\|\mathsf{majority}\left(\bm{V}_{\bm{X}{}^{(t)}}\right)-\bm{X}\|_{0,\mathcal{N}\left(v\right)}\leq\frac{1}{2}\|\bm{X}{}^{(t)}-\bm{X}\|_{0,\mathcal{N}\left(v\right)},\qquad1\leq v\leq n.\label{eq:geometric-convergence}
\end{equation}
This reveals the geometric convergence rate of $\bm{X}{}^{(t)}$,
namely, $\bm{X}^{(t)}$ converges to the truth within $O\left(\log n\right)$
iterations, as claimed.

The rest of this section is devoted to proving Theorem \ref{theorem:contraction}.
We will start by proving the result for any fixed candidate $\bm{Z}\in\mathcal{Z}_{\epsilon}$
independent of the samples, and then generalize to simultaneously
accommodate all $\bm{Z}\in\mathcal{Z}_{\epsilon}$. Our strategy is
to first quantify $\bm{V}_{\bm{X}}$ (which corresponds to the score
we obtain when only a single vertex is uncertain), and then control
the difference between $\bm{V}_{\bm{X}}$ and $\bm{V}_{\bm{Z}}$.
We make the observation that all entries of $\bm{V}_{\bm{X}}$ are
strictly below the decision boundary, as asserted by the following
lemma. 

\begin{lem}\label{lem:main-component}Fix any small constant $\delta>0$,
and suppose that $m\asymp n\log n$. Then one has
\[
\left(\bm{V}_{\bm{X}}\right)_{u}<\frac{1}{2}\left|\mathcal{S}\left(u\right)\right|-\delta\log n=\frac{1}{2}\left|\mathcal{S}\left(u\right)\right|-\delta\cdot O\left(\lambda d_{u}\right),\qquad1\leq u\leq n
\]
with probability exceeding $1-C_{1}\exp\left\{ -c_{1}\frac{m}{n}\left(1-e^{-D^{*}}\right)\right\} $
for some constants $C_{1},c_{1}>0$, provided that the following conditions
are satisfied: 

(1) Rings with $r\gtrsim\log^{2}n$: 
\[
m>\left(1+\overline{\xi}\left(\delta\right)\right)\frac{n\log n}{2\left(1-e^{-\mathsf{KL}\left(0.5\hspace{0.1em}\|\hspace{0.1em}\theta\right)}\right)};
\]

(2) Lines with $r=n^{\beta}$ for some constant $0<\beta<1$: 
\[
m>\left(1+\overline{\xi}\left(\delta\right)\right)\max\left\{ \beta,\frac{1}{2}\right\} \frac{n\log n}{1-e^{-\mathsf{KL}\left(0.5\hspace{0.1em}\|\hspace{0.1em}\theta\right)}};
\]

(3) Lines with $r=\gamma n$ for some constant $0<\gamma\leq1$: 
\[
m>\left(1+\overline{\xi}\left(\delta\right)\right)\left(1-\frac{1}{2}\gamma\right)\frac{n\log n}{1-e^{-\mathsf{KL}\left(0.5\hspace{0.1em}\|\hspace{0.1em}\theta\right)}};
\]

(4) Grids with $r=n^{\beta}$ for some constant $0<\beta<1/2$:
\[
m>\left(1+\overline{\xi}\left(\delta\right)\right)\max\left\{ 4\beta,\frac{1}{2}\right\} \frac{n\log n}{1-e^{-\mathsf{KL}\left(0.5\hspace{0.1em}\|\hspace{0.1em}\theta\right)}};
\]

(5) Small-world graphs:
\[
m>\left(1+\overline{\xi}\left(\delta\right)\right)\frac{n\log n}{2\left(1-e^{-\mathsf{KL}\left(0.5\hspace{0.1em}\|\hspace{0.1em}\theta\right)}\right)};
\]

In all these cases, $\overline{\xi}\left(\cdot\right)$ is some function
independent of $n$ satisfying $\overline{\xi}\left(\delta\right)\rightarrow0$
as $\delta\rightarrow0$. Here, we allow Cases (1), (2) and (4) to
have nonuniform sampling weight over different edges, as long as $\frac{\max_{(i,j)\in\mathcal{E}}w_{i,j}}{\min_{(i,j)\in\mathcal{E}}w_{i,j}}$
is bounded. \end{lem}\begin{proof}See Appendix \ref{sec:Proof-of-Lemma-main-component}.\end{proof}

It remains to control the difference between $\bm{V}_{\bm{X}}$ and
$\bm{V}_{\bm{Z}}$: 
\[
\bm{\Delta}_{\bm{Z}}:=\bm{V}_{\bm{Z}}-\bm{V}_{\bm{X}}.
\]
Specifically, we would like to demonstrate that \emph{most} entries
of $\bm{\Delta}_{\bm{Z}}$ are bounded in magnitude by $\delta\log n$
(or $\delta\cdot O\left(\lambda d_{u}\right)$), so that most of the
perturbations are absolutely controlled. To facilitate analysis, we
decouple the statistical dependency by writing 
\[
\bm{V}_{\bm{Z}}=\bm{F}_{\bm{Z}}+\bm{B}_{\bm{Z}},
\]
where $\bm{F}_{\bm{Z}}$ represents the votes using only forward samples,
namely,
\begin{eqnarray*}
\left(\bm{F}_{\bm{Z}}\right)_{u} & = & \sum_{i>u,\text{ }Y_{i,u}^{(l)}\in\mathcal{S}\left(u\right)}Y_{i,u}^{(l)}\oplus Z_{i},\qquad1\leq u\leq n.
\end{eqnarray*}
This is more convenient to work with since the entries of $\bm{F}_{\bm{Z}}$
(or $\bm{B}_{\bm{Z}}$) are jointly independent. In what follows,
we will focus on bounding $\bm{F}_{\bm{Z}}$, but all arguments immediately
apply to $\bm{B}_{\bm{Z}}$. To simplify presentation, we also decompose
$\bm{V}_{\bm{X}}$ into two parts $\bm{V}_{\bm{X}}=\bm{F}_{\bm{X}}+\bm{B}_{\bm{X}}$
in the same manner. 

Note that the $v$th entry of the difference
\begin{equation}
\bm{\Delta}^{\mathrm{F}}:=\bm{F}_{\bm{Z}}-\bm{F}_{\bm{X}}\label{eq:difference-forward}
\end{equation}
is generated by those entries from indices in $\mathcal{N}\left(v\right)$
satisfying $\bm{Z}_{v}\neq\bm{X}_{v}$. From the assumption (\ref{eq:error-spreading-condition}),
each $\bm{\Delta}_{v}^{\mathrm{F}}$ $(1\leq v\leq n)$ is dependent
on at most $O(\epsilon d_{v})$ non-zero entries of $\bm{Z}-\bm{X}$,
and hence on average each $\bm{\Delta}_{v}^{\mathrm{F}}$ is only
affected by $O\left(\lambda\cdot\epsilon d_{\mathrm{avg}}\right)$
samples. Moreover, each non-zero entry of $\bm{Z}-\bm{X}$ is bounded
in magnitude by a constant. This together with Lemma \ref{lem:Poisson-LD}
yields that: for any sufficiently large constant $c_{1}>0$, 
\begin{eqnarray}
\mathbb{P}\left\{ \left|\bm{\Delta}_{i}^{\mathrm{F}}\right|\geq\frac{c_{1}\lambda d_{\mathrm{avg}}}{\log\frac{1}{\epsilon}}\right\}  & \leq & 2\exp\left\{ -\Theta\left(c_{1}\lambda d_{\mathrm{avg}}\right)\right\} \leq2n^{-c_{2}},\label{eq:g_i_UB}
\end{eqnarray}
provided that $\lambda d_{\mathrm{avg}}\gtrsim\log n$ (which is the
regime of interest), where $c_{2}=\Theta\left(c_{1}\right)$ is some
absolute positive constant. In fact, for any index $i$, if $\left|\bm{\Delta}_{i}^{\mathrm{F}}\right|\geq\frac{c_{1}\lambda d_{\mathrm{avg}}}{\log\frac{1}{\epsilon}}$,
then picking sufficiently small $\epsilon>0$ we have
\[
\left|\bm{\Delta}_{i}^{\mathrm{F}}\right|\ll\lambda d_{\mathrm{avg}}\quad\text{of}\quad\left|\bm{\Delta}_{i}^{\mathrm{F}}\right|\ll\log n,
\]
and hence $\left(\bm{F}_{\bm{Z}}\right)_{i}$ and $\left(\bm{F}_{\bm{X}}\right)_{i}$
become sufficiently close. 

The preceding bound only concerns a single component. In order to
obtain overall control, we introduce a set of independent indicator
variables $\left\{ \eta_{i}\left(\bm{Z}\right)\right\} $: 
\[
\eta_{i}\left(\bm{Z}\right):=\begin{cases}
1,\quad & \text{if }\left|\bm{\Delta}_{i}^{\mathrm{F}}\right|\geq\frac{c_{1}\lambda d_{\mathrm{avg}}}{\log(1/\epsilon)},\\
0, & \text{else}.
\end{cases}
\]
For any $1\leq v\leq n$, applying Lemma \ref{lem:Chernoff-Hoeffding}
gives 
\begin{eqnarray*}
\mathbb{P}\left\{ \frac{1}{d_{v}}\sum_{i\in\mathcal{N}(v)}\eta_{i}\left(\bm{Z}\right)\geq\tau\right\}  & \leq & \exp\left\{ -d_{v}\mathsf{KL}\left(\tau\hspace{0.2em}\|\hspace{0.2em}\max_{i}\mathbb{E}\left[\eta_{i}\left(\bm{Z}\right)\right]\right)\right\} \\
 & \leq & \exp\left\{ -d_{v}\left(\tau\log\frac{\tau}{2n^{-c_{2}}}-\tau\right)\right\} ,
\end{eqnarray*}
where the last line follows from Fact \ref{fact:KL_LB} as well as
(\ref{eq:g_i_UB}). For any $\tau\geq1/n$, 
\[
\tau\log\frac{\tau}{2n^{-c_{2}}}-\tau\text{ }\gtrsim\text{ }\tau\log n,
\]
indicating that
\begin{eqnarray*}
\mathbb{P}\left\{ \frac{1}{d_{v}}\sum_{i\in\mathcal{N}(v)}\eta_{i}\left(\bm{Z}\right)\geq\tau\right\}  & \leq & \exp\left\{ -c_{3}\tau d_{\mathrm{avg}}\log n\right\} 
\end{eqnarray*}
for some universal constant $c_{3}>0$. If we pick $\epsilon>0$ and
$\tau>0$ to be sufficiently small, we see that with high probability
most of the entries have 
\[
\left|\bm{\Delta}_{i}^{\mathrm{F}}\right|<\frac{c_{1}\lambda d_{\mathrm{avg}}}{\log\frac{1}{\epsilon}}\ll\lambda d_{\mathrm{avg}}.
\]

We are now in position to derive the results in a more uniform fashion.
Suppose that $d_{\max}=Kd_{\mathrm{avg}}$. When restricted to $\mathcal{Z}_{\epsilon}$,
the neighborhood of each $v$ can take at most ${Kd_{\mathrm{avg}} \choose \epsilon Kd_{\mathrm{avg}}}2^{\epsilon Kd_{\mathrm{avg}}}$
different possible values. If we set $\tau=\frac{1}{4}\epsilon$,
then in view of the union bound,
\[
\begin{aligned} & \mathbb{P}\left\{ \exists\bm{Z}\in\mathcal{Z}_{\epsilon}\text{ s.t. }\text{ }\frac{1}{d_{v}}\sum_{i\in\mathcal{N}(v)}\eta_{i}\left(\bm{Z}\right)\geq\tau\right\} \leq{Kd_{\mathrm{avg}} \choose \epsilon Kd_{\mathrm{avg}}}2^{\epsilon Kd_{\mathrm{avg}}}\exp\left\{ -c_{3}\tau d_{\mathrm{avg}}\log n\right\} \\
 & \quad\leq\text{ }\left(2Kd_{\mathrm{avg}}\right)^{\epsilon Kd_{\mathrm{avg}}}\exp\left\{ -c_{3}\tau d_{\mathrm{avg}}\log n\right\} \\
 & \quad\leq\text{ }\exp\left\{ \left(1+o\left(1\right)\right)\epsilon Kd_{\mathrm{avg}}\log n\right\} \exp\left\{ -\frac{1}{4}c_{3}\epsilon d_{\mathrm{avg}}\log n\right\} \\
 & \quad\leq\text{ }\exp\left\{ -\left(\frac{1}{4}c_{3}-\left(1+o\left(1\right)\right)K\right)\epsilon d_{\mathrm{avg}}\log n\right\} .
\end{aligned}
\]

Since $\bm{Z},\bm{X}\in\left\{ 0,1\right\} ^{n}$, it suffices to
consider the case where $\epsilon\in\left\{ \frac{i}{d_{v}}\mid1\leq v\leq n,\text{ }1\leq i\leq d_{v}\right\} $,
which has at most $O\left(n^{2}\right)$ distinct values. Set $c_{3}$
to be sufficiently large and apply the union bound (over both $v$
and $\epsilon$) to deduce that: with probability exceeding $1-\exp\left(-\Theta\left(\epsilon d_{\mathrm{avg}}\log n\right)\right)\geq1-O\left(n^{-10}\right)$,
\begin{equation}
\text{card}\left\{ i\in\mathcal{N}\left(v\right):\text{ }\left|\bm{\Delta}_{i}^{\mathrm{F}}\right|\geq\frac{c_{1}\lambda d_{\mathrm{avg}}}{\log\frac{1}{\epsilon}}\right\} \leq\frac{1}{4}\epsilon d_{v},\qquad1\leq v\leq n,\label{eq:UniformUB-g}
\end{equation}
holds simultaneously for all $\bm{Z}\in\mathcal{Z}_{\epsilon}$ and
all $\epsilon\geq\frac{1}{d_{\mathrm{max}}}\asymp\frac{1}{d_{\mathrm{avg}}}$. 

The uniform bound (\ref{eq:UniformUB-g}) continues to hold if $\bm{\Delta}^{\mathrm{F}}$
is replaced by $\bm{\Delta}^{\mathrm{B}}$. Putting these together
suggests that with probability exceeding $1-\exp\left(-\Theta\left(\epsilon d\log n\right)\right)$,
\begin{align*}
 & \text{card}\left\{ i\in\mathcal{N}\left(v\right):\text{ }\left|\left(\bm{\Delta}_{\bm{Z}}\right)_{i}\right|\geq\frac{2c_{1}\lambda d_{\mathrm{avg}}}{\log\frac{1}{\epsilon}}\right\} \\
 & \quad\leq\text{card}\left\{ i\in\mathcal{N}\left(v\right):\text{ }\left|\left(\bm{\Delta}^{\mathrm{F}}\right)_{i}\right|\geq\frac{c_{1}\lambda d_{\mathrm{avg}}}{\log\frac{1}{\epsilon}}\right\} +\text{card}\left\{ i\in\mathcal{N}\left(v\right):\text{ }\left|\left(\bm{\Delta}^{\mathrm{B}}\right)_{i}\right|\geq\frac{c_{1}\lambda d_{\mathrm{avg}}}{\log\frac{1}{\epsilon}}\right\} \\
 & \quad\leq\frac{1}{2}\epsilon d_{v},\qquad1\leq v\leq n
\end{align*}
holds simultaneously for all $\bm{Z}\in\mathcal{Z}_{\epsilon}$ and
all $\epsilon\geq\frac{1}{d_{\mathrm{max}}}$. 

Taking $\delta$ to be $2c_{1}/\log\frac{1}{\epsilon}$ in (\ref{eq:signalUB}),
we see that all but $\frac{1}{2}\epsilon d_{v}$ entries of $\bm{V}_{\bm{Z}}=\bm{V}_{\bm{X}}+\bm{\Delta}_{\bm{Z}}$
at indices from $\mathcal{N}\left(v\right)$ exceed the voting boundary.
Consequently, the majority voting yields 
\[
\left\Vert \mathsf{majority}\left(\bm{V}_{\bm{Z}}\right)-\bm{X}\right\Vert _{0,\mathcal{N}\left(v\right)}\leq\frac{1}{2}\epsilon d_{v},\qquad1\leq v\leq n
\]
or, equivalently, 
\[
\mathsf{majority}\left(\bm{V}_{\bm{Z}}\right)\in\mathcal{Z}_{\frac{1}{2}\epsilon},\qquad\forall\bm{Z}\in\mathcal{Z}_{\epsilon}
\]
as claimed. 

When it comes to the multi-linked reads, we need to make some modification
to the vectors defined above. Specifically, we define the score vector
$\bm{V}_{\bm{Z}}$ and $\bm{V}_{\bm{X}}$ to be
\begin{eqnarray}
\left(\bm{V}_{\bm{Z}}\right)_{u} & = & \sum_{Y_{e}^{(l)}\in\mathcal{S}\left(u\right)}\log\frac{\mathbb{P}\left\{ Y_{e}^{(l)}\mid X_{u}=1,X_{i}=Z_{i}\text{ }(\text{for all }i\neq u\text{ and }u\in e)\right\} }{\mathbb{P}\left\{ Y_{e}^{(l)}\mid X_{u}=0,X_{i}=Z_{i}\text{ }(\text{for all }i\neq u\text{ and }u\in e)\right\} },\\
\left(\bm{V}_{\bm{X}}\right)_{u} & = & \sum_{Y_{e}^{(l)}\in\mathcal{S}\left(u\right)}\log\frac{\mathbb{P}\left\{ Y_{e}^{(l)}\mid X_{u}=1,X_{i}=0\text{ }(\text{for all }i\neq u\text{ and }u\in e)\right\} }{\mathbb{P}\left\{ Y_{e}^{(l)}\mid X_{u}=0,X_{i}=0\text{ }(\text{for all }i\neq u\text{ and }u\in e)\right\} },\label{eq:defn-Vx-1}
\end{eqnarray}
and replace the majority voting procedure as
\[
\mathsf{majority}_{v}\left(Z_{v}\right):=\begin{cases}
1,\quad & \text{if }Z_{v}\geq0;\\
0, & \text{else}.
\end{cases}
\]
With these changes in place, the preceding proof extends to the multi-linked
sample case with little modification, as long as $L$ remains a constant.
We omit the details for conciseness.

\section{Performance Guarantees of {\SpectralStitch} \label{sec:Proof-spectral-stitch}}

We start from the estimates $\left\{ X_{j}^{\mathcal{V}_{l}}:j\in\mathcal{V}_{l}\right\} $
obtained in Stage 1. Combining Lemma \ref{lem:Spectral} and the union
bound, we get
\[
\frac{1}{|\mathcal{V}_{l}|}\min\left\{ \sum_{j\in\mathcal{V}_{l}}\bm{1}\left\{ X_{j}^{\mathcal{V}_{l}}\neq X_{j}\right\} ,\sum_{j\in\mathcal{V}_{l}}\bm{1}\left\{ X_{j}^{\mathcal{V}_{l}}\oplus1\neq X_{j}\right\} \right\} =o\left(1\right),\quad l=1,2,\cdots
\]
with probability exceeding $1-O\left(n^{-c}\right)$ for any constant
$c>0$. In other words, we achieve approximate recovery---up to some
global phase---for each vertex group $\mathcal{V}_{l}$. The goal
of Stage 2 is then to calibrate these estimates so as to make sure
all groups enjoy the same global phase. Since each group suffers from
a fraction $o(1)$ of errors and any two adjacent groups share $O(|\mathcal{V}_{l}|)$
vertices, we can easily see that two groups of estimates $\left\{ X_{j}^{\mathcal{V}_{l}}:j\in\mathcal{V}_{l}\right\} $
and $\left\{ X_{j}^{\mathcal{V}_{l-1}}:j\in\mathcal{V}_{l-1}\right\} $
have positive correlation, namely, 
\[
\sum_{j\in\mathcal{V}_{l}\cap\mathcal{V}_{l-1}}X_{j}^{\mathcal{V}_{l}}\oplus X_{j}^{\mathcal{V}_{l-1}}\leq\frac{1}{2}\left|\mathcal{V}_{l}\cap\mathcal{V}_{l-1}\right|,
\]
 only when they share the same global phase. As a result, there are
at most $o(n)$ occurring to the estimates $\left\{ X_{i}^{(0)}\mid1\leq i\leq n\right\} $
obtained in Stage 2. Moreover, the way we choose $\mathcal{V}_{l}$
ensures that the neighborhood $\mathcal{N}_{v}$ of each vertex $v$
is contained within at most $O\left(\frac{d_{\mathrm{avg}}}{|\mathcal{V}_{1}|}\right)$
groups, thus indicating that 
\[
\frac{1}{|\mathcal{N}_{v}|}\min\left\{ \sum_{j\in\mathcal{N}_{v}}\bm{1}\left\{ X_{j}^{(0)}\neq X_{j}\right\} ,\sum_{j\in\mathcal{N}_{v}}\bm{1}\left\{ X_{j}^{(0)}\oplus1\neq X_{j}\right\} \right\} =o\left(1\right),\quad v=1,\cdots,n;
\]
that is, the estimation errors are fairly spread out across the network.
Finally, {\SpectralExpand}  and {\SpectralStitch} employ exactly
the same local refinement stage, and hence the proof for Stage 3 in
{\SpectralExpand}  readily applies here. This concludes the proof.

\section{Minimax Lower Bound\label{sec:Fundamental-lower-bound}}

This section contains the proof for the converse parts of Theorems
\ref{theorem:rings}-\ref{theorem:beyond-pairwise}; that is, the
minimax probability of error $\inf_{\psi}P_{\mathrm{e}}\left(\psi\right)\rightarrow1$
unless $m\geq(1-\epsilon)m^{*}$ in all of these theorems.

\subsection{Pairwise samples with uniform weight\label{sub:Converse-Pairwise-samples-uniform}}

We begin with the simplest sampling model: pairwise measurements with
uniform sampling rate at each edge, which are the scenarios considered
in Theorems \ref{theorem:rings}-\ref{theorem:Lines-Grids}. The key
ingredient to establish the minimax lower bounds is to prove the following
lemma. 

\begin{lem}\label{lemma:lower-bound}Fix any constant $\epsilon>0$,
and suppose that $N_{i,j}\overset{\text{ind.}}{\sim}\mathsf{Poisson}\left(\lambda\right)$
for all $(i,j)\in\mathcal{E}$. Consider any vertex subset $\mathcal{U}\subseteq\mathcal{V}$
with $\left|\mathcal{U}\right|\geq n^{\epsilon}$, and denote by $\tilde{d}$
the maximum degree of the vertices lying within $\mathcal{U}$. If
\begin{equation}
\lambda\tilde{d}\leq\left(1-\epsilon\right)\frac{\log\left|\mathcal{U}\right|}{1-e^{-D^{*}}},\label{eq:assumption-converse}
\end{equation}
then the probability of error $\inf_{\psi}P_{\mathrm{e}}\left(\psi\right)\rightarrow1$
as $n\rightarrow\infty$. \end{lem}\begin{proof}See Appendix \ref{sec:Proof-of-Lemma-Lower-Bound}.\end{proof}

We are now in position to demonstrate how Lemma \ref{lemma:lower-bound}
leads to tight lower bounds. In what follows, we let $d_{\mathrm{avg}}$
denote the average vertex degree in $\mathcal{G}$. 
\begin{itemize}
\item \textbf{Rings}. When $\mathcal{G}=\left(\mathcal{V},\mathcal{E}\right)$
is a ring $\mathcal{R}_{r}$ with connectivity radius $r$, set $\mathcal{U}=\mathcal{V}=\left\{ 1,\cdots,n\right\} $
and fix any small constant $\epsilon>0$. It is self-evident that
$\tilde{d}=d_{\mathrm{avg}}$. Applying Lemma \ref{lemma:lower-bound}
leads to a necessary recovery condition 
\begin{equation}
\lambda d_{\mathrm{avg}}>\left(1-\epsilon\right)\frac{\log n}{1-e^{-D^{*}}}.\label{eq:deg-ring-condition}
\end{equation}
Since $m=\lambda|\mathcal{E}|=\frac{1}{2}\lambda nd_{\mathrm{avg}}$,
this condition (\ref{eq:deg-ring-condition}) is equivalent to
\[
m>\left(1-\epsilon\right)\cdot\frac{n\log n}{2\left(1-e^{-D^{*}}\right)}.
\]

\item \textbf{Lines with $r=n^{\beta}$ for some constant $0<\beta<1$.
}Take $\mathcal{U}=\left\{ 1,\cdots,\epsilon r\right\} $ for some
sufficiently small constant $0<\epsilon<\beta$, which obeys $|\mathcal{U}|=\epsilon n^{\beta}\geq n^{\epsilon}$
for large $n$ and $\tilde{d}=\left(1+O\left(\epsilon\right)\right)d_{\mathrm{avg}}/2$.
In view of Lemma \ref{lemma:lower-bound}, a necessary recovery condition
is 
\[
\lambda\tilde{d}>\left(1-\epsilon\right)\frac{\log\left|\mathcal{U}\right|}{1-e^{-D^{*}}},
\]
\[
\Longleftrightarrow\quad\frac{1}{2}\lambda d_{\mathrm{avg}}>\frac{1-\epsilon}{1+O\left(\epsilon\right)}\cdot\frac{\beta\log n+\log\epsilon}{1-e^{-D^{*}}}.
\]
In addition, if we pick $\mathcal{U}=\mathcal{V}$, then $\tilde{d}=d_{\mathrm{avg}}$.
Lemma \ref{lemma:lower-bound} leads to another necessary condition:
\[
\lambda d_{\mathrm{avg}}>\left(1-\epsilon\right)\cdot\frac{\log n}{1-e^{-D^{*}}}.
\]
Combining these conditions and recognizing that $\epsilon$ can be
arbitrarily small, we arrive at the following necessary recovery condition
\begin{equation}
\frac{1}{2}\lambda d_{\mathrm{avg}}>\left(1-\epsilon\right)\max\left\{ \beta,\frac{1}{2}\right\} \frac{n\log n}{1-e^{-D^{*}}}.\label{eq:lambda-r-line}
\end{equation}
When $\beta<1$, the edge cardinality obeys $|\mathcal{E}|=(1+o(1))nd_{\mathrm{avg}}/2$,
allowing us to rewrite (\ref{eq:lambda-r-line}) as
\[
m=\lambda|\mathcal{E}|>\frac{1-\epsilon}{1+o(1)}\max\left\{ \beta,\frac{1}{2}\right\} \frac{n\log n}{1-e^{-D^{*}}}.
\]

\item \textbf{Lines with $r=\gamma n$ for some constant $0<\gamma\leq1$.
}Take $\mathcal{U}=\left\{ 1,\cdots,\epsilon r\right\} $ for some
sufficiently small constant $\epsilon>0$, which obeys $|\mathcal{U}|=\epsilon\gamma n\geq n^{\epsilon}$
for large $n$ and $\tilde{d}=\left(1+O\left(\epsilon\right)\right)r$.
Lemma \ref{lemma:lower-bound} reveals the following necessary recovery
condition: 
\[
\lambda\tilde{d}>\left(1-\epsilon\right)\frac{\log\left|\mathcal{U}\right|}{1-e^{-D^{*}}}
\]
\begin{equation}
\Longleftrightarrow\quad\lambda r>\frac{1-\epsilon}{1+O\left(\epsilon\right)}\cdot\frac{\log n+\log\left(\epsilon\gamma\right)}{1-e^{-D^{*}}}.\label{eq:LB-ring-linear}
\end{equation}
On the other hand, the total number of edges in $\mathcal{G}$ is
given by
\[
|\mathcal{E}|=\frac{1+o\left(1\right)}{2}\left(n^{2}-\left(n-r\right)^{2}\right)=\left(1+o\left(1\right)\right)nr\left(1-\frac{1}{2}\frac{r}{n}\right)=\left(1+o\left(1\right)\right)nr\left(1-\frac{1}{2}\gamma\right).
\]
This taken collectively with (\ref{eq:LB-ring-linear}) establishes
the necessary condition
\begin{align}
m & =\lambda|\mathcal{E}|=\left(1+o\left(1\right)\right)\lambda nr\left(1-\frac{1}{2}\gamma\right)\label{eq:davg-linear-line}\\
 & >\left(1-O\left(\epsilon\right)\right)\left(1-\frac{1}{2}\gamma\right)\frac{n\log n}{1-e^{-D^{*}}},\nonumber 
\end{align}
which completes the proof for this case by recognizing that $\epsilon$
can be arbitrary.
\item \textbf{Grids with $r=n^{\beta}$ for some constant $0<\beta<1$.
}Consider a sub-square of edge length $\epsilon r$ lying in the bottom
left corner of the grid, and let $\mathcal{U}$ consist of all $\epsilon^{2}r^{2}$
vertices residing within the sub-square. This obeys $|\mathcal{U}|=\epsilon^{2}n^{2\beta}>n^{\epsilon}$
for large $n$ and small $\epsilon$, and we also have $\tilde{d}=\left(1+O\left(\epsilon^{2}\right)\right)d_{\mathrm{avg}}/4$.
According to Lemma \ref{lemma:lower-bound}, a necessary recovery
condition is 
\[
\lambda\tilde{d}>\left(1-\epsilon\right)\frac{\log\left|\mathcal{U}\right|}{1-e^{-D^{*}}}
\]
or, equivalently,
\[
\frac{1}{4}\lambda d_{\mathrm{avg}}>\frac{1-\epsilon}{1+O\left(\epsilon^{2}\right)}\cdot\frac{2\left(\beta\log n+\log\epsilon\right)}{1-e^{-D^{*}}}.
\]
In addition, by taking $\mathcal{U}=\mathcal{V}$ one has $\tilde{d}=d_{\mathrm{avg}}$;
applying Lemma \ref{lemma:lower-bound} requires 
\[
\lambda d_{\mathrm{avg}}>\left(1-\epsilon\right)\cdot\frac{\log n}{1-e^{-D^{*}}}
\]
for exact recovery. Putting these two conditions together we derive
\[
\lambda d_{\mathrm{avg}}>\left(1-O\left(\epsilon\right)\right)\max\left\{ 8\beta,1\right\} \frac{n\log n}{1-e^{-D^{*}}},
\]
which is equivalent to
\[
m=\lambda|\mathcal{E}|>\left(1-O\left(\epsilon\right)\right)\max\left\{ 4\beta,\frac{1}{2}\right\} \frac{n\log n}{1-e^{-D^{*}}}
\]
since $|\mathcal{E}|=(1+o\left(1\right))nd_{\mathrm{avg}}/2$.
\end{itemize}

\subsection{Pairwise samples with nonuniform weight}

The preceding analyses concerning the minimax lower bounds can be
readily extended to the sampling model with nonuniform weight, which
is the focus of Theorem \ref{theorem:nonuniform-weight}. To be precise,
defining the weighted degree of any node $v$ as
\begin{equation}
d_{v}^{w}:=\sum_{i:(i,v)\in\mathcal{E}}w_{i,v},\label{eq:weighted-degree}
\end{equation}
we can generalize Lemma \ref{lemma:lower-bound} as follows.

\begin{lem}\label{lemma:lower-bound-weighted} Suppose that $\frac{\max_{(i,j)\in\mathcal{E}}w_{i,j}}{\min_{(i,j)\in\mathcal{E}}w_{i,j}}$
is bounded. Then Lemma \ref{lemma:lower-bound} continues to the hold
for the sampling model with nonuniform weight, provided that $\tilde{d}$
is defined as the maximum weighted degree within $\mathcal{V}_{1}$
and that $N_{i,j}\overset{\text{ind.}}{\sim}\mathsf{Poisson}(\lambda w_{i,j})$
for all $(i,j)\in\mathcal{E}$. \end{lem}\begin{proof}See Appendix
\ref{sec:Proof-of-Lemma-Lower-Bound}.\end{proof}

This lemma allows us to accommodate the following scenarios, as studied
in Theorem \ref{theorem:nonuniform-weight}. 
\begin{itemize}
\item \textbf{Lines / rings / grids under nonuniform sampling.} In view
of Lemma \ref{lemma:lower-bound-weighted}, the preceding proof in
Section \ref{sub:Converse-Pairwise-samples-uniform} continues to
hold in the presence of nonuniform sampling weight, provided that
$d_{\mathrm{avg}}$ is replaced with the average weighted degree $\frac{1}{n}\sum_{v=1}^{n}d_{v}^{w}$. 
\item \textbf{Small-world graphs}. The proof for rings is applicable for
small-world graphs as well, as long as $d_{\mathrm{avg}}$ is replaced
by the average weighted degree. 
\end{itemize}

\subsection{Multi-linked samples}

Finally, the above results immediately extend to the case with multi-linked
samples. 

\begin{lem}\label{lemma:lower-bound-multi} Consider the model with
multi-linked samples introduced in the main text, and suppose that
$L$ and $\epsilon>0$ are both fixed constants. Let $\mathcal{U}\subseteq\mathcal{V}$
be any vertex subset obeying $\left|\mathcal{U}\right|\geq n^{\epsilon}$,
and denote by $\tilde{d}$ the maximum degree (defined with respect
to the hyper-edges) of the vertices within $\mathcal{U}$. If
\begin{equation}
\lambda\tilde{d}\leq\left(1-\epsilon\right)\frac{\log\left|\mathcal{U}\right|}{1-e^{-D^{*}}},\label{eq:assumption-converse-2}
\end{equation}
then the probability of error $\inf_{\psi}P_{\mathrm{e}}\left(\psi\right)\rightarrow1$
as $n\rightarrow\infty$. \end{lem}\begin{proof}See Appendix \ref{sec:Proof-of-Lemma-Lower-Bound}.\end{proof}

\begin{lem}\label{lemma:lower-bound-1}Fix any constant $\epsilon>0$,
and suppose that $N_{i,j}\overset{\text{ind.}}{\sim}\mathsf{Poisson}\left(\lambda\right)$
for all $(i,j)\in\mathcal{E}$. Consider any vertex subset $\mathcal{U}\subseteq\mathcal{V}$
with $\left|\mathcal{U}\right|\geq n^{\epsilon}$, and denote by $\tilde{d}$
the maximum degree of the vertices lying within $\mathcal{U}$. If
\begin{equation}
\lambda\tilde{d}\leq\left(1-\epsilon\right)\frac{\log\left|\mathcal{U}\right|}{1-e^{-D^{*}}},\label{eq:assumption-converse-1}
\end{equation}
then the probability of error $\inf_{\psi}P_{\mathrm{e}}\left(\psi\right)\rightarrow1$
as $n\rightarrow\infty$. \end{lem}

When specialized to rings, setting $\mathcal{U}=\left\{ 1,\cdots,n\right\} $
with $\tilde{d}=d_{\mathrm{avg}}$ gives rise to the necessary condition
\begin{equation}
\lambda d_{\mathrm{avg}}>\left(1-\epsilon\right)\frac{\log n}{1-e^{-D^{*}}},\label{eq:multilinked-necessity}
\end{equation}
where $d_{\mathrm{avg}}$ represents the average number of hyper-edge
degree. Since each hyper-edge covers $L$ vertices, accounting for
the over-count factor gives $m=\frac{1}{L}n\lambda d_{\mathrm{avg}}$,
allowing us to rewrite (\ref{eq:multilinked-necessity}) as 
\[
m>\left(1-\epsilon\right)\frac{n\log n}{L\left(1-e^{-D^{*}}\right)}.
\]
This establishes the converse bound in the presence of multi-linked
samples.

\section{Chernoff Information for Multi-linked Samples}
\label{sec:Chernoff-Information-multi-linked}

Suppose now that each vertex $v$ is involved in $N_{v}$ multi-linked
samples or, equivalently, $N_{v}\left(L-1\right)$ pairwise samples.
Careful readers will note that these parity samples are not independent.
The key step in dealing with such dependency is not to treat them
as $N_{v}\left(L-1\right)$ independent samples, but instead $N_{v}$
independent groups. Thus, it suffices to compute the Chernoff information
associated with each group, as detailed below. 

Without loss of generality, suppose only $X_{1}$ is uncertain and
$X_{2}=\cdots=X_{n}=0$. Consider a multi-linked sample that covers
$X_{1},\cdots,X_{L}$. According to our model, each $L$-wise sample
is an independent copy of (\ref{eq:L-wise-model}). Since we never observe
the global phase in any sample, a sufficient statistic for $Y_{e}$
is given by 
\[
\tilde{Y}_{e}=\left(Z_{1}\oplus Z_{2},Z_{1}\oplus Z_{3},\cdots,Z_{1}\oplus Z_{L}\right).
\]

By definition (\ref{eq:Chernorff-info}), the Chernoff information $D^{*}$
is the large-deviation exponent when distinguishing between the conditional
distributions of
\begin{equation}
\tilde{Y}_{e}\mid\left(X_{1},\cdots,X_{L}\right)=\left(0,\cdots,0\right)\quad\text{and}\quad\tilde{Y}_{e}\mid\left(X_{1},\cdots,X_{L}\right)=\left(1,\cdots,0\right),\label{eq:two-case}
\end{equation}
which we discuss as follows.
\begin{itemize}
\item When $X_{1}=\cdots=X_{L}=0$:

\begin{itemize}
\item if $Z_{1}=0$ (which occurs with probability $1-p$), then $\tilde{Y}_{e}\sim\mathsf{Binomial}\left(L-1,p\right)$;
\item if $Z_{1}=1$ (which occurs with probability $p$), then $\tilde{Y}_{e}\sim\mathsf{Binomial}\left(L-1,1-p\right)$;
\end{itemize}
\item When $X_{1}=1$ and $X_{2}=\cdots=X_{L}=0$:

\begin{itemize}
\item if $Z_{1}=0$ (which occurs with probability $p$), then $\tilde{Y}_{e}\sim\mathsf{Binomial}\left(L-1,p\right)$;
\item if $Z_{1}=1$ (which occurs with probability $1-p$), then $\tilde{Y}_{e}\sim\mathsf{Binomial}\left(L-1,1-p\right)$.
\end{itemize}
\end{itemize}
To summarize, one has 
\begin{eqnarray*}
\tilde{Y}_{e}\mid\left(X_{1},\cdots,X_{L}\right)=\left(0,0,\cdots,0\right) & \sim & \left(1-p\right)\mathsf{Binomial}\left(L-1,p\right)+p\mathsf{Binomial}\left(L-1,1-p\right):=P_{0};\\
\tilde{Y}_{e}\mid\left(X_{1},\cdots,X_{L}\right)=\left(1,0,\cdots,0\right) & \sim & p\mathsf{Binomial}\left(L-1,p\right)+\left(1-p\right)\mathsf{Binomial}\left(L-1,1-p\right):=P_{1}.
\end{eqnarray*}

To derive a closed-form expression, we note that a random variable
$W_{0}\sim P_{0}$ obeys 
\begin{eqnarray}
P_{0}\left(W_{0}=i\right) & = & \left(1-p\right){L-1 \choose i}p^{i}\left(1-p\right)^{L-i-1}+p{L-1 \choose i}\left(1-p\right)^{i}p^{L-i-1}\nonumber \\
 & = & {L-1 \choose i}\left\{ p^{i}\left(1-p\right)^{L-i}+\left(1-p\right)^{i}p^{L-i}\right\} .\label{eq:P0-pdf}
\end{eqnarray}
Similarly, if $W_{1}\sim P_{1}$, then 
\begin{eqnarray}
P_{1}\left(W_{1}=i\right) & = & {L-1 \choose i}\left\{ p^{i+1}\left(1-p\right)^{L-i-1}+\left(1-p\right)^{i+1}p^{L-i-1}\right\} .\label{eq:P1-pdf}
\end{eqnarray}
By symmetry (i.e.~$P_{0}\left(W_{0}=i\right)=P_{1}\left(W_{1}=L-1-i\right)$),
one can easily verify that (\ref{eq:Chernorff-info}) is attained
when $\tau=1/2$, giving
\begin{align}
 & D\left(P_{0},P_{1}\right)=-\log\left\{ \sum_{i=0}^{L-1}\sqrt{P_{0}\left(W_{0}=i\right)P_{1}\left(W_{1}=i\right)}\right\} \nonumber \\
 & =-\log\left\{ \sum_{i=0}^{L-1}{L-1 \choose i}\sqrt{\left\{ p^{i}\left(1-p\right)^{L-i}+\left(1-p\right)^{i}p^{L-i}\right\} \left\{ p^{i+1}\left(1-p\right)^{L-i-1}+\left(1-p\right)^{i+1}p^{L-i-1}\right\} }\right\} .\label{eq:D-closed-form}
\end{align}


\section{Proof of Auxiliary Lemmas}

\subsection{Proof of Lemma \ref{lemma:Chernoff-L-infity}\label{sec:Proof-of-Lemma-Chernoff-L-infty}}

For notational convenience, set
\begin{equation}
b_{i}:=\sqrt{\left\{ p^{i}\left(1-p\right)^{L-i}+\left(1-p\right)^{i}p^{L-i}\right\} \left\{ p^{i+1}\left(1-p\right)^{L-i-1}+\left(1-p\right)^{i+1}p^{L-i-1}\right\} }.\label{eq:defn-bi}
\end{equation}
For any $i<\frac{L}{2}-\log L$, one can verify that
\begin{align*}
p^{i}\left(1-p\right)^{L-i}+\left(1-p\right)^{i}p^{L-i} & =p^{i}\left(1-p\right)^{L-i}\left\{ 1+\left(\frac{p}{1-p}\right)^{L-2i}\right\} \\
 & =\left(1+o_{L}\left(1\right)\right)p^{i}\left(1-p\right)^{L-i}
\end{align*}
and
\begin{align*}
p^{i+1}\left(1-p\right)^{L-i-1}+\left(1-p\right)^{i+1}p^{L-i-1} & =p^{i+1}\left(1-p\right)^{L-i-1}\left\{ 1+\left(\frac{p}{1-p}\right)^{L-2i-2}\right\} \\
 & =\left(1+o_{L}\left(1\right)\right)p^{i+1}\left(1-p\right)^{L-i-1}.
\end{align*}
These identities suggest that
\begin{align*}
\sum_{i=0}^{L/2-\log L}{L-1 \choose i}b_{i} & =\left(1+o_{L}\left(1\right)\right)\sum_{i=0}^{L/2-\log L}{L-1 \choose i}\sqrt{\left\{ p^{i}\left(1-p\right)^{L-i}\right\} \left\{ p^{i+1}\left(1-p\right)^{L-i-1}\right\} }\\
 & =\left(1+o_{L}\left(1\right)\right)\sqrt{p\left(1-p\right)}\sum_{i=0}^{L/2-\log L}{L-1 \choose i}p^{i}\left(1-p\right)^{L-i-1}\\
 & =\left(1+o_{L}\left(1\right)\right)\sqrt{p\left(1-p\right)},
\end{align*}
where the last line makes use of the following fact. 

\begin{fact}\label{fact:middle}Fix any $0<p<1/2$. Then one has
\[
\sum_{i=0}^{L/2-\log L}{L-1 \choose i}p^{i}\left(1-p\right)^{L-i-1}=1-o_{L}\left(1\right).
\]
\end{fact}\begin{proof}To simplify writing, we concentrate on the
case where $L$ is even. From the binomial theorem, we see that
\begin{equation}
\sum_{i=0}^{L-1}{L-1 \choose i}p^{i}\left(1-p\right)^{L-i-1}=1.\label{eq:sum-bi}
\end{equation}
Hence, it suffices to control $\sum_{i=L/2-\log L+1}^{L-1}{L-1 \choose i}p^{i}\left(1-p\right)^{L-i-1}$.
To this end, we first make the observation that
\begin{align}
\sum_{i=L/2-\log L+1}^{L/2-\log L-2}{L-1 \choose i}p^{i}\left(1-p\right)^{L-i-1} & \leq\left(2\log L\right)\max_{i\geq\frac{L}{2}-\log L+1}{L-1 \choose i}\left(\frac{p}{1-p}\right)^{i}\left(1-p\right)^{L-1}\nonumber \\
 & \leq\left(2\log L\right)\cdot{L-1 \choose L/2}\left(\frac{p}{1-p}\right)^{\frac{L}{2}-\log L+1}\left(1-p\right)^{L-1}\nonumber \\
 & \overset{(\text{i})}{\leq}\left\{ \left(2\log L\right)\left(1-p\right)^{2\log L-3}\right\} \cdot2^{L}\left(p\left(1-p\right)\right)^{\frac{L}{2}-\log L+1}\nonumber \\
 & \overset{(\text{ii})}{\leq}o_{L}\left(1\right)\cdot\left[2\left(p\left(1-p\right)\right)^{\frac{1}{2}-\frac{\log L}{L}}\right]^{L}\nonumber \\
 & =o_{L}\left(1\right),\label{eq:middle-part-b}
\end{align}
where (i) comes from the inequalities ${L-1 \choose L/2}\leq2^{L-1}\leq2^{L}$,
and (ii) holds because $\log L\left(1-p\right)^{2\log L-3}=o_{L}\left(1\right)$.
The last identity is a consequence of the inequality 
\[
\sqrt{p\left(1-p\right)}<1/2\qquad(\forall p<1/2),
\]
as well as the fact that $\left(p\left(1-p\right)\right)^{-\frac{\log L}{L}}\rightarrow1$
($L\rightarrow\infty$) and hence 
\[
\left(p\left(1-p\right)\right)^{\frac{1}{2}-\frac{\log L}{L}}=\sqrt{p(1-p)}\left(p\left(1-p\right)\right)^{-\frac{\log L}{L}}<1/2.
\]

On the other hand, the remaining terms can be bounded as 
\begin{align*}
 & \sum_{i=\frac{L}{2}+\log L-1}^{L-1}{L-1 \choose i}p^{i}\left(1-p\right)^{L-i-1}=\sum_{i=\frac{L}{2}+\log L-1}^{L-1}{L-1 \choose L-i-1}p^{i}\left(1-p\right)^{L-i-1}\\
 & \quad=\sum_{i=0}^{\frac{L}{2}-\log L}{L-1 \choose i}p^{L-i-1}\left(1-p\right)^{i}=\sum_{i=0}^{\frac{L}{2}-\log L}{L-1 \choose i}p^{i}\left(1-p\right)^{L-i-1}\cdot\left(\frac{p}{1-p}\right)^{L-2i-1}\\
 & \quad=o_{L}\left(1\right)\cdot\sum_{i=0}^{\frac{L}{2}-\log L}{L-1 \choose i}p^{i}\left(1-p\right)^{L-i-1}.
\end{align*}
Putting the above results together yields 
\begin{align*}
1 & =\left(\sum_{i=0}^{L/2-\log L}+\sum_{i=\frac{L}{2}+\log L-1}^{L-1}+\sum_{i=L/2-\log L+1}^{L/2-\log L-2}\right){L-1 \choose i}p^{i}\left(1-p\right)^{L-i-1}\\
 & =\left(1+o_{L}\left(1\right)\right)\sum_{i=0}^{L/2-\log L}{L-1 \choose i}p^{i}\left(1-p\right)^{L-i-1}+o_{L}\left(1\right),
\end{align*}
which in turn gives
\begin{align*}
\sum_{i=0}^{L/2-\log L}{L-1 \choose i}\left\{ p^{i}\left(1-p\right)^{L-i-1}\right\}  & =1-o_{L}\left(1\right)
\end{align*}
as claimed.\end{proof}

Following the same arguments, we arrive at
\begin{align*}
\sum_{i=L/2+\log L}^{L-1}{L-1 \choose i}b_{i} & =\left(1+o_{L}\left(1\right)\right)\sqrt{p\left(1-p\right)}.
\end{align*}
Moreover, 
\begin{align*}
\sum_{i=L/2-\log L+1}^{L/2+\log L-1}{L-1 \choose i}b_{i} & \leq\sum_{i=L/2-\log L+1}^{L/2+\log L-1}{L-1 \choose i}\left\{ p^{i}\left(1-p\right)^{L-i}+\left(1-p\right)^{i}p^{L-i}\right\} \\
 & \qquad+\sum_{i=L/2-\log L+1}^{L/2+\log L-1}{L-1 \choose i}\left\{ p^{i+1}\left(1-p\right)^{L-i-1}+\left(1-p\right)^{i+1}p^{L-i-1}\right\} \\
 & =O\left\{ \sum_{i=L/2-\log L+1}^{L/2+\log L-1}{L-1 \choose i}\left\{ p^{i}\left(1-p\right)^{L-i-1}\right\} \right\} \\
 & =o_{L}\left(1\right),
\end{align*}
where the last line follows the same step as in the proof of Fact
\ref{fact:middle} (cf\@.~(\ref{eq:middle-part-b})). Taken together
these results lead to
\begin{align*}
 & \sum_{i=0}^{L-1}{L-1 \choose i}b_{i}=\left\{ \sum_{i=0}^{L/2-\log L-1}+\sum_{i=L/2+\log L-1}^{L-1}+\sum_{i=L/2-\log L}^{L/2+\log L}\right\} {L-1 \choose i}b_{i}\\
 & \quad=2\left(1+o_{L}\left(1\right)\right)\sqrt{p\left(1-p\right)},
\end{align*}
thus demonstrating that
\[
D\left(P_{0},P_{1}\right)=-\log\left\{ 2\left(1+o_{L}\left(1\right)\right)\sqrt{p\left(1-p\right)}\right\} =\left(1+o_{L}\left(1\right)\right)\mathsf{KL}\left(0.5\|p\right).
\]

\subsection{Proof of Lemma \ref{lem:Chernoff-information}\label{sec:Proof-of-Lemma-Chernoff}}

Let $M$ be the alphabet size for $Z_{i}$. The standard method of
types result (e.g. \cite[Chapter 2]{csiszar2004information} and \cite[Section 11.7-11.9]{cover2006elements})
reveals that
\begin{equation}
\frac{1}{\left(N_{z}+1\right)^{M}}\exp\left\{ -\left(1+\frac{\epsilon}{2}\right)N_{z}D^{*}\right\} \leq P_{0}\left(\left.\frac{P_{1}\left(\bm{Z}\right)}{P_{0}\left(\bm{Z}\right)}\geq1\right|N_{z}\right)\leq\exp\left\{ -N_{z}D^{*}\right\} ;\label{eq:LB-crude-Chernoff}
\end{equation}
here, the left-hand side holds for sufficiently large $N_{z}$, while
the right-hand side holds for arbitrary $N_{z}$ (see \cite[Exercise 2.12]{csiszar2011information}
or \cite[Theorem 1]{csiszar1984sanov} and recognize the convexity
of the set of types under consideration). Moreover, since $D^{*}>0$
and $M$ are fixed, one has $\frac{1}{\left(N_{z}+1\right)^{M}}=\exp\left(-M\log\left(N_{z}+1\right)\right)\geq\exp\left\{ -\frac{1}{2}\epsilon N_{z}D^{*}\right\} $
for any sufficiently large $N_{z}$, thus indicating that
\begin{equation}
P_{0}\left(\left.\frac{P_{1}\left(\bm{Z}\right)}{P_{0}\left(\bm{Z}\right)}\geq1\right|N_{z}\right)\geq\exp\left\{ -\left(1+\epsilon\right)N_{z}D^{*}\right\} \label{eq:lower-bound-Chernoff}
\end{equation}
as claimed. 

We now move on to the case where $N_{z}\sim\mathsf{Poisson}\left(N\right)$.
Employing (\ref{eq:lower-bound-Chernoff}) we arrive at 
\begin{eqnarray}
P_{0}\left(\frac{P_{1}\left(\bm{Z}\right)}{P_{0}\left(\bm{Z}\right)}\geq1\right) & = & \sum_{l=0}^{\infty}\mathbb{P}\left(N_{z}=l\right)P_{0}\left(\left.\frac{P_{1}\left(\bm{Z}\right)}{P_{0}\left(\bm{Z}\right)}\geq1\right|N_{z}=l\right)\label{eq:expansion}\\
 & \geq & \sum_{l=\tilde{N}}^{\infty}\frac{N^{l}e^{-N}}{l!}\exp\left\{ -\left(1+\epsilon\right)lD^{*}\right\} \label{eq:inequality}\\
 & = & e^{-(N-N_{0})}\sum_{l=\tilde{N}}^{\infty}\frac{N_{0}^{l}\exp\left(-N_{0}\right)}{l!}\label{eq:inequality2}
\end{eqnarray}
for any sufficiently large $\tilde{N}$, where we have introduced
$N_{0}:=Ne^{-(1+\epsilon)D^{*}}$. Furthermore, taking $\tilde{N}=\log N_{0}$
we obtain
\begin{eqnarray*}
\sum_{l=\tilde{N}}^{\infty}\frac{N_{0}^{l}}{l!}\exp\left(-N_{0}\right) & = & 1-\sum_{l=0}^{\tilde{N}}\frac{N_{0}^{l}}{l!}\exp\left(-N_{0}\right)\geq1-\sum_{l=0}^{\tilde{N}}N_{0}^{l}\exp\left(-N_{0}\right)\\
 & \geq & 1-\left(\tilde{N}+1\right)N_{0}^{\tilde{N}}\exp\left(-N_{0}\right)\\
 & = & 1-\left(\log N_{0}+1\right)N_{0}^{\log N_{0}}\exp\left(-N_{0}\right)=1-o_{N}\left(1\right)\\
 & \geq & 0.5
\end{eqnarray*}
as long as $N$ is sufficiently large. Substitution into (\ref{eq:inequality2})
yields
\begin{eqnarray}
P_{0}\left(\frac{P_{1}\left(\bm{Z}\right)}{P_{0}\left(\bm{Z}\right)}\geq1\right) & \geq0.5e^{-\left(N-N_{0}\right)} & \geq\exp\left(-\left(1+\epsilon\right)N\left(1-e^{-(1+\epsilon)D^{*}}\right)\right).\label{eq:inequality2-1}
\end{eqnarray}
This finishes the proof of the lower bound in \eqref{eq:Poisson-Chernoff}
since $\epsilon>0$ can be arbitrary. 

Additionally, applying the upper bound (\ref{eq:LB-crude-Chernoff})
we derive
\begin{align*}
\eqref{eq:expansion} & \leq\sum_{l=0}^{\infty}\frac{N^{l}e^{-N}}{l!}\cdot e^{-lD^{*}}=\exp\left(-N\left(1-e^{-D^{*}}\right)\right),
\end{align*}
establishing the upper bound \eqref{eq:Poisson-Chernoff}.

\subsection{Proof of Lemma \ref{lemma:Poisson}\label{sec:Proof-of-Lemma-Poisson}}

We start with the general case, and suppose that the Chernoff information
(\ref{eq:Chernorff-info}) is attained by $\tau=\tau^{*}\in\left[0,1\right]$.
It follows from the Chernoff bound that
\begin{eqnarray*}
P_{0}\left\{ \left.\sum_{i=1}^{N}\log\frac{P_{1}\left(Z_{i}\right)}{P_{0}\left(Z_{i}\right)}\geq-\epsilon\lambda\text{ }\right|N=k\right\}  & = & P_{0}\left\{ \tau^{*}\sum_{i=1}^{k}\log\frac{P_{1}\left(Z_{i}\right)}{P_{0}\left(Z_{i}\right)}\geq-\tau^{*}\cdot\epsilon\lambda\text{ }\right\} \leq\frac{\prod_{i=1}^{k}\mathbb{E}_{P_{0}}\left[\left(\frac{P_{1}\left(Z_{i}\right)}{P_{0}\left(Z_{i}\right)}\right)^{\tau^{*}}\right]}{\exp\left(-\tau^{*}\cdot\epsilon\lambda\right)}\\
 & = & \exp\left(\tau^{*}\cdot\epsilon\lambda\right)\left(\mathbb{E}_{P_{0}}\left[\left(\frac{P_{1}\left(Z_{i}\right)}{P_{0}\left(Z_{i}\right)}\right)^{\tau^{*}}\right]\right)^{k}\\
 & = & \exp\left(\tau^{*}\cdot\epsilon\lambda\right)\left(\sum\nolimits _{z}P_{0}^{1-\tau^{*}}\left(z\right)P_{1}^{1-\tau^{*}}\left(z\right)\right)^{k}\\
 & \leq & \exp\left(\epsilon\lambda\right)\exp\left(-kD^{*}\right).
\end{eqnarray*}
This suggests that
\begin{eqnarray*}
P_{0}\left\{ \sum_{i=1}^{N}\log\frac{P_{1}\left(Z_{i}\right)}{P_{0}\left(Z_{i}\right)}\geq-\epsilon\lambda\right\}  & = & P_{0}\left\{ \left.\sum_{i=1}^{N}\log\frac{P_{1}\left(Z_{i}\right)}{P_{0}\left(Z_{i}\right)}\geq-\epsilon\lambda\text{ }\right|N=k\right\} \mathbb{P}\left\{ N=k\right\} \\
 & \leq & \exp\left(\epsilon\lambda\right)\mathbb{E}_{N\sim\mathsf{Poisson}\left(\lambda\right)}\left[\exp\left(-ND^{*}\right)\right]\\
 & = & \exp\left(\epsilon\lambda\right)\exp\left\{ -\lambda\left(1-e^{-D^{*}}\right)\right\} ,
\end{eqnarray*}
where the last identity follows from the moment generating function
of Poisson random variables. This establishes the claim for the general
case. 

When specialized to the Bernoulli case, the log-likelihood ratio is
given by
\begin{align*}
\log\frac{P_{1}\left(Z_{i}\right)}{P_{0}\left(Z_{i}\right)} & =\mathbb{I}\left\{ Z_{i}=0\right\} \log\frac{\theta}{1-\theta}+\mathbb{I}\left\{ Z_{i}=1\right\} \log\frac{1-\theta}{\theta}\\
 & =\left\{ 2\mathbb{I}\left\{ Z_{i}=1\right\} -1\right\} \log\frac{1-\theta}{\theta}.
\end{align*}
When $0<\theta<1/2$, this demonstrates the equivalence between the
following two inequalities: 
\[
\sum_{i=1}^{N}\log\frac{P_{1}\left(Z_{i}\right)}{P_{0}\left(Z_{i}\right)}\geq-\epsilon\lambda\qquad\Longleftrightarrow\qquad\sum_{i=1}^{N}\mathbb{I}\left\{ Z_{i}=1\right\} \geq\frac{1}{2}N-\frac{\epsilon\lambda}{2\log\frac{1-\theta}{\theta}}.
\]
Recognizing that $\sum_{i=1}^{N}Z_{i}=\sum_{i=1}^{N}\mathbb{I}\left\{ Z_{i}=1\right\} $
and replacing $\epsilon$ with $\epsilon\cdot2\log\frac{1-\theta}{\theta}$,
we complete the proof.

\subsection{Proof of Lemma \ref{lem:Poisson-LD}\label{sec:Proof-of-Lemma-Poisson-LD}}

For any constant $c_{1}\geq2e$,
\begin{eqnarray*}
\mathbb{P}\left\{ N\geq\frac{c_{1}\lambda}{\log\frac{1}{\epsilon}}\right\}  & = & \sum_{k\geq\frac{c_{1}\lambda}{\log(1/\epsilon)}}\mathbb{P}\left\{ N=k\right\} \text{ }=\sum_{k\geq\frac{c_{1}\lambda}{\log(1/\epsilon)}}\frac{\left(\epsilon\lambda\right)^{k}}{k!}\exp\left(-\epsilon\lambda\right)\\
 & \overset{(\text{i})}{\leq} & \sum_{k\geq\frac{c_{1}\lambda}{\log(1/\epsilon)}}\frac{\left(\epsilon\lambda\right)^{k}}{\left(k/e\right)^{k}}\text{ }=\text{ }\sum_{k\geq\frac{c_{1}\lambda}{\log(1/\epsilon)}}\left(\frac{\epsilon e\lambda}{k}\right)^{k}\text{ }\leq\text{ }\sum_{k\geq\frac{c_{1}\lambda}{\log(1/\epsilon)}}\left(\frac{\epsilon e\lambda}{\frac{c_{1}\lambda}{\log(1/\epsilon)}}\right)^{k}\\
 & \overset{(\text{ii})}{\leq} & \sum_{k\geq\frac{c_{1}\lambda}{\log(1/\epsilon)}}\left(\frac{e\sqrt{\epsilon}}{c_{1}}\right)^{k}\text{ }\overset{(\text{iii})}{\leq}\text{ }2\left(\frac{e\sqrt{\epsilon}}{c_{1}}\right)^{\frac{c_{1}\lambda}{\log(1/\epsilon)}}\\
 & \leq & 2\exp\left\{ -\log\left(\frac{c_{1}}{e}\cdot\frac{1}{\sqrt{\epsilon}}\right)\frac{c_{1}\lambda}{\log\frac{1}{\epsilon}}\right\} \leq2\exp\left\{ -\log\left(\frac{1}{\sqrt{\epsilon}}\right)\frac{c_{1}\lambda}{\log\frac{1}{\epsilon}}\right\} \\
 & = & 2\exp\left\{ -\frac{c_{1}\lambda}{2}\right\} ,
\end{eqnarray*}
where (i) arises from the elementary inequality $a!\geq\left(\frac{a}{e}\right)^{a}$,
(ii) holds because $\epsilon\log\frac{1}{\epsilon}\leq\sqrt{\epsilon}$
holds for any $0<\epsilon\leq1$, and (iii) follows due to the inequality
$\sum_{k\geq K}a^{k}\leq\frac{a^{K}}{1-a}\leq2a^{K}$ as long as $0<a\leq1/2$.

\subsection{Proof of Lemmas \ref{lemma:lower-bound}-\ref{lemma:lower-bound-multi}\label{sec:Proof-of-Lemma-Lower-Bound}}

(1) We start by proving Lemma \ref{lemma:lower-bound}, which contains
all ingredients for proving Lemmas \ref{lemma:lower-bound-weighted}-\ref{lemma:lower-bound-multi}.
First of all, we demonstrate that there are many vertices in $\mathcal{U}$
that are isolated in the subgraph induced by $\mathcal{U}$. In fact,
let $\mathcal{U}_{0}$ be a random subset of $\mathcal{U}$ of size
$\frac{\left|\mathcal{U}\right|}{\log^{3}n}$. By Markov's inequality,
the number of samples with two endpoints lying in $\mathcal{U}_{0}$---denoted
by $N_{\mathcal{U}_{0}}$---is bounded above by
\begin{eqnarray*}
N_{\mathcal{U}_{0}} & \lesssim & \log n\cdot\mathbb{E}\left[\lambda\cdot\left|\mathcal{E}\left(\mathcal{U}_{0},\mathcal{U}_{0}\right)\right|\right]\lesssim\lambda\left(\frac{1}{\log^{6}n}\left|\mathcal{E}\left(\mathcal{U},\mathcal{U}\right)\right|\right)\log n\lesssim\lambda\left(\frac{1}{\log^{6}n}\left|\mathcal{U}\right|\tilde{d}\right)\log n\\
 & \overset{(\text{i})}{\lesssim} & \frac{\left|\mathcal{U}\right|}{\log^{4}n}=o\left(\left|\mathcal{U}_{0}\right|\right)
\end{eqnarray*}
with high probability, where $\mathcal{E}(\mathcal{U},\tilde{\mathcal{U}})$
denotes the set of edges linking $\mathcal{U}$ and $\tilde{\mathcal{U}}$,
and (i) follows from the assumption (\ref{eq:assumption-converse}).
As a consequence, one can find $\left(1-o\left(1\right)\right)\left|\mathcal{U}_{0}\right|$
vertices in $\mathcal{U}_{0}$ that are involved in absolutely no
sample falling within $\mathcal{E}\left(\mathcal{U}_{0},\mathcal{U}_{0}\right)$.
Let $\mathcal{U}_{1}$ be the set of these isolated vertices, which
obeys 
\begin{equation}
\left|\mathcal{U}_{1}\right|=\left(1-o\left(1\right)\right)\left|\mathcal{U}_{0}\right|=\frac{\left(1-o\left(1\right)\right)\left|\mathcal{U}\right|}{\log^{3}n}\geq\frac{\left|\mathcal{U}\right|}{2\log^{3}n}\label{eq:size-V3}
\end{equation}
for large $n$. We emphasize that the discussion so far only concerns
the subgraph induced by $\mathcal{U}$, which is independent of the
samples taken over $\mathcal{E}\left(\mathcal{U},\overline{\mathcal{U}}\right)$. 

Suppose the ground truth is $X_{i}=0$ ($1\leq i\leq n$). For each
vertex $v\in\mathcal{U}_{1}$, construct a representative singleton
hypothesis $\bm{X}^{v}$ such that
\[
X_{i}^{v}=\begin{cases}
1,\quad & \text{if }i=v,\\
0, & \text{else}.
\end{cases}
\]
Let $P_{0}$ (resp.~$P_{v}$) denote the output distribution given
$\bm{X}=\bm{0}$ (resp.~$\bm{X}=\bm{X}^{v}$). Assuming a uniform
prior over all candidates, it suffices to study the ML rule which
achieves the best error exponent. For each $v\in\mathcal{U}_{1}$,
since it is isolated in $\mathcal{U}_{1}$, all information useful
for differentiating $\bm{X}=\bm{X}^{v}$ and $\bm{X}=\bm{0}$ falls
within the positions $v\times\overline{\mathcal{U}_{0}}$, which in
total account for at most $\tilde{d}$ entries. The main point is
that for any $v,u\in\mathcal{U}_{1}$, the corresponding samples over
$v\times\overline{\mathcal{U}_{0}}$ and $u\times\overline{\mathcal{U}_{0}}$
are statistically independent, and hence the events $\left\{ \frac{P_{v}\left(\bm{Y}\right)}{P_{0}\left(\bm{Y}\right)}\geq1\right\} _{v\in\mathcal{U}_{1}}$
are independent conditional on $\mathcal{U}_{1}$.

We now consider the probability of error conditional on $\mathcal{U}_{1}$.
Observe that $P_{v}\left(\bm{Y}\right)$ and $P_{0}\left(\bm{Y}\right)$
only differ over those edges incident to $v$. Thus, Lemma \ref{lem:Chernoff-information}
suggests that
\[
P_{0}\left(\frac{P_{v}\left(\bm{Y}\right)}{P_{0}\left(\bm{Y}\right)}\geq1\right)\geq\exp\left\{ -\left(1+o\left(1\right)\right)\lambda\tilde{d}\left(1-e^{-D^{*}}\right)\right\} 
\]
for any $v\in\mathcal{U}_{1}$. The conditional independence of the
events $\left\{ \frac{P_{v}\left(\bm{Y}\right)}{P_{0}\left(\bm{Y}\right)}\geq1\right\} $
gives
\begin{eqnarray}
\mathbb{P}_{\mathrm{e}}\left(\psi_{\mathrm{ml}}\right) & \geq & P_{0}\left(\exists v\in\mathcal{U}_{1}:\text{ }\frac{P_{v}\left(\bm{Y}\right)}{P_{0}\left(\bm{Y}\right)}\geq1\right)=1-\prod_{v\in\mathcal{U}_{1}}\left\{ 1-P_{0}\left(\frac{P_{v}\left(\bm{Y}\right)}{P_{0}\left(\bm{Y}\right)}\geq1\right)\right\} \nonumber \\
 & \geq & 1-\left\{ 1-\exp\left[-\left(1+o\left(1\right)\right)\lambda\tilde{d}\left(1-e^{-D^{*}}\right)\right]\right\} ^{|\mathcal{U}_{1}|}\\
 & \geq & 1-\left\{ 1-\exp\left[-\left(1+o\left(1\right)\right)\lambda\tilde{d}\left(1-e^{-D^{*}}\right)\right]\right\} ^{\frac{\left|\mathcal{U}\right|}{2\log^{3}n}}\label{eq:inequality2-2}\\
 & \geq & 1-\exp\left\{ -\exp\left[-\left(1+o\left(1\right)\right)\lambda\tilde{d}\left(1-e^{-D^{*}}\right)\right]\frac{\left|\mathcal{U}\right|}{2\log^{3}n}\right\} ,\label{eq:ML-LB}
\end{eqnarray}
where (\ref{eq:inequality2-2}) comes from (\ref{eq:size-V3}), and
the last inequality follows since $1-x\leq\exp\left(-x\right)$. 

With (\ref{eq:ML-LB}) in place, we see that ML fails with probability
approaching one if
\[
\exp\left\{ -\left(1+o\left(1\right)\right)\lambda\tilde{d}\left(1-e^{-D^{*}}\right)\right\} \frac{\left|\mathcal{U}\right|}{\log^{3}n}\rightarrow\infty,
\]
which would hold under the assumption (\ref{eq:assumption-converse}).

(2) Now we turn to Lemma \ref{lemma:lower-bound-weighted}. Without
loss of generality, it is assumed that $w_{i,j}=\Theta\left(1\right)$
for all $(i,j)\in\mathcal{E}$. The preceding argument immediately
carries over to the sampling model with nonuniform weight, as long
as all vertex degrees are replaced with the corresponding weighted
degrees (cf.~(\ref{eq:weighted-degree})). 

(3) Finally, the preceding argument remains valid for proving Lemma
\ref{lemma:lower-bound-multi} with minor modification. Let $\mathcal{U}_{0}$
be a random subset of $\mathcal{U}$ of size $\frac{\left|\mathcal{U}\right|}{\log^{3}n}$,
denote by $\mathcal{E}_{\mathcal{U}_{0}}$ the collection of hyper-edges
with at least two endpoints in $\mathcal{U}_{0}$, and let $N_{\mathcal{U}_{0}}$
represent the number of samples that involve at least two nodes in
$\mathcal{U}_{0}$. When $L$ is a fixed constant, applying Markov's
inequality one gets
\begin{eqnarray*}
N_{\mathcal{U}_{0}} & \lesssim & \log n\cdot\mathbb{E}\left[\lambda\cdot\left|\mathcal{E}_{\mathcal{U}_{0}}\right|\right]\lesssim\lambda\left({L \choose 2}\left(\frac{1}{\log^{3}n}\right)^{2}\left|\mathcal{E}_{\mathcal{U}}\right|\right)\log n\lesssim\lambda\left(\frac{1}{\log^{6}n}\left|\mathcal{U}\right|\tilde{d}\right)\log n\\
 & \overset{(\text{i})}{\lesssim} & \frac{\left|\mathcal{U}\right|}{\log^{4}n}=o\left(\left|\mathcal{U}_{0}\right|\right)
\end{eqnarray*}
with high probability, where $\tilde{d}$ denotes the maximum vertex
degree (defined w.r.t.~the hyper-edges) in $\mathcal{U}$. That said,
there exist $\left(1-o\left(1\right)\right)\left|\mathcal{U}_{0}\right|$
vertices in $\mathcal{U}_{0}$ that are involved in absolutely no
sample falling within $\mathcal{E}_{\mathcal{U}_{0}}$, and if we
let $\mathcal{U}_{1}$ be the set of these isolated vertices, then
\begin{equation}
\left|\mathcal{U}_{1}\right|=\left(1-o\left(1\right)\right)\left|\mathcal{U}_{0}\right|=\frac{\left(1-o\left(1\right)\right)\left|\mathcal{U}\right|}{\log^{3}n}\geq\frac{\left|\mathcal{U}\right|}{2\log^{3}n}\label{eq:size-V3-1}
\end{equation}
for large $n$. Repeating the remaining argument in Part (1) finishes
the proof.

\subsection{Proof of Lemma \ref{lem:Spectral}\label{sec:Proof-of-Lemma-Spectral}}

To begin with, set $\tilde{\lambda}=1-\exp\left(-\lambda\right)$,
which satisfies $1\geq\tilde{\lambda}\gtrsim\log n/n$ under the assumption
of this lemma. Then the sample matrix $\bm{A}$ generated in Algorithm
\ref{alg:Algorithm-spectral} obeys
\begin{align}
\mathbb{E}\left[\bm{A}\right] & =\tilde{\lambda}\left(\bm{1}\bm{1}^{\top}-\bm{I}\right)\left\{ \mathbb{E}\left[Y_{i,j}^{(1)}=0\right]-\mathbb{E}\left[Y_{i,j}^{(1)}=1\right]\right\} \nonumber \\
 & =\tilde{\lambda}\left(1-2\theta\right)\bm{1}\bm{1}^{\top}-\tilde{\lambda}\left(1-2\theta\right)\bm{I},\label{eq:EA-expression}
\end{align}
where the first term of (\ref{eq:EA-expression}) is the dominant
component. Moreover, we claim for the time being that the fluctuation
$\tilde{\bm{A}}:=\bm{A}-\mathbb{E}[\bm{A}]$ obeys 
\begin{equation}
\|\tilde{\bm{A}}\|\lesssim\sqrt{\tilde{\lambda}n}\label{eq:A-tilde-fluctuation}
\end{equation}
with probability at least $1-O\left(n^{-10}\right)$. In view of the
Davis-Kahan sin-$\Theta$ Theorem \cite{davis1970rotation}, the leading
eigenvector $\bm{u}$ of $\bm{A}=\tilde{\lambda}\left(1-2\theta\right)\bm{1}\bm{1}^{\top}+\tilde{\bm{A}}-\tilde{\lambda}\left(1-2\theta\right)\bm{I}$
satisfies
\begin{align*}
 & \min\left\{ \left\Vert \bm{u}-\frac{1}{\sqrt{n}}\bm{1}\right\Vert ,\left\Vert -\bm{u}-\frac{1}{\sqrt{n}}\bm{1}\right\Vert \right\} \lesssim\frac{\|\tilde{\bm{A}}\|+\left\Vert \tilde{\lambda}\left(1-2\theta\right)\bm{I}\right\Vert }{\tilde{\lambda}n\left(1-2\theta\right)-\|\tilde{\bm{A}}\|-\left\Vert \tilde{\lambda}\left(1-2\theta\right)\bm{I}\right\Vert }\\
 & \quad\lesssim\frac{\sqrt{\tilde{\lambda}n}+\tilde{\lambda}}{\tilde{\lambda}n}\text{ }\text{ }\ll1,
\end{align*}
which is sufficient to guarantee (\ref{eq:Error-spectral}). In fact,
suppose without loss of generality that $\left\Vert \bm{u}-\frac{1}{\sqrt{n}}\bm{1}\right\Vert \leq\left\Vert -\bm{u}-\frac{1}{\sqrt{n}}\bm{1}\right\Vert $.
According to the rounding procedure, 
\[
X_{i}^{(0)}=X_{i}=1\quad\text{if }\left|\bm{u}_{i}-\frac{1}{\sqrt{n}}\right|<\frac{1}{2\sqrt{n}},
\]
leading to an upper bound on the Hamming error
\begin{align*}
\frac{\left\Vert \bm{X}^{(0)}-\bm{1}\right\Vert _{0}}{n} & \leq\text{ }\frac{1}{n}\sum_{i=1}^{n}\mathbb{I}\left\{ X_{i}^{(0)}\neq X_{i}\right\} \text{ }\leq\text{ }\frac{1}{n}\sum_{i=1}^{n}\mathbb{I}\left\{ \left|\bm{u}_{i}-\frac{1}{\sqrt{n}}\right|\geq\frac{1}{2\sqrt{n}}\right\} \\
 & \leq\text{ }\frac{1}{n}\left\{ \frac{\left\Vert \bm{u}-\frac{1}{\sqrt{n}}\bm{1}\right\Vert ^{2}}{\left(1/(2\sqrt{n})\right)^{2}}\right\} \lesssim\text{ }\frac{1}{\tilde{\lambda}n}+\frac{1}{n^{2}}=o\left(1\right),
\end{align*}
as claimed in this lemma. 

It remains to justify the claim (\ref{eq:A-tilde-fluctuation}), for
which we start by controlling the mean $\mathbb{E}[\|\tilde{\bm{A}}\|]$.
The standard symmetrization argument \cite[Page 133]{tao2012topics}
reveals that
\begin{equation}
\mathbb{E}\left[\|\tilde{\bm{A}}\|\right]\leq\sqrt{2\pi}\mathbb{E}\left[\|\tilde{\bm{A}}\circ\bm{G}\|\right],\label{eq:symmetrization}
\end{equation}
where $\bm{G}$ is a symmetric standard Gaussian matrix (i.e.~$\left\{ \bm{G}_{i,j}\mid i\geq j\right\} $
are i.i.d.~standard Gaussian variables), and $\tilde{\bm{A}}\circ\bm{G}:=[\tilde{\bm{A}}_{i,j}\tilde{\bm{G}}_{i,j}]_{1\leq i,j\leq n}$
represents the Hadamard product of $\tilde{\bm{A}}$ and $\bm{G}$.
To control $\mathbb{E}[\|\tilde{\bm{A}}\circ\bm{G}\|]$, it follows
from \cite[Theorem 1.1]{bandeira2014sharp} that
\begin{equation}
\mathbb{E}\left[\left.\|\tilde{\bm{A}}\circ\bm{G}\|\right|\tilde{\bm{A}}\right]\lesssim\max_{i}\left\{ \sqrt{\sum\nolimits _{j=1}^{n}\tilde{\bm{A}}_{i,j}^{2}}\right\} +\sqrt{\log n},\label{eq:norm-AG}
\end{equation}
depending on the size of $\max_{i}\left\{ \sqrt{\sum_{j=1}^{n}\tilde{\bm{A}}_{i,j}^{2}}\right\} $.
First of all, with probability exceeding $1-O\left(n^{-10}\right)$
one has
\[
\sqrt{\sum\nolimits _{j=1}^{n}\tilde{\bm{A}}_{i,j}^{2}}\lesssim\tilde{\lambda}n,\qquad1\leq i\leq n,
\]
which arises by taking Chernoff inequality \cite[Appendix A]{alon2015probabilistic}
together with the union bound, and recognizing that $\mathbb{E}[\sum_{j=1}^{n}\tilde{\bm{A}}_{i,j}^{2}]\leq\tilde{\lambda}n$
and $\tilde{\lambda}n\gtrsim\log n$. In this regime, substitution
into (\ref{eq:norm-AG}) gives
\begin{equation}
\mathbb{E}\left[\left.\|\tilde{\bm{A}}\circ\bm{G}\|\right|\tilde{\bm{A}}\right]\lesssim\sqrt{\tilde{\lambda}n}+\sqrt{\log n}.\label{eq:norm-AG-1}
\end{equation}
Furthermore, the trivial bound $\sqrt{\sum_{j=1}^{n}\tilde{\bm{A}}_{i,j}^{2}}\leq n$
taken together with (\ref{eq:norm-AG}) gives
\begin{equation}
\mathbb{E}\left[\left.\|\tilde{\bm{A}}\circ\bm{G}\|\right|\tilde{\bm{A}}\right]\lesssim\sqrt{n}+\sqrt{\log n}\label{eq:norm-AG-2}
\end{equation}
in the complement regime. Put together (\ref{eq:symmetrization}),
(\ref{eq:norm-AG-1}) and (\ref{eq:norm-AG-2}) to arrive at
\begin{align}
\mathbb{E}\left[\|\tilde{\bm{A}}\|\right] & \leq\mathbb{E}\left[\mathbb{E}\left[\left.\|\tilde{\bm{A}}\circ\bm{G}\|\right|\tilde{\bm{A}}\right]\right]\nonumber \\
 & \lesssim\text{ }\mathbb{P}\left\{ \max_{i}\sqrt{\sum_{j=1}^{n}\tilde{\bm{A}}_{i,j}^{2}}\lesssim\tilde{\lambda}n\right\} \left(\sqrt{\tilde{\lambda}n}+\sqrt{\log n}\right)+\left(1-\mathbb{P}\left\{ \max_{i}\sqrt{\sum_{j=1}^{n}\tilde{\bm{A}}_{i,j}^{2}}\lesssim\tilde{\lambda}n\right\} \right)\left(\sqrt{n}+\sqrt{\log n}\right)\nonumber \\
 & \lesssim\text{ }\sqrt{\tilde{\lambda}n}+\sqrt{\log n}+\frac{1}{n^{10}}\left(\sqrt{n}+\sqrt{\log n}\right)\nonumber \\
 & \asymp\sqrt{\tilde{\lambda}n},\label{eq:EA-bound}
\end{align}
where the last inequality holds as long as 
\begin{equation}
\tilde{\lambda}\gtrsim\frac{\log n}{n}.\label{eq:sampling-rate-LB}
\end{equation}

To finish up, we shall connect $\|\tilde{\bm{A}}\|$ with $\mathbb{E}[\|\tilde{\bm{A}}\|]$
by invoking the Talagrand concentration inequality. Note that the
spectral norm $\|\bm{M}\|$ is a 1-Lipschitz function in $\bm{M}$,
which allows to apply \cite[Theorem 2.1.13]{tao2012topics} to yield
\begin{align}
\mathbb{P}\left\{ \left|\|\tilde{\bm{A}}\|-\mathbb{E}[\|\tilde{\bm{A}}\|]\right|\geq c_{1}\sqrt{\tilde{\lambda}n}\right\}  & \leq C_{2}\exp\left(-c_{2}c_{1}^{2}\tilde{\lambda}n\right)\label{eq:deviation-A-tilde}
\end{align}
for some constants $c_{1},c_{2},C_{2}>0$. Combining (\ref{eq:EA-bound})-(\ref{eq:deviation-A-tilde})
and taking $c_{1}$ to be sufficiently large lead to 
\[
\|\tilde{\bm{A}}\|\lesssim\sqrt{\tilde{\lambda}n}
\]
with probability at least $1-O\left(n^{-10}\right)$, concluding the
proof.

\subsection{Proof of Lemma \ref{lem:main-component}\label{sec:Proof-of-Lemma-main-component}}

It follows from Lemma \ref{lemma:Poisson} that for any small constant
$\delta>0$
\[
\mathbb{P}\left\{ \left(\bm{V}_{\bm{X}}\right)_{v}\geq\frac{1}{2}|\mathcal{S}(v)|-\delta\lambda d_{v}\right\} \leq\exp\left\{ -\left(1-o\left(1\right)\right)\left(1-\xi\left(\delta\right)\right)\lambda d_{v}\left(1-e^{-D^{*}}\right)\right\} ,
\]
where $D^{*}$ represents the Chernoff information. Recalling that
$\lambda d_{v}\gtrsim\log n$ for all $1\leq v\leq n$ and applying
the union bound, we get
\begin{equation}
\mathbb{P}\left\{ \exists1\leq v\leq n:\text{ }\left(\bm{V}_{\bm{X}}\right)_{v}\geq\frac{1}{2}|\mathcal{S}(v)|-\delta\log n\right\} \leq\sum_{v=1}^{n}\exp\left\{ -\left(1-o\left(1\right)\right)\left(1-\tilde{\xi}\left(\delta\right)\right)\lambda d_{v}\left(1-e^{-D^{*}}\right)\right\} \label{eq:signalUB}
\end{equation}
for some function $\tilde{\xi}\left(\delta\right)$ that vanishes
as $\delta\rightarrow0$. We can now analyze different sampling models
on a case-by-case basis. 

\textbf{(1) Rings}. All vertices have the same degree $d_{\mathrm{avg}}$.
Since the sample complexity is $m=\frac{1}{2}\lambda nd_{\mathrm{avg}}$,
we arrive at 
\begin{align*}
(\ref{eq:signalUB}) & \leq n\exp\left\{ -\left(1-o\left(1\right)\right)d_{\mathrm{avg}}\cdot\left(1-\overline{\xi}\left(\delta\right)\right)\lambda\left(1-e^{-D^{*}}\right)\right\} \\
 & =n\exp\left\{ -\left(1-o\left(1\right)\right)\left(1-\overline{\xi}\left(\delta\right)\right)\frac{2m}{n}\left(1-e^{-D^{*}}\right)\right\} ,
\end{align*}
which tends to zero as long as 
\[
m>\frac{1+\delta}{1-\tilde{\xi}\left(\delta\right)}\cdot\frac{n\log n}{2\left(1-e^{-D^{*}}\right)}.
\]

\textbf{(2) Lines with $r=n^{\beta}$ for some constant $0<\beta<1$}.
The first and the last $r$ vertices have degrees at least $\left(1-o\left(1\right)\right)\frac{1}{2}d_{\mathrm{avg}}$,
while all remaining $n-2r$ vertices have degrees equal to $\left(1-o\left(1\right)\right)d_{\mathrm{avg}}$.
This gives 
\begin{align*}
(\ref{eq:signalUB}) & \leq2r\cdot\exp\left\{ -\left(1-o\left(1\right)\right)\left(1-\tilde{\xi}\left(\delta\right)\right)\lambda\cdot\frac{1}{2}d_{\mathrm{avg}}\left(1-e^{-D^{*}}\right)\right\} \\
 & \quad\quad+\left(n-2r\right)\exp\left\{ -\left(1-o\left(1\right)\right)\left(1-\tilde{\xi}\left(\delta\right)\right)\lambda d_{\mathrm{avg}}\left(1-e^{-D^{*}}\right)\right\} \\
 & \leq2\exp\left\{ \beta\log n-\left(1-o\left(1\right)\right)\left(1-\tilde{\xi}\left(\delta\right)\right)\frac{m}{n}\left(1-e^{-D^{*}}\right)\right\} \\
 & \quad\quad+\exp\left\{ \log n-\left(1-o\left(1\right)\right)\left(1-\tilde{\xi}\left(\delta\right)\right)\frac{2m}{n}\left(1-e^{-D^{*}}\right)\right\} ,
\end{align*}
which converges to zero as long as 
\[
\begin{cases}
m & >\left(1+o\left(1\right)\right)\frac{1+\delta}{1-\tilde{\xi}\left(\delta\right)}\cdot\frac{\beta n\log n}{1-e^{-D^{*}}};\\
m & >\left(1+o\left(1\right)\right)\frac{1+\delta}{1-\tilde{\xi}\left(\delta\right)}\cdot\frac{n\log n}{2\left(1-e^{-D^{*}}\right)}.
\end{cases}
\]

\textbf{(3) Lines with $r=\gamma n$ for some constant $0<\gamma\leq1$}.
Each vertex has degree exceeding $\left(1-o\left(1\right)\right)r$,
indicating that 
\begin{align*}
(\ref{eq:signalUB}) & \leq n\exp\left\{ -\left(1-o\left(1\right)\right)\left(1-\overline{\xi}\left(\delta\right)\right)\lambda r\left(1-e^{-D^{*}}\right)\right\} \\
 & \leq\exp\left\{ \log n-\left(1-o\left(1\right)\right)\left(1-\overline{\xi}\left(\delta\right)\right)\lambda r\left(1-e^{-D^{*}}\right)\right\} \\
 & =\exp\left\{ \log n-\left(1-o\left(1\right)\right)\left(1-\overline{\xi}\left(\delta\right)\right)\frac{m}{n\left(1-\frac{1}{2}\gamma\right)}\left(1-e^{-D^{*}}\right)\right\} 
\end{align*}
where the last line follows from (\ref{eq:davg-linear-line}). This
converges to zero as long as 
\[
m>\frac{1+\delta}{1-\tilde{\xi}\left(\delta\right)}\left(1-\frac{1}{2}\gamma\right)\frac{n\log n}{1-e^{-D^{*}}}.
\]

\textbf{(4) Grids with $r=n^{\beta}$ for some constant $0<\beta<1$}.
Note that $d_{\mathrm{avg}}\asymp r^{2}=n^{2\beta}$. There are at
least $n-\pi r^{2}$ vertices with degrees equal to $(1-o\left(1\right))d_{\mathrm{avg}}$,
while the remaining vertices have degree at least $(1-o\left(1\right))d_{\mathrm{avg}}/4$.
This gives 
\begin{align*}
(\ref{eq:signalUB}) & \leq\pi r^{2}\cdot\exp\left\{ -(1-o\left(1\right))\left(1-\tilde{\xi}\left(\delta\right)\right)\lambda\frac{d_{\mathrm{avg}}}{4}\left(1-e^{-D^{*}}\right)\right\} \\
 & \quad\quad+\left(n-\pi r^{2}\right)\exp\left\{ -(1-o\left(1\right))\left(1-\tilde{\xi}\left(\delta\right)\right)\lambda d_{\mathrm{avg}}\left(1-e^{-D^{*}}\right)\right\} \\
 & \leq4\exp\left\{ 2\beta\log n-\left(1-o\left(1\right)\right)\left(1-\tilde{\xi}\left(\delta\right)\right)\lambda\cdot\frac{d_{\mathrm{avg}}}{4}\left(1-e^{-D^{*}}\right)\right\} \\
 & \quad\quad+\exp\left\{ \log n-\left(1-o\left(1\right)\right)\left(1-\tilde{\xi}\left(\delta\right)\right)\lambda d_{\mathrm{avg}}\left(1-e^{-D^{*}}\right)\right\} ,
\end{align*}
which vanishes as long as 
\[
\begin{cases}
\lambda d_{\mathrm{avg}} & >\left(1+o\left(1\right)\right)\frac{1+\delta}{1-\tilde{\xi}\left(\delta\right)}\cdot\frac{8\beta\log n}{1-e^{-D^{*}}};\\
\lambda d_{\mathrm{avg}} & >\left(1+o\left(1\right)\right)\frac{1+\delta}{1-\tilde{\xi}\left(\delta\right)}\cdot\frac{\log n}{1-e^{-D^{*}}}.
\end{cases}
\]
This together with the fact that $m=\frac{1}{2}n\lambda d_{\mathrm{avg}}$
establishes the proof for this case.

Finally, for the cases of lines (with $r=n^{\beta}$ for some constant
$0<\beta<1$) / rings / grids with non-uniform sampling weight, it
suffices to replace $d_{\mathrm{avg}}$ with the average weighted
degree (see Section \ref{sub:Beyond-lines-Stage2}). The case of small-world
graphs follows exactly the same argument as in the case of rings with
nonuniform weights.

\bibliography{bibfileLocality}
\bibliographystyle{unsrt}

\end{document}